\documentclass[12pt]{article} %
\usepackage{amsmath,amssymb,amsthm,mathrsfs,bm
}
\numberwithin{equation}{section} %
\theoremstyle{plain}
   \newtheorem{thm}{\hspace{\parindent}{\sc Theorem}}[section] %
   \newtheorem{pro}[thm]{\hspace{\parindent}Proposition}
   \newtheorem{cor}[thm]{\hspace{\parindent}Corollary}
   \newtheorem{lem}[thm]{\hspace{\parindent}Lemma}
\theoremstyle{remark} %
   \newtheorem{rem}{\hspace{\parindent}Remark}[section] %
\newtheorem{exmp}{\hspace{\parindent}Example}[section]%
\newcommand{\adots}{a = 0,1,2,\dots}
\newcommand{\bC}{\mathbb{C}}
\newcommand{\bR}{\mathbb{R}}

\newcommand{\Csw}{\mathcal{C}_{\ds w}(t,s)}

\newcommand{\Cspace}{C^{\infty}}
\newcommand{\Czerospace}{C^{\infty}_0(\bR^d)}
\newcommand{\domain}{[0,T]\times \bR^d}
\newcommand{\dH}{\mathfrak{H}}
\newcommand{\ddp}{\mathfrak{p}}

\newcommand{\df}{\mathfrak{f}}
\newcommand{\dl}{\mathfrak{l}}
\newcommand{\dM}{\mathfrak{M}}
\newcommand{\ds}{\mathfrak{s}}
\newcommand{\dw}{\mathfrak{w}}
\newcommand{\dW}{\mathfrak{W}}
\newcommand{\Fw}{\mathcal{F}_w}
\newcommand{\h}{\mathfrak{h}}
\newcommand{\Htilde}{\widetilde{H}}
\newcommand{\jdots}{j = 0,1,\dots,\nu-1}
\newcommand{\kdelta}{K_{w\Delta }}
\newcommand{\limepsilon}{\lim_{\epsilon\rightarrow 0+0}}
\newcommand{\Ll}{(L^2)^l}
\newcommand{\Mlc}{M_{l}(\bC)}
\newcommand{\qdelta}{q_{\Delta}}
\newcommand{\qts}{q^{t,s}_{x,y}}
\newcommand{\rittaiim}{\text{\rm{Im\hspace{0.05cm}}}}
\newcommand{\rittaire}{\text{\rm{Re\hspace{0.05cm}}}}
\newcommand{\Rjw}{R_{JW}}
\newcommand{\Rtjw}{\widetilde{R}_{JW}}
\newcommand{\Rtlw}{\widetilde{R}_{LW}}
\newcommand{\Rtrw}{\widetilde{R}_{RW}}

\newcommand{\Sspace}{{\cal S}}
\newcommand{\Tphi}{\mathcal{T}_{\varphi}}
\newcommand{\Ujw}{U_{JW}}
\newcommand{\Ulw}{U_{LW}}
\newcommand{\Urw}{U_{RW}}
\newcommand{\Utjw}{\widetilde{U}_{JW}}
\newcommand{\Utlw}{\widetilde{U}_{LW}}
\newcommand{\Utrw}{\widetilde{U}_{RW}}
\newcommand{\Wsx}{W_{\ds}(x)}
\newcommand{\Wstx}{W_{\ds}(t,x)}
\newcommand{\ts}{t,s}
\newcommand{\zbar}{\overline{z}}
\def\dbar{{\mathchar'26\mkern-12mud}}
\begin{document}
\title{From 
 each  of  Feynman's and  von Neumann's postulates  to the restricted Feynman path integrals: a mathematical theory of temporally continuous quantum measurements}
\author{Wataru Ichinose
\thanks{This work was supported by JSPS KAKENHI Grant Number JP22K03384.
 }
} 
\date{}
\maketitle %
\begin{quote}
{\small Department of Mathematics, Shinshu University,
Matsumoto 390-8621, Japan. 
\quad E-mail: ichinose@math.shinshu-u.ac.jp}%
\end{quote}\par
\begin{abstract}
Feynman proposed a postulate or a method of quantization in his celebrated paper in 1948.  Applying Feynman's postulate to  temporally continuous quantum measurements  of the positions of  particles,  Mensky   proposed the restricted Feynman path integrals for continuous quantum measurements after phenomenological considerations. Our aim in the present paper is to give a rigorous proof  that  Mensky's restricted Feynman path integrals emerge out of the Feynman's postulate under a simple approximation. In addition,  it is proved that the  restricted Feynman path integrals emerge out of  von Neumann's postulate on instantaneous measurements as well as Feynman's postulate.  The quantum systems that we study include spin systems. These results are applied to  formulations of the multi-split experiments, the quantum Zeno and the Aharanov-Bohm effects.
\end{abstract}
%
%
%
%

\section{Introduction}%
 Let $T > 0$ be an arbitrary constant, $0 \leq t \leq T$ and 
$x = (x_1,\dotsc,x_d)\in \bR^d$. We consider a particle with mass $m > 0$ and charge $\mathfrak{q} \in \bR$ moving in $\bR^d$ with  electric field $E(t,x) = (E_1,
\dots, E_d) \in \bR^d$ and a  magnetic strength tensor $B(t,x) = (B_{jk}(t,x))_{1\leq j < k \leq d}
\in \bR^{d(d-1)/2}$. 
Let   $(V(t,x),A(t,x)) = (V,A_1,
\dots,A_d)  \in \bR^{d+1}$ be an electromagnetic potential, i.e.
\begin{align} \label{1.1}
     & E = -\frac{\partial A}{\partial t} - \frac{\partial V}{\partial 
x},\notag \\
          & B_{jk} =  \frac{\partial A_k}{\partial x_j}  -\frac{\partial 
A_j}{\partial x_k}
\quad (1 \leq j <  k \leq d),
\end{align}
where  $\partial V/\partial x = (\partial V/\partial x_1,\dots,\partial V/\partial x_d)$. 
Then the Lagrangian function and the classical action are given by
\begin{equation} \label{1.2}
      \mathscr{L}(t,x,\dot{x})=  \frac{m}{2}|\dot{x}|^2 + \mathfrak{q}\dot{x}\cdot A(t,x) - \mathfrak{q}V(t,x),
\   \dot{x}\in \bR^d
\end{equation}
and 
\begin{equation} \label{1.3}
   S(t,s;q) = \int_s^t {\mathscr{L}}(\theta,q(\theta),\dot{q}(\theta))d\theta,\quad \dot{q}(\theta) = \frac{dq(\theta)}{d\theta}
\end{equation}
for a path $q(\theta) \in \bR^d\ (s \leq \theta \leq t)$, respectively.
The corresponding Schr\"odinger equation is defined by
\begin{align} \label{1.4}
& i\hbar \frac{\partial u}{\partial t}(t)  = H(t)u(t)\notag\\ 
& := \left[ \frac{1}{2m}\sum_{j=1}^d
      \left(\frac{\hbar}{i}\frac{\partial}{\partial x_j} - \mathfrak{q}A_j(t,x)\right)^2 + \mathfrak{q}V(t,x)\right]u(t),
\end{align}
where $\hbar$ is the Planck constant.
Throughout this paper we always consider solutions to the Schr\"odinger equations in the sense of distribution. Let $L^2 = L^2(\bR^d)$ denote the space of all square integrable functions on
$\bR^d$ with the inner
product $(f,g) := \int f(x)\overline{g(x)}dx$ and the norm $\Vert f\Vert$, where $\overline{g(x)}$ denotes the complex conjugate of $g(x)$. 
\par
Let's perform a measurement of the position of the particle during the time interval $[0,T]$.
Let $\{a(t) \in \bR^d;0 \leq t \leq T\}$ be its result and $\delta> 0$ its resolution or error  of the measuring device. The measurement gives a change of the probability amplitude of the particle, called wave-function reduction (cf. p.36 in \S 10 of \cite{Dirac}, \S 2.1.1 of \cite{Mensky 1993}, \S III.3 and \S IV.3 of \cite{Neumann} and \S 1.4 of \cite{Sakurai}). Let $f \in L^2$ be a probability amplitude of the particle at an initial time $t=0$.  Then,  Feynman's postulates I and II on p. 371 of \cite{Feynman 1948} say that the probability amplitude in the continuous measurement is heuristically given by the `` sum " of $e^{i\hbar^{-1}S(t,0;q)}f(q(0))$ over the set $\Gamma(t,x;a,\delta)$ of all paths $q$ satisfying $q(t) = x$ and $|q(\theta) - a(\theta)| \leq \delta$ for all $\theta \in [0,t]$, i.e.
\begin{equation} \label{1.5}
     \int_{\Gamma(t,x;a,\delta)} e^{i\hbar^{-1}S(t,0;q)}f(q(0)){\cal D}q.
\end{equation}
\par
An alternative Feynman path integral description has been proposed by Mensky in \S 4.2 of \cite{Mensky 1993} and \S 5.1.3 of \cite{Mensky 2000}, written formally as the `` sum " of $e^{i\hbar^{-1}S(t,0;q)}\dw_a(q)f(q(0))$ with a weight $\dw_a(q)$  over the set $\Gamma(t,x)$  of all paths $q$ satisfying $q(t) = x$, i.e.
\begin{equation} \label{1.6}
     \int_{\Gamma(t,x)} e^{i\hbar^{-1}S(t,0;q)}\dw_a(q)f(q(0)){\cal D}q
\end{equation}
(cf. \S 10.5.4 of \cite{Albeverio et all}, \S 5.1 of \cite{Mazzucchi}).
At first, Mensky defined $\dw_a(q)$ by
\begin{equation*}
\exp \left\{-\frac{1}{t\delta^2}\int_0^t|q(\theta) - a(\theta)|^2d\theta \right\}
\end{equation*}
(cf. (4.9) of \cite{Mensky 1993}), but after phenomenological considerations  he finally defined it by
\begin{equation} \label{1.7}
   \dw_a(q) = \exp \left\{-\gamma \int_0^t|q(\theta) - a(\theta)|^2d\theta \right\}
\end{equation}
with a constant $\gamma > 0$  for the reason that \eqref{1.7} is realistic. Mensky called $\gamma$ the width of the  quantum corridor and \eqref{1.6} the restricted path integral. In the present paper we call \eqref{1.6} the restricted Feynman path integral (RFPI).
In \cite{Ichinose 2023} the author has given the mathematical definition of \eqref{1.6} with \eqref{1.7} in the $L^2$ space and proved that \eqref{1.6} is equal to the solution to
\begin{equation} \label{1.8}
i\hbar \frac{\partial u}{\partial t}(t)  = \bigl[H(t) -i\hbar W(t,x;a)\bigr]u(t)
\end{equation}
with $u(0) = f$, where $W(t,x;a) = \gamma|x-a(t)|^2$.
\par
	Mensky made an effort in \S 2.3.3 and \S 4.2 of \cite{Mensky 1993}, \S 3 of \cite{Mensky 1998} and \S 5.1.3 of \cite{Mensky 2000} to justify  that \eqref{1.6} with \eqref{1.7} expresses the probability amplitude in the continuous measurement. See also \cite{Calarco, Sverdlov} for other efforts for justifications. All their justifications are not rigorous and seem inadequate to be understood.	
\par
 Our aim in the present paper is to prove by means of rigorous arguments that RFPIs \eqref{1.6} with \eqref{1.7} and $\gamma = 1/(2\delta^2)$ emerge out of both of the Feynman postulate \eqref{1.5} and the continuous limit of a sequence of von Neumann's instantaneous reductions.
 \par
 Our fundamental idea for the proof of the above emergence of RFPIs  is simple.  We write \eqref{1.5} formally as
\begin{equation*} 
\int_{\Gamma(t,x)}e^{i\hbar^{-1}S(t,0;q)} \Biggl\{ \prod_{\theta \in [0,t]} \Omega\bigl((q(\theta) - a(\theta))/\delta\bigr) \Biggr\}
f(q(0)) {\cal D}q,
\end{equation*}
where 
	\begin{equation}\label{1.9}
\Omega(x) = \begin{cases} 1, & |x| \leq 1, \\
                           0, & |x| > 1
	\end{cases}
\end{equation}
and $\prod_{\theta \in [0,t]}(\cdot)$ denotes the formal product of an uncountable number of functions.
Replacing $\Omega(x)$ with 
\begin{equation} \label{1.10}
G(x) = e^{-|x|^2/2}
\end{equation}
in the above as in \S 3-2 of \cite{Feynman-Hibbs} and \S 2.3.3 of \cite{Mensky 1993}, we approximate the above expression by
\begin{equation*} 
\int_{\Gamma(t,x)}e^{i\hbar^{-1}S(t,0;q)} \Biggl\{ \prod_{\theta \in [0,t]} G\bigl((q(\theta) - a(\theta))/\delta\bigr)\Biggr\} f(q(0)){\cal D}q.
\end{equation*}
As a result, we have
\begin{equation} \label{1.11}
\int_{\Gamma(t,x)}e^{i\hbar^{-1}S(t,0;q)}
\Biggl\{\exp -\frac{1}{2\delta^2} \sum_{\theta \in [0,t]}|(q(\theta) - a(\theta))|^2
 \Biggr\}f(q(0)) {\cal D}q
\end{equation}
formally as an approximation of \eqref{1.5}. In the present paper we will rigorously   prove the following two results in a  general form: (i) The expression \eqref{1.11} is well defined in $L^2$ and gives the RFPI \eqref{1.6}. (ii)  The RFPI \eqref{1.6} emerges out of the continuous limit of a sequence of von Neumann's instantaneous reductions: $L^2 \ni g \to G\bigl((\cdot - a(t))/\delta \bigr)^{\Delta' t}g \in L^2$ with a separating time $\Delta' t$. It should be noted that we don't use the reduction: $L^2 \ni g \to G\bigl((\cdot - a(t))/\delta \bigr)g \in L^2$ commonly used by many authors (cf. \cite{Caves,Exner-Ichinose,Friedman,Jacobs,Mensky 1993} and etc.). Our  results will be stated just below in more detail.
\par
	In the present paper we will consider more general functions $W(t,x)$ than $|x-a(t)|^2/(2\delta)^2$ and write 
\begin{equation} \label{1.12}
   S_w(t,s;q) = S(t,s;q) + i\hbar\int_s^tW(\theta,q(\theta))d\theta.
\end{equation}
Let $\nu \geq 1$ be an arbitrary integer, $\Delta: 0 = \tau_0 < \tau_1 < \dots < \tau_{\nu-1} < \tau_{\nu}= t$ a subdivision of $[0,t]$ 
 and $|\Delta|:= \max\{\tau_{j+1}- \tau_j; j = 0,1,\dots,\nu-1\}$.  We take constants $\kappa_j \in [\tau_j,\tau_{j+1}]\ (j=0,1,\dots,\nu-1)$ arbitrarily.
 One of our aims in the present paper is to prove
\begin{align} \label{1.13}
&\lim_{|\Delta|\to 0}\int_{\Gamma(t,x)}e^{i\hbar^{-1}S(t,0;q)}
\Bigl\{\exp - \sum_{j=0}^{\nu-1}(\tau_{j+1}-\tau_{j})W(\kappa_{j},q(\kappa_{j}) )
 \Bigr\}f(q(0)) {\cal D}q \notag \\
 & = \int_{\Gamma(t,x)}e^{i\hbar^{-1}S(t,0;q)}
e^{-\int_{0}^{t}W(\theta,q(\theta))d\theta}f(q(0)) {\cal D}q \notag \\
& \equiv \int_{\Gamma(t,x)}e^{i\hbar^{-1}S_{w}(t,0;q)}f(q(0)) {\cal D}q
\end{align}
 in $L^{2}$ for $f \in L^2$, where both of the Feynman path integrals on the left-hand  
and the right-hand sides of \eqref{1.13} have been defined rigorously in \cite{Ichinose 2006} and \cite{Ichinose 2023} respectively.
The formula \eqref{1.13} shows that RFPIs emerge out of Feynman's postulate.
\par
	Let $U(t,s)g$ be the solution to \eqref{1.4} with $u(s) = g$. Then from Theorem 2.2 of \cite{Ichinose 2006} we can see that \eqref{1.13} is expressed as
\begin{align} \label{1.14}
& \lim_{|\Delta|\to 0}U(t,\kappa_{\nu-1})e^{-(t-\tau_{\nu-1})W(\kappa_{\nu-1})}
U(\kappa_{\nu-1},\kappa_{\nu-2})e^{-(\tau_{\nu-1}-\tau_{\nu-2})W(\kappa_{\nu-2})} 
\notag \\
& \quad \cdots e^{-(\tau_{2}-\tau_{1})W(\kappa_{1})}U(\kappa_{1},\kappa_{0})e^{-\tau_{1}W(\kappa_{0})}U(\kappa_0,0)f 
 \notag \\
& = \int_{\Gamma(t,x)}e^{i\hbar^{-1}S(t,0;q)}
e^{-\int_{0}^{t}W(\theta,q(\theta))d\theta}f(q(0)) {\cal D}q.
\end{align}
We note that each of $e^{-(\tau_{j+1}-\tau_{j})W(\kappa_{j}) }g = e^{-(\tau_{j+1}-\tau_{j})W(\kappa_{j},x) }g(x)$ for $g \in L^{2}$ means von Neumann's projection  at time $\kappa_j$ with a weight, because $e^{-(\tau_{j+1}-\tau_{j})W(\kappa_{j}) }g$ is expressed as the ``sum'' of projections $g(y)\delta_d(x-y)$ of $g$   at $y \in \bR^{d}$ with weights  $e^{-(\tau_{j+1}-\tau_{j})W(\kappa_{j},y) }$, i.e. $\displaystyle{\int} e^{-(\tau_{j+1}-\tau_{j})W(\kappa_{j},y) }g(y)\delta_d(x-y)dy $, where $\delta_d(x)$ is the Dirac delta function. Consequently 
 the formula \eqref{1.14} shows that RFPIs emerge out of von Neumann's postulate.
 \par
	Let $\sigma > 0$ be a constant.  Then, besides \eqref{1.13} and \eqref{1.14}, we will prove
\begin{align} \label{1.15}
&\lim_{|\Delta|\to 0}\int_{\Gamma(t,x)}e^{i\hbar^{-1}S(t,0;q)}
\Bigl\{\exp - \sum_{j=0}^{\nu-1}(\tau_{j+1}-\tau_{j})^{1+\sigma}W(\kappa_{j},q(\kappa_{j}) )
 \Bigr\}f(q(0)) {\cal D}q \notag \\
 & = \int_{\Gamma(t,x)}e^{i\hbar^{-1}S(t,0;q)}
f(q(0)) {\cal D}q 
\end{align}
and 
\begin{align} \label{1.16}
& \lim_{|\Delta|\to 0}U(t,\kappa_{\nu-1})e^{-(t-\tau_{\nu-1})^{1+\sigma}W(\kappa_{\nu-1})}
U(\kappa_{\nu-1},\kappa_{\nu-2})e^{-(\tau_{\nu-1}-\tau_{\nu-2})^{1+\sigma}W(\kappa_{\nu-2})} \notag \\
&\cdot \cdots e^{-(\tau_{2}-\tau_{1})^{1+\sigma}W(\kappa_{1})}U(\kappa_{1},\kappa_{0})e^{-\tau_{1}^{1+\sigma}W(\kappa_{0})}U(\kappa_0,0)f
 \notag \\
& = \int_{\Gamma(t,x)}e^{i\hbar^{-1}S(t,0;q)}f(q(0)) {\cal D}q.
\end{align}
Both of \eqref{1.15} and \eqref{1.16} mean that the measurements given by $e^{-(\tau_{j+1}-\tau_{j})^{1+\sigma}W(\kappa_{j})}\\
 (\jdots)$ don't influence the movement of the particle at all.  \par
 Ultimately we will extend the formulas \eqref{1.13} - \eqref{1.16} to the system of a particle with spin, where the 
positions of all components of the spin are continuously measured during the time interval.
\par
	  Mensky in \S 2.3.3 of \cite{Mensky 1993} and Caves in \S III.C of \cite{Caves} (cf. (4.16) in \cite{Caves}) attempted to study
\begin{align} \label{1.17}
& \lim_{\nu \to \infty}(\sqrt{\pi}\delta)^{-d\nu/2}U(t,\tau_{\nu-1})e^{-W(\tau_{\nu-1})}
U(\tau_{\nu-1},\tau_{\nu-2})e^{-W(\tau_{\nu-2})}\cdots \notag \\
& \qquad \quad \cdot U(\tau_{2},\tau_{1}) e^{-W(\tau_{1})}U(\tau_{1},0)f
\end{align}
in place of the left-hand side of  \eqref{1.14}, where $W(t,x) = (2\delta^{2})^{-1}|x-a(t)|^{2}$ as in \eqref{1.7} and $\tau_{j+1}-\tau_{j} = \Delta't := t/\nu\ (\jdots)$.
On the other hand, Jacobs in \cite{Jacobs} also studied \eqref{1.17} and  derived the stochastic Schr\"odigner equations 
similar to \eqref{1.8}. There, he has replaced each $(\sqrt{\pi}\delta)^{-d/2}\bigl\{\exp -W(\tau_{j},x)\bigr\}\ (\jdots)$ in \eqref{1.17} with $(2\pi \delta^{2}/\Delta't)^{-d/2}\bigl\{\exp -(\Delta't)W(\tau_{j},x)\bigr\}$, which gives the distribution of the independent measurement error $\xi_j$ at time $\tau_j$, in order that the  variance $\delta^2/(\nu\Delta't)$ of the error's average $\nu^{-1}\sum_{j=0}^{\nu-1}\xi_{j}$
will remain fixed as $\nu \to \infty$ (cf. \S 3.1.1, (3.36) in \cite{Jacobs} and lines 3-12 from the bottom on p.33 in \cite{Mensky 1993}). Under this replacement, \eqref{1.17} is much like the left-hand side of \eqref{1.14}.
\par
	 As applications of the formulas  \eqref{1.13} and \eqref{1.14}, in the present papaer   we will give mathematical formulations to the multi-split experiments, 
	including the Stern-Gerlach experiment (cf.  chaps. 1, 3 and 5 of \cite{F-L-S III}, chap. 1 of \cite{Feynman-Hibbs} and \S 1.1 of \cite{Sakurai}), the quantum Zeno effect (cf. \S 2.2 of \cite{Mensky 1993} and \S 12.5 of \cite{Peres}) and the Aharonov-Bohm effect (cf. \cite{Aharonov-Bohm}, \S 15-5 of \cite{F-L-S II}, \S 4.2 of \cite{Peres} and \S 2.7 of \cite{Sakurai}). Then  we get the expressions of their probability amplitudes
	in both forms of RFPI and the continuous limt of von Neumann's projections. 
	In each of the above problems the particle can't move out of the open set $\mathcal{O}$ specified in $\bR^d$. Let $\chi_{\mathcal{O}}$ be the indication function of $\mathcal{O}$, i.e.
	$\chi_{\mathcal{O}}(x) = 1$ if $x \in \mathcal{O}$ and otherwise $\chi_{\mathcal{O}}(x) = 0$. 
	Many authors study 
\begin{equation} \label{1.18}
 \lim_{\nu \to \infty}\chi_{\mathcal{O}}(\cdot)U(t,\tau_{\nu-1})\chi_{\mathcal{O}}(\cdot)
U(\tau_{\nu-1},\tau_{\nu-2})\chi_{\mathcal{O}}(\cdot)\cdots \chi_{\mathcal{O}}(\cdot)U(\tau_{1},0)\chi_{\mathcal{O}}(\cdot)f
\end{equation}
with $\tau_{j+1}-\tau_{j} = t/\nu\ (\jdots)$ for the above problems in the same form as \eqref{1.17} (cf. \cite{Exner-Ichinose, Friedman, Mensky 1993}), using that $\chi_{\mathcal{O}}(\cdot)g$ is expressed as the ``sum'' of  projections $g(y)\delta_d(x-y)$ at $y \in \bR^{d}$ over 
$\mathcal{O}$, i.e. $\displaystyle{\int_{\mathcal{O}}}g(y)\delta_d(x-y)dy$. It should be noted that we do not have a logical basis for why we use \eqref{1.18} for the description of continuous measurements and the above problems, but
 we can write \eqref{1.18} as 
\begin{align} \label{1.19}
 \lim_{\nu \to \infty}& \chi_{\mathcal{O}}(\cdot)^{t-\tau_{\nu-1}}U(t,\tau_{\nu-1})\chi_{\mathcal{O}}(\cdot)^{\tau_{\nu-1}-\tau_{\nu-2}}
U(\tau_{\nu-1},\tau_{\nu-2})\chi_{\mathcal{O}}(\cdot)^{\tau_{\nu-2}-\tau_{\nu-3}}\cdots \notag \\
& \qquad \cdot \chi_{\mathcal{O}}(\cdot)^{\tau_{2}-\tau_{1}}U(\tau_{2},\tau_1)\chi_{\mathcal{O}}(\cdot)^{\tau_1}f
\end{align}
in the same form as \eqref{1.14}. This means that the commonly used expression \eqref{1.18} or \eqref{1.19} follows from Feynman's postulate like \eqref{1.15}, though it is not strictly proved.
\par
	We will prove our formulas \eqref{1.13}-\eqref{1.16} as follows. Let $U_w(t,s)f$ be the solution to \eqref{1.8} with $u(s) = f$, where $W(t,x;a)$ is replaced with a general $W(t,x)$.  In   \cite{Ichinose 2023} we  have already proved 
\begin{equation} \label{1.20}
U_w(t,0)f = \int_{\Gamma(t,x)}e^{i\hbar^{-1}S_{w}(t,0;q)}f(q(0)) {\cal D}q.
\end{equation}
First, we will show by using the energy method  for partial differential equations that the left-hand sides of \eqref{1.14} and \eqref{1.16} are equal to $U_w(t,0)f$ and $U(t,0)f$, respectively.  Then we get \eqref{1.14} and \eqref{1.16} from \eqref{1.20}.  As their corollaries we can prove \eqref{1.13} and \eqref{1.15}, applying the results in \cite{Ichinose 2006} to the left-hand sides of \eqref{1.13} and \eqref{1.15}, respectively.
\par
	The plan of the present paper is as follows.  In \S 2 we will state the mathematical results on the Feynman path integrals
	a little more generally than in \cite{Ichinose 2006,Ichinose 2023}.
	Their complete proofs will be given in the appendix. In \S 3 we will state all main results as Theorems 3.1 - 3.5. In \S 4 the applications of our main results will be stated. In \S 5 we prove Theorem 3.1, which shows \eqref{1.14} in general.  In \S 6 Theorem 3.2 - 3.5 will be proved.
 \section{Feynman path integrals}
	In this section we introduce some theorems on the Feynman path integrals being a little more general than in \cite{Ichinose 2006,Ichinose 2023}
	 to state the main results in \S  3.  Their complete proofs will be given in the appendix B.
	\par
	 For a multi-index $\alpha = (\alpha_1,\dots,\alpha_d)$ and $x \in \bR^d$ we write $|x|^2 = \sum_{j=1}^dx_j^2,\\ <x> = \sqrt{1+|x|^2}, \partial_{x_j} = \partial/\partial x_j$,
	 $|\alpha| = 
\sum_{j=1}^d
\alpha_j$, $x^{\alpha} =  x_1^{\alpha_1}
\cdots  x_d^{\alpha_d}$ and $\partial_x^{\alpha} =  \partial_{x_1}^{\alpha_1}\cdots  \partial_{x_d}^{\alpha_d}$.  We denote the space of all rapidly decreasing functions on $\bR^d$ by $\Sspace = \Sspace(\bR^d)$.
 In the present paper we often use symbols $C, C_{\alpha}, C_{\alpha\beta}$, $C_a$  and $\delta_{\alpha}$ to write down constants, though these values are different in general. We introduce the weighted Sobolev spaces 
\begin{align} \label{2.1}
& B^a(\mathbb{R}^d)  := \{f \in  L^2(\mathbb{R}^d);
 \|f\|_a := \|f\| + \sum_{|\alpha| =  a}\bigl(\|x^{\alpha}f\| +\|\partial_x^{\alpha}f\|\bigr) < \infty\} 
 \end{align}
 for $a = 1,2,\dots.$ We denote the dual space of $B^a$ by $B^{-a}$ (cf. Lemma 2.4 in \cite{Ichinose 1995}) and the $L^2$ space by $B^0$. We write $M_l(\bC)$ for the space of all complex square matrices of degree $l$ and $\Vert A \Vert_{\bC^l}$ 
 for the operator norm of $A \in M_l(\bC)$ from $\bC^l $ to $\bC^l$, where $l \geq 1$ is an integer.  In the present paper we will use elementary results on vector-valued functions on $[0,T]$ without notice, i.e. they take values  in a Banach space.  See  \S 1 in Chapter 1 of \cite{Kato} or 
 Chapter 10 of \cite{Kolmogorov-Fomin} for elementary results.
 \par
	Hereafter we always consider a particle with $l$ spin components for the sake of simplicity. Then we will prove the formulas 
	corresponding to \eqref{1.13}-\eqref{1.16}.
	In the case of $l = 1$ we obtain the results 
	for a particle without spin.  Let $H_{\ds}(t,x) = (h_{ij}(t,x);i,j = 1,2,\dots,l)$ and $W_{\ds}(t,x) = (w_{ij}(t,x);i,j = 1,2,\dots,l)$ in  
	$\Mlc$ be Hermitian matrices denoting  a spin term and  a weight term respectively, where each of $h_{ij}(t,x)$ and $w_{ij}(t,x)$ denotes the $(i,j)$ component. In particular, $W_{\ds}$ is assumed to be bounded below in $\domain$ and corresponding to $\gamma |x - a(t)|^2$ in \eqref{1.7}. We define the effective Lagrangian function (cf. \S  4.2 of \cite{Mensky 1993}) by
\begin{equation} \label{2.2}
\mathscr{L}_{sw}(t,x,\dot{x}) = \mathscr{L}(t,x,\dot{x}) - \hbar H_{\ds}(t,x) + i\hbar W_{\ds}(t,x),
 \end{equation}
using \eqref{1.2}.  Then the corresponding  quantum equation becomes
\begin{equation}  \label{2.3}
   i\hbar\frac{\partial u}{\partial t}(t) = \bigl[H(t) + \hbar H_{\ds}(t,x)  - i\hbar W_{\ds}(t,x)\bigr]u(t) \equiv H_w(t)u(t)
\end{equation}
with a non-hermitian potential $i\hbar W_{\ds}(t,x)$, where $H(t)$ is defined by \eqref{1.4}.  For a path $q(\theta) \in \bR^d \ (s \leq \theta \leq t)$ we define $\mathcal{F}_w(\theta,s;q)= \left(\mathcal{F}_{ij}(\theta,s;q);i,j = 1,2,\dots,l\right) \\ (s \leq \theta \leq t)$ in $\Mlc$ by the solution $\mathcal{U}(\theta)$  to 
\begin{equation} \label{2.4}
\frac{d}{d\theta}\,\mathcal{U}(\theta) = -\bigl\{iH_{\ds}(\theta,q(\theta)) + W_{\ds}(\theta,q(\theta))\bigr\}\mathcal{U}(\theta)
 \end{equation}
with $\mathcal{U}(s) = I$ as in (2.13) of \cite{Ichinose 2023}, where $I$ is the identity matrix.  
	\par
	Now we will define the restricted Feynman path integrals or RFPIs as follows, which give a mathematical meaning to \eqref{1.6} or the right-hand side of \eqref{1.13} in the case of a particle with spin $0$.  We  take a $\chi \in \Cspace_0(\bR^d)$, i.e. an infinitely differentiable function on $\bR^d$ with compact support such that $\chi(0) = 1$, and  fix it through the present paper. For simplicity we suppose that $\chi$ is real-valued. 
	Let $x \in \bR^d$ be fixed.  Let $\nu \geq 1$ be an arbitrary integer and  $\Delta: 0 = \tau_0 < \tau_1 < \dots < \tau_{\nu-1} < \tau_{\nu}= t$   a subdivision of $[0,t]$. We take arbitrary points
$x^{(j)} \in \bR^d\ (j = 0,1,\dotsc,\nu-1)$ and  determine the piecewise free moving path or the piecewise straight line $\qdelta(\theta;x^{(0)},\dotsc,x^{(\nu-1)},x) \in \bR^d \ (0 \leq \theta \leq t)$ by joining  $x^{(j)}$ at $\tau_j\ (j = 0,1, \dotsc,\nu, x^{(\nu)} = x)$  in order. 
 Then we will determine the approximation  of the RFPI by 
 \begin{align} \label{2.5}
& K_{w\Delta}(t,0)f  =  \limepsilon  \prod_{j=0}^{\nu-1}\sqrt{\frac{m}{2\pi i(\tau_{j+1} - 
\tau_{j})}}^{\ d}
        \int_{\bR^d}\cdots\int_{\bR^d} e^{i\hbar^{-1} S(t,0;\qdelta)}\notag \\
       & \quad \times  \mathcal{F}_w(t,0;\qdelta)f(\qdelta(0)) \prod_{j=1}^{\nu-1}\chi(\epsilon x^{(j)})
          dx^{(0)}dx^{(1)}\cdots dx^{(\nu-1)}
 \end{align}
for $f \in \Cspace_0(\bR^d)^l = \bC^l\otimes \Cspace_0(\bR^d)$ as in (2.15) of \cite{Ichinose 2023}. The right-hand side of \eqref{2.5} is called an oscillatory integral and will be denoted by
 \begin{align*} 
   \prod_{j=0}^{\nu-1}\sqrt{\frac{m}{2\pi i(\tau_{j+1} - 
\tau_{j})}}^{\ d}
      & \text{Os} -  \int_{\bR^d}\cdots\int_{\bR^d} e^{i\hbar^{-1}S(t,0;\qdelta)} \Fw(t,0;\qdelta)\\
      & \times f(\qdelta(0)) 
          dx^{(0)}dx^{(1)}\cdots dx^{(\nu-1)}
 \end{align*}
 (cf. p. 45 of \cite{Kumano-go}).
The $B^a$-norm  of $f = {}^t(f_1,\dots,f_l) \in B^a(\bR^d)^l$ is defined by $\Vert f \Vert_a := \sqrt{\sum_{j=1}^l \Vert f_j\Vert^2_a}\ (a = 0,1,2,\dots)$.
\par
%
   Throughout the present paper we assume that all functions $\partial_x^{\alpha}E_j(t,x), \\ \partial_x^{\alpha}B_{jk}(t,x), \partial_x^{\alpha}h_{ij}(t,x)$ and $\partial_x^{\alpha}w_{ij}(t,x)$ are continuous in $\domain$ for all $\alpha$.  Then, all $\partial_x^{\alpha}\partial_tB_{jk}(t,x)$ are also continuous in $\domain$ for all $\alpha$ because of Faraday's law $\partial_tB_{jk} =  - \partial E_k/\partial x_j + \partial E_j/\partial x_k$, which follows from \eqref{1.1}.
\par
{\bf Assumption 2.A.}  We assume 
\begin{equation} \label{2.6}
|\partial_x^{\alpha}E_j(t,x)| \leq C_{\alpha},\  |\alpha| \geq 1,\quad j = 1,2,\dots,d,
\end{equation}
\begin{equation} \label{2.7}
|\partial_x^{\alpha}B_{jk}(t,x)| \leq C_{\alpha}<x>^{-(1+ \delta_{\alpha})},\  |\alpha| \geq 1, \quad 1 \leq j < k \leq d
\end{equation}
in $\domain$ with constants $C_{\alpha} \geq 0$ and $\delta_{\alpha} > 0$.
\par
{\bf Assumption 2.B.}  We assume \eqref{2.6} and 
\begin{equation} \label{2.8}
|\partial_x^{\alpha}\partial_tB_{jk}(t,x)| \leq C_{\alpha}<x>^{-(1+ \delta_{\alpha})},\  |\alpha| \geq 1, \quad 1 \leq j < k \leq d
\end{equation}
in $\domain$ with constants $C_{\alpha} \geq 0$ and $\delta_{\alpha} > 0$.
\par
{\bf Assumption 2.C.} We assume that $\partial_x^{\alpha}A_j(t,x)\ (j = 1,2,\dots,d)$ and $\partial_x^{\alpha}V(t,x)$ are continuous in $\domain$ for all $\alpha$ and satisfy
\begin{equation} \label{2.9}
|\partial_x^{\alpha}A_j(t,x)| \leq C_{\alpha},\ |\alpha| \geq 1,\quad j = 1,2,\dots,d,
\end{equation}
\begin{equation} \label{2.10}
|\partial_x^{\alpha}V(t,x)| \leq C_{\alpha}<x>,\quad |\alpha| \geq 1
\end{equation}
in $\domain$ with constants $C_{\alpha} \geq 0$.
\par
{\bf Assumption 2.D.}  There exist a continuous function $w(t,x) \geq 0$ in $\domain$ and a constant $C_W \geq 0$ such that
\begin{equation} \label{2.11}
W_{\ds}(t,x) \geq (w(t,x) - C_W)I,
\end{equation}
\begin{equation} \label{2.12}
\Vert\partial_x^{\alpha}W_{\ds}(t,x)\Vert_{\bC^l} \leq C_{\alpha}\bigl\{1  + w(t,x)\bigr\},\quad |\alpha| \geq 1,
\end{equation}
\begin{equation} \label{2.13}
\Vert\partial_x^{\alpha}W_{\ds}(t,x)\Vert_{\bC^l} \leq C_{\alpha}<x>,\quad |\alpha| \geq 1
\end{equation}
in $\domain$ with constants $C_{\alpha} \geq 0$.
\par
\begin{rem} \label{Remark 2.1}
In Theorems 2.1-2.4 of \cite{Ichinose 2023} we assumed
\begin{equation*} 
\Vert\partial_x^{\alpha}W_{\ds}(t,x)\Vert_{\bC^l} ^{p_{\alpha}}\leq C_{\alpha}\bigl\{1  + w(t,x)\bigr\},\quad |\alpha| \geq 1
\end{equation*}
in $\domain$ with constants $C_{\alpha} \geq 0$ and $p_{\alpha} \geq 1$  instead of \eqref{2.12}.  These assumptions are clearly  equivalent to \eqref{2.12} 
because $|f(x)| \leq |f(x)| ^{p_{\alpha}}$ if $|f(x)| \geq 1$ and $p_{\alpha} \geq  1$.
\end{rem}
\par
	We have proved the following in Theorem of \cite{Ichinose 1995} and (i) of Theorem 2.1 of \cite{Ichinose 2020}.
\begin{thm}  Suppose Assumption 2.C, \eqref{2.11} with $w(t,x) = 0$, \eqref{2.13} and 
 \begin{equation} \label{2.14}
 \Vert \partial_{x}^{\alpha}H_{\ds}(t,x)\Vert_{\bC^{l}} \leq  C_{\alpha}<x>, \quad |\alpha| \geq 1
 \end{equation}
 in $\domain$. Let $s \in [0,T)$.  Then for any $f \in (B^{a})^{l}\ (a = 0, \pm1,\pm2,\dots)$
 there exists the unique solution $u(t) = U_{w}(t,s)f$ in $C^0_t([s,T];(B^{a})^{l}) \cap C^1_t([s,T];(B^{a-2})^{l})$  with $u(s) = f$ to the equation \eqref{2.3},
 where  $C^j_t([s,T];(B^{a})^{l})$ denotes the space of all $(B^{a})^{l}$-valued, 
 $j$-times continuously differentiable functions in $t \in [s,T]$. This solution $U_{w}(t,s)f$ satisfies
 \begin{equation} \label{2.15}
 \Vert U_{w}(t,s)f\Vert_a\leq  C_a\Vert f\Vert_a,  \ 0 \leq s \leq t \leq T
 \end{equation}
 with a constant $C_{a} \geq 0$.  In particular, when $W_{\ds}(t,x) = 0$ in $\domain$, we have
 \begin{equation} \label{2.16}
 \Vert U(t,s)f\Vert = \Vert f\Vert,  \ 0 \leq s \leq t \leq T,
 \end{equation}
 where we wrote $U(t,s)f$ for $U_{w}(t,s)f$ with $W_{\ds}(t,x) = 0$ as in the introduction.
 \end{thm}
\par
	Theorems 2.2 - 2.5 mentioned below are a little more  general than the theorems in \cite{Ichinose 2006, Ichinose 2023}.  We will give
	their complete  proofs in the appendix B.
%
\begin{thm} \label{thm 2.2}
Suppose that Assumptions 2.A, 2.D and  
 \begin{equation} \label{2.17}
 \Vert \partial_{x}^{\alpha}H_{\ds}(t,x)\Vert_{\bC^{l}} \leq  C_{\alpha}, \quad |\alpha| \geq 1
 \end{equation}
 in $\domain$ are satisfied.  Then, there exists a constant $\rho^* > 0$ such that  the following statements hold for arbitrary potentials $(V,A)$ with continuous 
$V, \partial V/\partial x_j, \partial A_j/\partial t, \partial A_j/\partial x_k\ (j,k = 1,2,\dots,d)$ %
in $\domain$, all $\Delta$ satisfying $|\Delta| \leq \rho^*$ and all $t \in [0,T]$:
\\
(1) $K_{w\Delta}(t,0)f$ defined on $f \in \Czerospace^l$ by \eqref{2.5} is
determined independently of the choice of $\chi$ and $K_{w\Delta}(t,0)$ can be uniquely extended to a bounded operator on $(L^{2})^{l}$.
\\
(2) For all $f \in (L^{2})^{l}$, as $|\Delta| \to 0$, $\kdelta(t,0)f$ converges in $(L^{2})^{l}$ uniformly in $t \in [0,T]$ to 
an element $K_{w}(t,0)f \in (L^{2})^{l}$, which we call the  RFPI of $f$ and  express as 
\begin{equation*}
\int_{\Gamma(t,x)} e^{i\hbar^{-1}S(t,0;q)}\mathcal{F}_{w}(t,0;q)f(q(0)) {\cal D}q.
\end{equation*}
\\
(3) For all $f \in (L^{2})^{l}$, $K_{w}(t,0)f$  belongs to   $C_{t}^{0}([0,T];(L^{2})^{l})$.  In addition, 
$K_{w}(t,0)f$ is the unique solution  in $C^{0}_{t}([0,T];(L^{2})^{l})$ to \eqref{2.3} with $u(0) = f$. Consequently we have
\begin{equation*}
U_w(t,s)f = \int_{\Gamma(t,x)} e^{i\hbar^{-1}S(t,0;q)}\mathcal{F}_{w}(t,0;q)f(q(0)) {\cal D}q.
\end{equation*}
\\
(4) Let $\psi(t,x) \in C^1(\domain)$ be a real-valued function such that $\partial_{x_j}\partial_{x_k}\psi(t,x)$ and $\partial_{t}\partial_{x_j}\psi(t,x)$ $(j,k = 1,2,\dotsc,d)$ are continuous in $[0,T] \times \bR^d$ and consider the gauge transformation 
\begin{equation}  \label{2.18}
    V' = V -\frac{\partial\psi}{\partial  t}, \quad  A'_j = A_j + \frac{\partial\psi}{\partial  x_j}\quad (j= 1,2,\dots,d).
\end{equation}
We write \eqref{2.5} for this $(V',A')$  as $K'_{w\Delta}(t,0)f$.  Then we have the formula 
\begin{equation*}  
     K'_{w\Delta}(t,0)f  = e^{i\hbar^{-1}\psi (t,\cdot)}K_{w\Delta}(t,0)\left(e^{-i\hbar^{-1}\psi (0,\cdot)}f\right)
\end{equation*}
for  all $f \in (L^2)^{l}$  and  the analogous relation between the limits $K'_{w}(t,0)f$ and $K_{w}(t,0)f$.
\end{thm}
\begin{thm} \label{thm 2.3}
Suppose that either Assumption 2.A or 2.B is satisfied.  In addition, we suppose Assumptions 2.C, 2.D and \eqref{2.17}. Then there exists another constant $\rho^* > 0$  such that the same statements (1) - (4) as in Theorem 2.2 hold for all $\Delta$ satisfying $|\Delta| \leq \rho^*$ and all $t \in [0,T]$.  In addition, for all $f \in B^a(\bR^d)^l\ (a = 1,2,\dots)$ $K_{w\Delta}(t,0)f$  belongs to   $(B^a)^l$ and  as $|\Delta| \to 0$, $\kdelta(t,0)f$ converges in $(B^a)^l$ uniformly in $t \in [0,T]$ to $K_w(t,0)f$, which belongs to  $C^0_t([0,T];(B^a)^l)$.  
\end{thm} 
	For an arbitrary interger $N \geq 1$ we take $0 \leq t_1 < t_2 <\dots <t_{N-1} < t_N \leq t$ and matrix-valued functions $Z_j(x) \in \Mlc\ (j= 1,2,\dots,N)$ such that 
 \begin{equation} \label{2.19}
 \Vert \partial_{x}^{\alpha}Z_j(x)\Vert_{\bC^{l}} \leq  C_{\alpha}<x>^{\dM_j}
 \end{equation}
 for all $\alpha$ with integers $\dM_j \geq 0$.  We set $\dM = \sum_{j=1}^N \dM_j$. Let $\qdelta(\theta;x^{(0)},\dotsc,
 x^{(\nu-1)},\\x) \in \bR^d \ (0 \leq \theta \leq t)$ be the piecewise free moving path defined in the early part of this section.
  Noting that $ Z_j(q(t_j))
  \in \Mlc \ (j = 1,2,\dots,N)$ are determined
 for each path $q: [0,t] \to \bR^d$,  
  we define 
 \begin{align} \label{2.20}
& \left<\prod_{j=1}^NZ_j(\qdelta(t_j))\right>_wf  :=  \limepsilon  \prod_{j=0}^{\nu-1}\sqrt{\frac{m}{2\pi i(\tau_{j+1} - 
\tau_{j})}}^{\ d}
        \int_{\bR^d}\cdots\int_{\bR^d} e^{i\hbar^{-1}S(t,0;\qdelta)}\notag \\
       & \quad \times  \mathcal{F}_w(t,t_N;\qdelta)Z_N(\qdelta(t_N))\mathcal{F}_w(t_N,t_{N-1};\qdelta)\cdots Z_1(\qdelta(t_{1}))\notag \\
       &
     \qquad \times \mathcal{F}_w(t_1,0;\qdelta)f(\qdelta(0)) \prod_{j=1}^{\nu-1}\chi(\epsilon x^{(j)})
          dx^{(0)}dx^{(1)}\cdots dx^{(\nu-1)}
 \end{align}
for $f \in \Cspace_0(\bR^d)^l$ in the same way as we did (2.5).
%
\begin{thm} \label{thm 2.4}
Let $\dM_j = 0\ (j= 1,2,\dots,N)$ in \eqref{2.19}  and consider \eqref{2.20}.
We suppose the same assumptions as in Theorem 2.2 and let $\rho^*>0$ be the constant dertermined in Theorem 2.2.
Then we have the same statements (1), (2) and (4) for $\left<\prod_{j=1}^NZ_j(\qdelta(t_j))\right>_wf $ as in Theorem 2.2 in place of 
$\kdelta(t,0)f$.  In addition, we have the following statement:
\\
(3) For all $f \in (L^{2})^{l}$, the limit $\left<\prod_{j=1}^NZ_j(q(t_j))\right>_wf $ of $\left<\prod_{j=1}^NZ_j(\qdelta(t_j))\right>_wf$
as $|\Delta|\to 0$ belongs to   $C_{t}^{0}([t_N,T];(L^{2})^{l})$  and equals  
 \begin{equation} \label{2.21}
 U_w(t,t_N)Z_N(\cdot)U_w(t_N,t_{N-1})\cdots Z_1(\cdot)U_w(t_1,0)f.
  \end{equation}
\end{thm}
%
\begin{thm} \label{thm 2.5}
We suppose the same assumptions as in Theorem 2.3 and let $\rho^*>0$ be the constant dertermined in Theorem 2.3. We take $Z_j(x)\ (j=1,2,\dots,N)$ satisfying \eqref{2.19} and consider \eqref{2.20}.
Then we have the following statements for $|\Delta| \leq \rho^*$ and $a = 0,1,2,\dots$:
\\
(1) The function $\left<\prod_{j=1}^NZ_j(\qdelta(t_j))\right>_wf$ defined on $f \in \Czerospace^l$ is
determined independently of the choice of $\chi$ and  can be uniquely extended to a bounded operator from $(B^{a+\dM})^{l}$ to $(B^a)^l$.
\\
(2) For all $f \in (B^{a+\dM})^{l}$, as $|\Delta| \to 0$, $\left<\prod_{j=1}^NZ_j(\qdelta(t_j))\right>_wf$ converges in $(B^{a})^{l}$ to an element $\left<\prod_{j=1}^NZ_j(q(t_j))\right>_wf  \in (B^{a})^{l}$ uniformly in $t \in [t_N,T]$.
\\
(3) For all $f \in (B^{a+\dM})^{l}$, $\left<\prod_{j=1}^NZ_j(q(t_j))\right>_wf$  belongs to   $C_{t}^{0}([t_N,T];(B^a)^{l})$ and  
equals \eqref{2.21}.
\end{thm}
\begin{rem} \label{rem 2.2}  We consider a particle with $l$ spin components and measure  the position of each spin component during an interval $[0,T]$ of time.
Let $a^{(j)}(t)\ (j = 1,2,\dots,l)$ be the result for the $j$-th spin component and $\delta > 0$ the resolution or the error of the measuring device.  Let's  take the diagonal matrix in $M_l(\bC)$ defined by
\begin{equation} \label{2.22}
W_{\ds}(t,x) =
 \begin{pmatrix} 
 w_{11}(t,x)                                        \\
            & w_{22}(t,x)       &        & \text{\huge{0}}                           \\
    &  & \ddots     \\
 & \text{\huge{0}}    & & \ddots \\
                        & &    &          &w_{ll}(t,x)     
  \end{pmatrix}
 \end{equation}
 with $w_{jj}(t,x) = (2\delta^2)^{-1}|x - a^{(j)}(t)|^2$ in Theorems 2.2 and 2.3.
 We suppose that $a^{(j)}(t)\ (j= 1,2,\dots,l)$ are continuous on $[0,T]$.  Then $W_{\ds}(t,x)$ satisfies 
 Assumption 2.D as $w(t,x) = (4\delta^2)^{-1}|x|^2$. In fact we can easily prove \eqref{2.11} from
 \begin{equation*}
 \frac{1}{2\delta^2}|x - a^{(j)}(t)|^2 -  \frac{1}{4\delta^2}|x|^2 \geq \frac{|x|}{4\delta^2}(|x|-4A)
 = \frac{|x|^{2}}{4\delta^2}\left(1-\frac{4A}{|x|}\right),
 \end{equation*}
 where $A = \max_{1\leq j \leq l}\max_{0\leq t \leq T}|a^{(j)}(t)|$.
   This matrix-valued function $W_{\ds}(t,x)$
  corresponds to $W(t,x;a)$ in \eqref{1.8}.
\end{rem}
\begin{rem} \label{rem 2.3}
In the case of  $l = 1$ Theorems 2.2 and 2.3 in the present paper give the same results as in \cite{Ichinose 2023}.  On the other hand, when $l \geq 2$, they give more general results than Theorem 2.3 in \cite{Ichinose 2007} and Theorem 2.3 of \cite{Ichinose 2023}, where we have assumed
\begin{equation} \label{2.23}
W_{\ds} = 0\quad \text{or}\quad \Vert\partial_x^{\alpha}W_{\ds}(t,x)\Vert_{\bC^d} \leq C_{\alpha},
\end{equation}
\begin{equation} \label{2.24}
\Vert\partial_x^{\alpha}H_{\ds}(t,x)\Vert_{\bC^d} \leq C_{\alpha}
\end{equation}
in $\domain$ for all $\alpha$ in place of \eqref{2.13} and \eqref{2.17} respectively.
\end{rem}
\begin{rem} \label{rem 2.4}
When $l = 1$ and $W_{\ds}= 0$ , Theorems 2.4 and 2.5 in the present paper give  the same results as in Theorem 2.2 of \cite{Ichinose 2006}.   %
\end{rem}
\section{Main theorems and examples}
We will prove the formula \eqref{1.14} in a general form as follows.
%
\begin{thm} \label{thm 3.1}
Besides the assumptions of Theorem 2.1 we suppose that there exists a non-decreasing function $\sigma(\varrho) \geq 0\ (\varrho \in [0,T])$
satisfying 
\begin{align} \label{3.1}
& \sup_{x \in \bR^d} \Vert<x>^{-2}\partial_x^{\alpha}\bigl\{W_{\ds}(t,x)- W_{\ds}(s,x)\bigr\}\Vert_{\bC^l} \leq \sigma(t - s),\quad t,s \in [0,T],
\notag \\
& 0\leq \varrho \leq \sigma(\varrho) ,\quad  \lim_{\varrho\to 0+0}\sigma(\varrho) = 0.
\end{align}
 Let $U_w(t,s)f$ and $U(t,s)f$ be the solutions defined in Theorem 2.1 for  $f \in (B^a)^l\ (a = 0,1,2,\dots)$. 
For a subdivision $\Delta = \{\tau_j\}_{j=1}^{\nu-1}$ we take $\kappa_j$ and $\kappa'_j\ (\jdots)$ in $[\tau_j,\tau_{j+1}]$ arbitrarily.
Then we have
\begin{align} \label{3.2}
& \lim_{|\Delta|\to 0}U(t,\kappa_{\nu-1})e^{-(t-\tau_{\nu-1})W_{\ds}(\kappa'_{\nu-1})}
U(\kappa_{\nu-1},\kappa_{\nu-2})e^{-(\tau_{\nu-1}-\tau_{\nu-2})W_{\ds}(\kappa'_{\nu-2})} U(\kappa_{\nu-2},\kappa_{\nu-3}) \notag \\
&\quad \cdot \cdots e^{-(\tau_{2}-\tau_{1})W_{\ds}(\kappa'_{1})}U(\kappa_{1},\kappa_{0})e^{-\tau_{1}W_{\ds}(\kappa'_{0})}U(\kappa_0,0)f
= U_w(t,0)f
\end{align}
in $(B^a)^l$ uniformly in $t \in [0,T]$.
\end{thm}
%
	When $W_{\ds}(t,x) = 0$, we express \eqref{2.20}   and its limit as $|\Delta| \to 0$  as 
	\\ $\left<\prod_{j=1}^N Z_j(\qdelta(t_j))\right>f$ and $\left<\prod_{j=1}^NZ_j(q(t_j))\right>f$ respectively, removing $w$. In the following theorem we take  $\kappa_j$ and 
	$\exp\bigl\{-(\tau_{j+1}-\tau_j)W_{\ds}(\kappa_j,x)\bigl\}\ (j=0,\dots,\nu-1)$  as $t_j$ and $Z_j(x)\ (j=1,\dots,N)$ in \eqref{2.20}, respectively.
\begin{thm} \label{thm 3.2}
Besides the assumptions of Theorem 2.2 we suppose \eqref{3.1}. Let $(V,A)$ be a potential satisfying the properties stated in Theorem 2.2.
 For a subdivision $\Delta' = \{\tau_j\}_{j=1}^{\nu-1}$ we take $\kappa_j\ (\jdots)$  in $[\tau_j,\tau_{j+1}]$ arbitrarily. 
 Let $f \in (L^2)^l$.
 Then, there exists $\left<\prod_{j=0}^{\nu-1}\exp\Bigl\{-(\tau_{j+1}-\tau_j)W_{\ds}(\kappa_j,q(\kappa_j))\Bigr\}\right>f \in C^0_t([\kappa_{\nu-1},T];(L^2)^l)$   and we have
\begin{align} \label{3.3}
& \lim_{|\Delta'|\to 0}\left<\prod_{j=0}^{\nu-1}\exp\Bigl\{-(\tau_{j+1}-\tau_j)W_{\ds}(\kappa_j,q(\kappa_j))\Bigr\}\right>f \notag \\
& = \int_{\Gamma(t,x)} e^{i\hbar^{-1}S(t,0;q)}\mathcal{F}_{w}(t,0;q)f(q(0)) {\cal D}q
\end{align}
in $(L^2)^l$ uniformly in $t \in [0,T]$.  
\end{thm}
\begin{thm} \label{thm 3.3}
Besides the assumptions of Theorem 2.3 we suppose \eqref{3.1}. Let $\Delta'$ and $\kappa_j\ (\jdots)$ be as in Theorem 3.2. Let $f \in (B^a)^l\ (a = 0,1,2,\dots)$. Then, there exists $\Bigl<\prod_{j=0}^{\nu-1}\exp\Bigl\{-(\tau_{j+1}-\tau_j)
W_{\ds}(\kappa_j,q(\kappa_j))\Bigr\}\Bigr>f \in C^0_t([\kappa_{\nu-1},T];(B^a)^l)$   and we have \eqref{3.3} 
 in $(B^a)^l$ uniformly in $t \in [0,T]$.
\end{thm}
	We obtain the formulas \eqref{1.15} and \eqref{1.16}  in a general form as follows.
%
\begin{thm} \label{thm 3.4}
We suppose the same assumptions as in Theorem 2.1.
Let $\omega(\varrho) \in C^1([0,1])$ be a real-valued function satisfying
\begin{equation} \label{3.4}
\omega(0) = \omega'(0) = 0, \quad \omega'(\varrho) \geq \varrho \geq 0
\end{equation}
and that $\omega'(\varrho)$ is non-decreasing, where $\omega'(\varrho)= d \omega(\varrho)/d \varrho.$
Let $\Delta, \kappa_j$ and $\kappa'_j\ (j = 0,1,2,\dots,\nu-1)$ be as in Theorem 3.1.
Then for $f \in (B^a)^l\ (a = 0,1,2,\dots)$  we have
\begin{align} \label{3.5}
& \lim_{|\Delta|\to 0}U(t,\kappa_{\nu-1})e^{-\omega(t-\tau_{\nu-1})W_{\ds}(\kappa'_{\nu-1})}
U(\kappa_{\nu-1},\kappa_{\nu-2})e^{-\omega(\tau_{\nu-1}-\tau_{\nu-2})W_{\ds}(\kappa'_{\nu-2})} 
U(\kappa_{\nu-2},\kappa_{\nu-3}) \notag \\
&\quad \cdot  \cdots e^{-\omega(\tau_{2}-\tau_{1})W_{\ds}(\kappa'_{1})}U(\kappa_{1},\kappa_{0})e^{-\omega(\tau_{1})W_{\ds}(\kappa'_{0})}U(\kappa_0,0)f
= U(t,0)f
\end{align}
in $(B^a)^l$ uniformly in $t \in [0,T]$.
\end{thm}
%
\begin{thm} \label{thm 3.5}
Suppose the same assumptions as in Theorem 2.3. Let $\Delta'$ and $\kappa_j\ (\jdots)$ be as in Theorem 3.2 and  $\omega(\varrho)$ the function introduced in Theorem 3.4.
   Let $f \in (B^a)^l\ (a = 0,1,2,\dots)$. Then, there exists $\left<\prod_{j=0}^{\nu-1}\exp\bigl\{-\omega(\tau_{j+1}-\tau_j)W_{\ds}(\kappa_j,q(\kappa_j))\bigr\}\right>f \in C^0_t([\kappa_{\nu-1},T];(B^a)^l)$   and  we have
\begin{align} \label{3.6}
& \lim_{|\Delta'|\to 0}\left<\prod_{j=0}^{\nu-1}\exp\Bigl\{-\omega(\tau_{j+1}-\tau_j)W_{\ds}(\kappa_j,q(\kappa_j))\Bigr\}\right>f \notag \\
& = \int_{\Gamma(t,x)} e^{i\hbar^{-1}S(t,0;q)}\mathcal{F}(t,0;q)f(q(0)) {\cal D}q
\end{align}
in $(B^a)^l$ uniformly in $t \in [0,T]$, where $\mathcal{F}(t,0;q)$ denotes $\mathcal{F}_w(t,0;q)$  with $W_{\ds} = 0$.
\end{thm}
	The proofs of Theorems 3.1 - 3.5 will be given in \S 5 and \S 6. Besides the example in Remark 2.2, we will give  examples of $W_{\ds}(t,x)$ satisfying Assumption 2.D
	below. For the sake of simplicity we consider only the diagonal form \eqref{2.22}.
	\begin{exmp} \label{exmp 3.1}
	Let $w_{jj}(t,x)\ (j = 1,2,\dots,l)$ be  functions bounded below satisfying 
\begin{equation*} 
|\partial_x^{\alpha}w_{jj}(t,x)| \leq C_{\alpha}, \quad |\alpha| \geq 1.
\end{equation*}
Then $W_{\ds}(t,x)$ defined by \eqref{2.22} satisfies Assumption 2.D with $w(t,x) = 0$. This $W_{\ds}(t,x)$ will be used in 
Application 4 of \S 4.
	\end{exmp}
	We can easily prove the following.
	\begin{lem}
	Let us define
\begin{equation} \label{3.7}
f(t) = \begin{cases} e^{-1/t}, & t > 0,\\
0, & t \leq 0.
\end{cases}
\end{equation}
Then $f(t)$ belongs to $C^{\infty}(\bR)$ and satisfies
\begin{align} \label{3.8}
& |f(t)| \leq 1 < \infty\ \  (t \in \bR),\quad |f^{(j)}(t)| \leq C_{j}t^{-1-j}\ \ (t \geq 1), \notag \\
& |f^{(j)}(t)| \leq C_{j} f(t)\ \ (t \geq 1)
\end{align}
for $j = 1,2,\dots$ with constants $C_{j} \geq 0$, where $f^{(j)}(t) = d^{j}f/dt^{j}$.
	\end{lem}
\begin{exmp} \label{exmp 3.2}
Let $f(t)$ be the function defined by \eqref{3.7}.  For arbitrary points $a^{(i)} \in \bR^{d}$ and arbitrary constants $b_{i} \geq 0\ (i = 1,2,\dots,l)$
we set 
\begin{equation} \label{3.9}
 w_{ii}(x) = \left|x - a^{(i)}\right|^{2}f\left(\left|x-a^{(i)}\right|-b_{i}\right), \quad  w(x) = \frac{1}{2}|x|^{2}f(|x|)
\end{equation}
and define $W_{\ds}(x)$ by \eqref{2.22}.  Then this $W_{\ds}(x)$ satisfies Assumption 2.D as shown below.  We  note 
\begin{equation} \label{3.10}
\exp \bigl\{-w_{ii}(x)\bigr\}= \begin{cases}  1, & |x - a^{(i)}| \leq b_{i},\\
\exp \bigl\{-(b_{i} + \lambda)^{2}e^{-1/\lambda}\bigr\}, & |x - a^{(i)}|= b_{i} + \lambda
\end{cases}
\end{equation}
for $\lambda > 0$ from \eqref{3.7} and \eqref{3.9}.  First, we can prove \eqref{2.11} from
\begin{equation} \label{3.11}
\lim_{|x|\to \infty} \frac{w_{ii}(x)}{w(x)} = 2.
\end{equation}
We write $w_{ii}(x)$ as $g(x) = |x - a|^{2}f(|x-a|-b)$ with $a = a^{(i)}$ and $b = b_{i}$ for a while.  We have
\begin{equation} \label{3.12}
\partial_{x_{j}}g(x) = 2(x_{j}-a_{j})f(|x-a|-b) + |x-a|^{2}\frac{x_{j}-a_{j}}{|x-a|}f'(|x-a|-b),
\end{equation}
which shows \eqref{2.13} with $|\alpha| = 1$ from the first and the second inequalities of \eqref{3.8}.  In the same way we can prove
\eqref{2.13} generally.  Applying the third inequality of \eqref{3.8} to \eqref{3.12}, we have 
\begin{align*} 
|\partial_{x_{j}}g(x)|&\leq  2\frac{1}{|x-a|}|x - a|^{2}f(|x-a|-b) + C_{1}|x - a|^{2}f(|x-a|-b) \\
& \leq  C_{2}g(x), \quad |x - a| \geq b + 1,
\end{align*}
which shows \eqref{2.12} with $|\alpha| = 1$ from \eqref{3.11}.  In the same way we can prove
\eqref{2.12} generally.  The expression \eqref{3.10} is clear from \eqref{3.7}.  $W_{\ds}(x)$ in this remark will be used in Application 3 of \S 4.
\end{exmp}
\section{Applications}
\
{\bf Application 1.} Continuous quantum measurements of  the positions of spin components.
\par
	Let $W_{\ds}(x) = W_{\ds}(t,x)$ be defined in Remark 2.2.  We proved that this $\Wsx$ satisfies Assumption 2.D. 
	It is clear that this $\Wsx$ satisfies \eqref{3.1}. Consequently we can apply Theorems 3.1 - 3.5 to the quantum systems with this $\Wstx$.
	\\
	{\bf Application 2.} Multi-split experiments.
	\par
	We first consider a thin wall with a small hole in it.  Let
	\begin{equation*}
	\mathfrak{W} = \left\{(x',x_d) \in \bR^d; -\delta/2 \leq x_d \leq \delta/2\right\}, \ \delta > 0
	\end{equation*}
be a wall with thickness of $\delta$ and 
	\begin{equation*}
	\mathfrak{H} =  \left\{(x',x_d) \in \mathfrak{W}; |x' - a'| \leq \delta'\right\}, \ \delta' > 0
	\end{equation*}
a hole with a radius of $\delta'$ in $\dW$, where $a' \in \bR^{d-1}$. We suppose that the paths which a particle follows are restricted in $\dH \cup \dW^c$,
where $\dW^c$ is the complement of $\dW$. We write $\Omega(x_{d})$ on $\bR$ and $\Omega(x')$ on $\bR^{d-1}$ defined by \eqref{1.9} as $\Omega_{1}(x_{d})$ and 
$\Omega_{d-1}(x')$ respectively.
Then, the weight function corresponding to $\dw_a(x)$ in \eqref{1.6} is given by
\begin{equation}\label{4.1}
\dw(x) = \begin{cases} \Omega_{d-1}((x'-a')/\delta'), & x \in \dW, \\
                           1, & x \notin \dW.
	\end{cases}
\end{equation}
 We can write \eqref{4.1} formally as 
\begin{equation}\label{4.2}
\dw(x) = \exp\bigl[\Omega_1(x_d/\delta)\log \Omega_{d-1}((x'-a')/\delta')\bigr].
\end{equation}
Replacing $\Omega_1(x_d/\delta)$ and $\Omega_{d-1}((x'-a')/\delta')$ in \eqref{4.2} with $k(x_d) \in \Cspace(\bR)$ and $\exp\{-h_1(x'-a')\} \in \Cspace(\bR^{d-1})$
respectively, we approximate \eqref{4.1} by
\begin{equation}\label{4.3}
\dw(x) = \exp\{-W(x)\}, \quad W(x) = h_1(x'-a')k(x_d)
\end{equation}
in the same way as we did \eqref{1.5} by \eqref{1.11}.
\par
	Next, we consider the case that the above wall $\mathfrak{W}$ has $N$ holes
	\begin{equation*}
	\mathfrak{H}_j =  \left\{(x',x_d) \in \mathfrak{W}; |x' - a'^{(j)}| \leq \delta'\right\}, j = 1,2,\dots,N,
	\end{equation*}
where $a'^{(j)} \in \bR^{d-1}$ and $|a'^{(i)}-a'^{(j)}| > \delta'\ (i \not= j)$.  Then the expressions corresponding to \eqref{4.1}, \eqref{4.2} and \eqref{4.3} are written as 
\begin{equation*}
\dw(x) = \begin{cases} \sum_{j=1}^N \Omega_{d-1}((x'-a'^{(j)})/\delta'), & x \in \dW, \\
                           1, & x \notin \dW,
	\end{cases}
\end{equation*}
\begin{equation*}
\dw(x) = \exp\Biggl[\Omega_1(x_d/\delta)\log \Biggl\{\sum_{j=1}^N \Omega_{d-1}((x'-a'^{(j)})/\delta')\Biggr\}\Biggr]
\end{equation*}
and 
\begin{align}\label{4.4}
& \dw(x) = \exp\{-W(x)\},\notag \\ 
&  W(x) = \log N- k(x_d)\biggl[\log \biggl\{\sum_{j=1}^N\Bigl(\exp -h_1(x'-a'^{(j)})\Bigr) \biggr\}\biggr],
\end{align}
respectively, where we added the constant $\log N$ to simplify the arguments below.
%
	\begin{thm} \label{thm 4.1}
	We take $W(x)$ defined by \eqref{4.4}.  Let $h_1(x')$ and $k(x_d)$ be functions such that
\begin{equation}\label{4.5}
0 \leq h_1(x'), \quad 0 \leq k(x_d) \leq 1, 
\end{equation}
\begin{equation}\label{4.6}
|\partial_{x_d}^{\alpha}k(x_d)| \leq C_{\alpha} < \infty
\end{equation}
for $|\alpha| \geq 1$.  When $N = 1$, we assume 
\begin{equation}\label{4.7}
|\partial_{x'}^{\alpha} h_1(x')| \leq C_{\alpha}<x'>
\end{equation}
for all $\alpha$. When $N \geq 2$, we assume \eqref{4.7} for $|\alpha| \leq 1$ and 
\begin{equation}\label{4.8}
|\partial_{x'}^{\alpha} h_1(x')| \leq C_{\alpha} < \infty
\end{equation}
for $|\alpha| \geq 2$. Then we have
\begin{equation}\label{4.9}
W(x) \geq 0,
\end{equation}
\begin{align}\label{4.10}
|\partial_{x'}^{\alpha}\partial_{x_d}^{\beta} W(x)| & \leq C_{\alpha}<x'> |\partial_{x_d}^{\beta} k(x_d)| \notag \\
& \leq C_{\alpha\beta}<x'>,\quad |\alpha + \beta| \geq 1.
\end{align}
\end{thm}
	From Theorem 4.1 we can easily prove the following.
\begin{cor} \label{cor 4.2}
For arbitrary points $(a'^{(1i)},a'^{(2i)},\dots,a'^{(Ni)}) \in \bR^{N(d-1)}\ (i = 1,2,\dots,l)$ we determine $w_{ii}(x)$ by \eqref{4.4}
where each $a'^{(j)}$ is replaced with $a'^{(ji)}$. We define $\Wsx$  by \eqref{2.22}. Then, under the assumptions in Theorem 4.1 this $\Wsx$ satisfies  \eqref{2.11} with $w(t,x) = 0$ and \eqref{2.13}. Hence we can apply Theorems 3.1 and 3.4 to the quantum systems with this $\Wsx$.
\end{cor}
\begin{thm}\label{thm 4.3}
For $h_1(x')$ the function  in Theorem 4.1 we also assume
\begin{equation}\label{4.11}
\liminf_{|x'|\to \infty}\,\frac{h_1(x')}{|x'|} > 0.
\end{equation}
Let $h_2(x_d)$ be a function such that 
\begin{equation}\label{4.12}
0 \leq h_2(x_d),\quad |\partial_{x_d}^{\alpha}h_2(x_d)| \leq C_{\alpha} < \infty,\ |\alpha| \geq 1.
\end{equation}
We define $W(x)$ by \eqref{4.4}, replacing $k(x_d)$ with $\exp\{-h_2(x_d)\}$.  Then this $W(x)$ satisfies
\begin{equation}\label{4.13}
<x'>e^{-h_2(x_d)} \leq C^{*}(1+W(x))
\end{equation}
with a constant $C^{*} >  0$ and
\begin{equation}\label{4.14}
|\partial_{x}^{\alpha}W(x)|  \leq C_{\alpha}<x'>e^{-h_2(x_d)}  \leq C_{\alpha}<x'>,\quad  |\alpha| \geq 1.
\end{equation}
\end{thm}
%
	From Theorem 4.3 we can easily prove the following.
	\begin{cor} \label{cor 4.4}
	Let $\Wsx$ be the matrix-valued function determined in Corollary 4.2, where we replace $k(x_d)$ with $\exp\{-h_2(x_d)\}$
	as in Theorem 4.3. Then, under the assumptions of Theorem 4.3 this $\Wsx$ satisfies  Assumption 2.D, where
	$w(t,x) = (C^{*})^{-1}<x'>\exp\{-h_2(x_d)\}$ with another constant $C^{*} > 0$. Consequently we can apply Theorems 3.1 - 3.5 to the quantum systems with this $\Wsx$.
	\end{cor}
	\begin{exmp} \label{exmp 4.1}
	 Let $h_1(x') \in C^{\infty}(\bR^{d-1})$ and $h_2(x_d) \in C^{\infty}(\bR)$ be functions  satisfying 
\begin{equation*} 
h(y) = \frac{1}{2}|y|^2\ (|y| \leq 1),\ h(y) = |y|\ (|y| \geq 2),\  h(y) \geq \frac{1}{2} \ (|y| \geq 1).
\end{equation*}
Then it is clear that these $h_1(x')$ and $h_2(x_d)$ satisfy the assumptions of Theorem 4.3.

	\end{exmp}
 \begin{proof}[Proof of Theorem 4.1]
	The inequality \eqref{4.9} is clear from \eqref{4.5}.  When $N=1$, we can easily prove \eqref{4.10}  from \eqref{4.3}  at once,
	using \eqref{4.5}-\eqref{4.7}.
Consequently we have only to prove \eqref{4.10} for $N \geq 2$.
\par
	We first consider the case of $N = 2$.  Set $a'^{(1)} = 0$ and $a'^{(2)} = a'$ for the sake of simplicity. 
From \eqref{4.4} we can write
\begin{equation}\label{4.15}
W(x) = \log 2 + \Bigl[h_1(x')-\log\Bigl\{1+e^{h_1(x')-h_1(x'-a')}\Bigr\}\Bigr]k(x_d).
\end{equation}
From 
\begin{equation}\label{4.16}
h_1(x')-h_1(x'-a') = \sum_{j=1}^{d-1}\int_0^1 \frac{\partial h_1}{\partial x_j}(x'-a'+\theta a')d\theta \,a'_j
\end{equation}
and \eqref{4.8}  we have
\begin{equation}\label{4.17}
\left|\partial_{x'}^{\alpha}e^{h_1(x')-h_1(x'-a')}\right| \leq C_{\alpha}e^{h_1(x')-h_1(x'-a')}, \quad |\alpha| \geq 1.
\end{equation}
From \eqref{4.15} we have
\begin{equation}\label{4.18}
\frac{\partial W}{\partial x_j}(x) = \biggl[\frac{\partial h_1}{\partial x_j}(x') -\frac{\partial_{x_j}\,e^{h_1(x')-h_1(x'-a')}}{1+e^{h_1(x')-h_1(x'-a')}}\biggr]k(x_d)
\end{equation}
for $j = 1,2,\dots,d-1$. Consequently, noting \eqref{4.17}, we can prove \eqref{4.10} for $\alpha \not= 0$ and all $\beta$ from the assumptions
\eqref{4.5} - \eqref{4.8}.   Next we will prove \eqref{4.10} for $\alpha = 0$ and $\beta \not= 0$. For a  point $x' \in \bR^{d-1}$ satisfying 
$h_1(x') \leq h_1(x'-a')$ we have
\begin{align*}
|\partial_{x_d}^{\beta}W(x)| & = \left|\Bigl[h_1(x')-\log\Bigl\{1+e^{h_1(x')-h_1(x'-a')}\Bigr\}\Bigr]\partial_{x_d}^{\beta}k(x_d)\right| \\
& \leq \bigl\{h_1(x')+ \log 2\bigr\}|\partial_{x_d}^{\beta}k(x_d)| \leq (1+h_1(x'))|\partial_{x_d}^{\beta}k(x_d)| \\
& = \min \bigl\{1+h_{1}(x'), 1+h_{1}(x'-a')\bigr\}\cdot|\partial_{x_d}^{\beta}k(x_d)|
\end{align*}
from \eqref{4.15}. For a  point $x' \in \bR^{d-1}$ satisfying 
$h_1(x') > h_1(x'-a')$ we have the similar inequalities.  In the end we have
\begin{equation} \label{4.19}
|\partial_{x_d}^{\beta}W(x)|  \leq \min \bigl\{1+h_{1}(x'), 1+h_{1}(x'-a')\bigr\}\cdot|\partial_{x_d}^{\beta}k(x_d)|
\end{equation}
for all $x \in \bR^d$, which shows \eqref{4.10} for $\alpha = 0$.
\par
	We will prove \eqref{4.10}  in the general case of $N = 3,4,\dots$. We have
\begin{align}\label{4.20}
W(x) = \log N + & \Bigl[h_1(x'-a'^{(1)})-\log\Bigl\{1+\sum_{j=2}^Ne^{h_1(x'-a'^{(1)})-h_1(x'-a'^{(j)})}\Bigr\}\Bigr]\notag \\
& \quad \times k(x_d)
\end{align}
corresponding to \eqref{4.15}. Hence, noting \eqref{4.17}, we can prove \eqref{4.10} for $\alpha \not= 0$ and all $\beta$ as in the proof of the case of $N = 2$.  Let $\beta \not= 0$. Then we can prove
\begin{equation} \label{4.21}
|\partial_{x_d}^{\beta}W(x)|  \leq \min \Bigl\{1+h_1(x'-a'^{(j)}); j = 1,2,\dots,N\Bigr\}\cdot|\partial_{x_d}^{\beta}k(x_d)| 
\end{equation}
as in the proof of \eqref{4.19}.
   The expression \eqref{4.21} shows \eqref{4.10} for $\alpha = 0$.  Thus we have completed the proof.
\end{proof}
	\begin{proof}[Proof of Theorem 4.3]
We see from \eqref{4.12} that $k(x_d) = \exp\{-h_2(x_d)\}$ satisfies \eqref{4.5} and \eqref{4.6}.  Hence Theorem 4.1 holds.
Noting \eqref{4.12}, from \eqref{4.10} we have
\begin{align}\label{4.22}
|\partial_{x'}^{\alpha}\partial_{x_d}^{\beta}W(x)| &
 \leq C_{\alpha}<x'>\left|\partial_{x_d}^{\beta}e^{-h_2(x_d)}\right| \notag \\
 &  \leq C_{\alpha\beta}<x'>e^{-h_2(x_d)},\quad  |\alpha+\beta| \geq 1,
\end{align}
which shows \eqref{4.14}.
\par
	We will  prove \eqref{4.13}.  Let $N= 1$.  Then, using \eqref{4.3} and the assumption \eqref{4.11}, we can prove \eqref{4.13} because we have
	\begin{align*}
	& <x'>e^{-h_2(x_d)} \leq C_1<x'-a'>e^{-h_2(x_d)} \\
	&\quad  \leq C_2\Bigl(1+ h_1(x'-a')e^{-h_2(x_d)}\Bigr) = C_2\bigl(1+W(x)\bigr)
	\end{align*}
	with constants $C_j \geq 0\ (j= 1,2)$. Let $N=2$.
	  For the sake of simplicity we set
$a'^{(1)} = 0$ and $a'^{(2)} = a'$. For a  point $x' \in \bR^{d-1}$ satisfying 
$h_1(x') \leq h_1(x'-a')$ we have
\begin{align} \label{4.23}
W(x) & = \log 2 + \Bigl[h_1(x')-\log\Bigl\{1+e^{h_1(x')-h_1(x'-a')}\Bigr\}\Bigr]e^{-h_2(x_d)}
 \notag \\
& \geq \Bigl[h_1(x')+ \log 2 -\log\Bigl\{1+e^{h_1(x')-h_1(x'-a')}\Bigr\}\Bigr]e^{-h_2(x_d)} \notag \\
& \geq h_1(x')e^{-h_2(x_d)} = \min \bigl\{h_{1}(x'), h_{1}(x'-a')\bigr\}\cdot e^{-h_2(x_d)}
\end{align}
from \eqref{4.20}.
As in the proof of   \eqref{4.19} we can prove
\begin{equation*} 
W(x) \geq \min \bigl\{h_1(x'), h_1(x'-a')\bigr\}\cdot e^{-h_2(x_d)} 
\end{equation*}
for all $x \in \bR^{d}$, which shows   \eqref{4.13} as in the case of $N=1$.  In general for $N = 3,4,\dots$   
we have
\begin{equation} \label{4.24}
W(x) \geq  \min \bigl\{h_1(x'-a'^{(j)}); j = 1,2,\dots,N\bigr\}\cdot e^{-h_2(x_d)}
\end{equation}
as in the proof of $N=2$.
Consequently we can prove \eqref{4.13}  as in the cases of $N = 1,2$.
	\end{proof}
	
	\begin{rem} \label{rem 4.1}
We consider a finite number of thin walls with some holes.  For each wall we take $\Wsx$ determined in Corollary 4.2 (resp. Corollary 4.4). Let $\widetilde{W_{\ds}}(x)$ be the sum of these $\Wsx$. Then, $\widetilde{W_{\ds}}(x)$ satisfies the properties stated in Corollary 4.2 (resp. Corollary 4.4), where we take the sum of all $w(t,x)$ corresponding to $W_{\ds}(x)$ as $w(t,x)$
for $\widetilde{W_{\ds}}(x)$ in Corollary 4.4.
	\end{rem}
	\begin{rem} \label{rem 4.2}
We suppose that a free particle is  within the distance of $\delta > 0$ of $a \in \bR$ at a fixed time $t_0 \in (0,T)$.  Under this restriction on paths followed by the particle, Feynman and Hibbs expressed the probability amplitude by the path integral in \S 3-2 of \cite{Feynman-Hibbs}. See (3.20) and (3.21) in \cite{Feynman-Hibbs}. It should be noted that their Feynman path integral is not the one in Corollaries 4.2 and 4.4 of the present paper, but the one  expressed by our \eqref{2.20} or \eqref{2.21} with $W_{\ds} = 0$.
\end{rem}
	{\bf Application 3.} The quantum Zeno-effect.   
	\par
	Suppose that the particle with the $i$-th component is always restricted in the ball $\{x \in \bR^{d}; |x - a^{(i)}| \leq b_{i}\}$ with $a^{(i)} \in \bR^{d}$ and $b_{i} > 0\ (i = 1,2,\dots,l)$ when we perform the measurement of their positions during the time interval $[0,t]$.  Then, the weight corresponding to $\dw_{a}(x)$ in \eqref{1.6} is given by the diagonal matrix
\begin{equation} \label{4.25}
 \begin{pmatrix} 
 \Omega_{1}(|x-a^{(1)}|/b_{1})                                       \\
            & \Omega_{1}(|x-a^{(2)}|/b_{2})        &        & \text{\huge{0}}                           \\
    &  & \ddots     \\
 & \text{\huge{0}}    & & \ddots \\
                        & &    &          &\Omega_{1}(|x-a^{(l)}|/b_{l})
  \end{pmatrix}.
 \end{equation}
 Let $\Wsx$ and $w(x)$ be the functions defined in Example 3.2.  Then, noting \eqref{3.10}, we will approximate \eqref{4.25} by $n\Wsx$ with a large integer $n$.  Since $\Wsx$ satisfies Assumption 2.D, we can apply Theorem 3.1 - 3.5 to the quantum systems with this $n\Wsx$.
 \begin{rem}  \label{rem 4.3}
 Let $H(t)$ in \eqref{1.4}  be the free Hamiltonian operator $-\Delta$, i.e. $A_{j} = 0\ (j= 1,2,\dots,d), V = 0$ and $m=\hbar/2$ , and $\mathcal{O}$ an open set in 
 $\bR^{d}$ with a smooth boundary.  It was proved  in \cite{Exner-Ichinose, Friedman} that \eqref{1.18} in the present paper is equal to
 the solution
 $\Bigl[\exp\{-i\hbar^{-1}t(-\Delta_{\mathcal{O}})\}\Bigr]\chi_{\mathcal{O}}(\cdot)f$ with $\chi_{\mathcal{O}}(\cdot)f$ at $t = 0$ to the Dirichlet problem in $\mathcal{O}$
 for $-\Delta$, where $f \in L^{2}(\bR^{d})$.
 \end{rem}
 {\bf Application 4.} Aharonov-Bohm effect. 
 \par
  Let $\mathcal{O}$ be a bounded open set in $\bR^{d}$ and suppose
 \[
 \mathrm{supp}\,  E_{j}(t,\cdot) \subset \mathcal{O}, \quad   \mathrm{supp}\, B_{jk}(t,\cdot) \subset \mathcal{O}
 \]
 all of the time $[0,T]$ for all $E_{j}$ and $B_{jk}$ in \eqref{1.1}, where $\mathrm{supp}\,  E_{j}(t,\cdot)$ denotes the support of the  function $E_{j}(t,x)$ of $x$ variables.  Assume that a particle can never enter in $\mathcal{O}$.  This is the problem that Aharonov and Bohm discussed in \cite{Aharonov-Bohm}.  The weight corresponding to $\dw_{a}(x)$ in \eqref{1.6} is given by
\begin{equation}\label{4.26}
\dw_{\mathcal{O}}(x) = \begin{cases} 0, & x \in \mathcal{O}, \\
                           1, & x \notin \mathcal{O}.
	\end{cases}
\end{equation}
Let $\mathcal{O}'$ be another bounded open set in $\bR^{d}$ such that  $\overline{\mathcal{O}} \subset \mathcal{O}'$, where $\overline{\mathcal{O}}$ is the closure of $\mathcal{O}$.  Taking a non-negative function $h_{\mathcal{O}}(x) \in C^{\infty}_{0}(\bR^{d})$
such that 
\begin{equation}\label{4.27}
h_{\mathcal{O}}(x) = \begin{cases} 1, & x \in \mathcal{O}, \\
                           0, & x \notin \mathcal{O}'
	\end{cases}
\end{equation}
(cf. Theorem 1 on page 93 of \cite{Matsushima}),  we will approximate \eqref{4.26} by $\exp\{-nh_{\mathcal{O}}(x)\}$
with a large integer $n$.  That is, we define $\Wstx$ by \eqref{2.22}, setting $w_{ii}(t,x) = nh_{\mathcal{O}}(x)$.  Then we see from Example 3.1 that this $\Wstx$ satisfies Assumption 2.D with $w(t,x) = 0$.  Hence we can apply Theorems 3.1 - 3.5 to the quantum systems with this $\Wstx$.
\begin{rem} \label{rem4.4}
Taking the parameter $n$ to infinity in  Applications 3 and 4, we  can obtain the results corresponding to \eqref{1.18} or \eqref{1.19} by the study on what function the solution $u_n(t)$ to the  equation \eqref{2.3} with the parameter $n$ converges to as $n \to \infty$.  Since our aim in the present paper is to give a mathematical formulation,  we don't study it in the present paper. Results on it will be published somewhere.
\end{rem}
%
\section{Proof of Theorem 3.1}
For the sake of simplicity we always assume $\hbar =1$ hereafter.
Let $C_{W} \geq 0$ be the constant  in Assumption 2.D.  If we replace $u(t)$ in \eqref{2.3} with $e^{-C_{W}(t-s)}u(t)$, $\Wstx$ in \eqref{2.3} changes to $\Wstx + C_{W}$.  Consequently, replacing $\Wstx$ with $\Wstx + C_W$, we can assume $C_{W} = 0$ without loss of generality in the proof of Theorems 2.2 - 2.5 and 3.1 - 3.5.   Hence hereafter throughout \S 5, \S 6 and the appendix we always assume $C_{W} = 0$.
\par
	Let $p(x,\xi,x')$ be a function on $\bR^{3d}$ such that  there exists an integer $M \in \bR$ satisfying
 \begin{equation} \label{5.1}
|\partial_{\xi}^{\alpha}\partial_{x}^{\beta}\partial_{x'}^{\beta'}p(x,\xi,x')| 
\leq C_{\alpha\beta\beta'}(1+|x|+|\xi| + |x'|)^M 
 \end{equation}
for all $\alpha, \beta$ and $\beta'$ with constants $ C_{\alpha\beta\beta'}$.  Then we define the pseudo-differential operator $p(x,D_x,x')f$ with the symbol $p(x,\xi,x')$
by
 \begin{equation} \label{5.2}
\int e^{ix\cdot \xi}\ \dbar\xi \int e^{-ix'\cdot \xi}p(x,\xi,x')f(x')dx',\quad \dbar\xi = (2\pi)^{-d}d\xi
 \end{equation}
for $f \in \Sspace(\bR^d)$, where $x\cdot \xi = \sum_{j=1}^dx_j\xi_j$ (cf. \S 2 in chap.2 of \cite{Kumano-go}).
\par
	We have proved the following in Lemmas 2.3 - 2.5 of \cite{Ichinose 1995}.
	\begin{lem} \label{lem 5.1}
(i) For $a = 0,1,2,\dots$ there exist constant $\mu_a \geq 0$ and $q_a(x,\xi)$ such that
 \begin{align} \label{5.3}
q_a(x,D_x)f & = (\mu_a + <x>^a + <D_x>^a)^{-1}f \notag \\
& \equiv \Lambda_a(x,D_x)^{-1}f, \quad f \in \Sspace,
 \end{align}
 \begin{equation} \label{5.4}
 |\partial_{\xi}^{\alpha}\partial_{x}^{\beta}q_a(x,\xi)| 
\leq C_{\alpha\beta}(1+|x|^a+|\xi|^a)^{-1} 
 \end{equation}
for all $\alpha$ and $\beta$. (ii) The norm $\Vert f \Vert_a$ for $f \in B^a\ (a = 0, 1,2,\dots)$ is equivalent to $\Vert\Lambda_a(x,D_x)f\Vert$. (iii) For $p(x,\xi,x')$ a function satisfying \eqref{5.1} we have
 \begin{equation} \label{5.5}
\Vert p(x,D_x,x')f\Vert_a \leq
 C_{a}\Vert f \Vert_{a+M}
 \end{equation}
for $a = 0,1,2,\dots$.
	\end{lem}
	\begin{pro} \label{pro 5.2}
Let $U_w(t,s)f$ be the solution to the equation \eqref{2.3} with $u(s) = f \in (B^a)^l \ (a = 0,1,2,\dots)$ determined in Theorem 2.1.
Then under the assumptions of Theorem 2.1 we have
 \begin{equation} \label{5.6}
\Vert U_w(t,s)f\Vert_a \leq e^{K_a(t-s)}\Vert f \Vert_{a}, \ 0 \leq s \leq t \leq T
 \end{equation}
with constants $K_a \geq 0$, where $K_0 =  0$.
	\end{pro}
	\begin{proof}
We first note that we are assuming 
 \begin{equation} \label{5.7}
\Wstx \geq 0, \ (t,x) \in \domain.
\end{equation}
Let $f \in B^2$. Then we see $u(t) = U_w(t,s)f \in C^0_t([s,T];(B^2)^l) \cap C^1_t([s,T];(L^2)^l)$ from Theorem 2.1.
Hence we have
 \begin{equation} \label{5.8}
 \frac{d}{dt}\Vert u(t)\Vert^2 = 2\rittaire(\partial_tu(t),u(t)) = -2(W_{\ds}(t) u(t),u(t)) \leq 0
   \end{equation}
  from the equation \eqref{2.3}, which shows 
 \begin{equation*} 
\Vert U_w(t,s)f\Vert \leq \Vert f \Vert, \ 0 \leq s \leq t \leq T
 \end{equation*}
for  $f \in (B^2)^l$.  Let $f \in (L^2)^l$ and take $f_j \in (B^2)^l\ (j = 1,2,\dots)$ such that $f_j \to f$ in $(L^2)^l$ as $j \to \infty$.
Since we have  $\Vert U_w(t,s)f_j\Vert \leq \Vert f_j \Vert$, we can see \eqref{5.6} with $a = 0$ for $f \in (L^2)^l$ by \eqref{2.15}.
\par
	Let $f \in (B^3)^l$.  Then we see $u(t) = U_w(t,s)f \in C^0_t([s,T];(B^3)^l) \cap C^1_t([s,T];(B^1)^l)$ from Theorem 2.1.  We have
 \begin{align} \label{5.9}
 i\partial_t\,\partial_{x_j} u(t)& = (H + H_{\ds} - iW_{\ds})\partial_{x_j} u(t) \notag \\
 & +  \bigl[\partial_{x_j},H + H_{\ds} - iW_{\ds} \bigr]\Lambda_1^{-1}\bigl(\Lambda_1u(t)\bigr)
    \end{align}
  from \eqref{2.3}, where $[\cdot,\cdot]$ denotes the commutator. From the assumptions we can easily prove
 \begin{equation} \label{5.10}
\sup_{0 \leq t \leq T}\Vert \bigl[\partial_{x_j},H + H_{\ds} - iW_{\ds}\bigr]\Lambda_1^{-1}\Vert_{L^2\to L^2} < \infty
   \end{equation}
 as in the proof of Lemma 3.1 of \cite{Ichinose 1995} by using the Calder\'on-Vaillancourt theorem (cf. Theorem 1.6 on page 224 of \cite{Kumano-go}, Theorem 5.1 of \cite{Z}), where $\Vert\cdot\Vert_{L^2\to L^2}$ denotes the operator norm from $(L^2)^l$
into $(L^2)^l$.  Consequently we can prove 
 \begin{align*}
 \frac{d}{dt}\Vert \partial_{x_j}u(t)\Vert^2 & = 2\rittaire \bigl(\partial_t\partial_{x_j}u(t),\partial_{x_j}u(t)\bigr) \leq C_1\Vert\Lambda_1u(t)\Vert\cdot \Vert\partial_{x_j}u(t)\Vert
   \end{align*}
  as in the proof of \eqref{5.8}.
In the same way we have
\begin{equation*} 
\frac{d}{dt}\Vert x_ju(t)\Vert^2 \leq C_2\Vert\Lambda_1u(t)\Vert\cdot \Vert x_{j}u(t)\Vert.
 \end{equation*}
Hence, setting
 \begin{align*} 
v_{1\eta}(t) = \sqrt{\Vert u(t)\Vert^{2}+ \eta} + \sum_{j=1}^{d}\Bigl(\sqrt{\Vert x_{j}u(t)\Vert^{2}+ \eta} + \sqrt{\Vert \partial_{x_{j}}u(t)\Vert^{2}+ \eta} \ \Bigr)
\end{align*}
with a constant $\eta > 0$ for a while, and noting \eqref{2.1} and (ii) of Lemma 5.1, we have
 \begin{equation} \label{5.11}
\frac{d}{dt}v_{1\eta}(t) \leq \frac{(C_1+C_2)d}{2}\Vert\Lambda_1u(t)\Vert \leq K_{1}\Vert u(t)\Vert_{1} \leq K_{1}v_{1\eta}(t) 
\end{equation}
together with \eqref{5.8}, where  $K_{1}$ is a constant independent of $\eta > 0$.  Then we have
 \begin{equation*} 
v_{1\eta}(t) \leq  e^{K_{1}(t-s)}v_{1\eta}(s), \ 0 \leq s \leq t \leq T,
\end{equation*}
 which shows 
 \begin{equation} \label{5.12}
\Vert U_w(t,s)f\Vert_1 \leq e^{K_1(t-s)}\Vert f \Vert_{1}, \ 0 \leq s \leq t \leq T
 \end{equation}
for $f \in (B^3)^l$ as $\eta \to 0$.  From \eqref{5.12} we can prove \eqref{5.6} with $a = 1$ for $f \in (B^1)^l$ as in the  proof of the case of $a = 0$.
\par
	We can prove \eqref{5.6} generally as in the proof of \eqref{5.6} with $a = 0, 1$.  To make sure we will give the outline of its proof.  Let $f \in (B^{a+2})^l\ (a = 2,3,\dots)$.  Then $u(t) = U_w(t,s)f \in C^0_t([s,T];(B^{a+2})^l) \cap C^1_t([s,T];(B^{a})^l)$. 
	 For $|\alpha| = a$ we have
 \begin{align} \label{5.13}
 i\partial_t\,\partial_{x}^{\alpha} u(t)& = (H + H_{\ds} - iW_{\ds})\partial_{x}^{\alpha} u(t) \notag \\
 & +  \bigl[\partial_{x}^{\alpha},H + H_{\ds} - iW_{\ds} \bigr]\Lambda_a^{-1}\bigl(\Lambda_au(t)\bigr),
    \end{align}
 \begin{equation} \label{5.14}
\sup_{0 \leq t \leq T}\Vert \bigl[\partial_{x}^{\alpha},H + H_{\ds} - iW_{\ds}\bigr]\Lambda_a^{-1}\Vert_{L^2\to L^2} < \infty,
   \end{equation}
where we used $|\xi|^{a-1}|x| \leq |\xi|^{a}+|x|^a$ to prove \eqref{5.14}.  Hence we have
 \begin{equation*}
 \frac{d}{dt}\Vert \partial_{x}^{\alpha}u(t)\Vert^2 
 \leq C_3\Vert\Lambda_au(t)\Vert\cdot \Vert\partial_{x}^{\alpha}u(t)\Vert .
   \end{equation*}
  In the same way we have
 \begin{equation*}
 \frac{d}{dt}\Vert x^{\alpha}u(t)\Vert^2 
  \leq  C_4\Vert\Lambda_au(t)\Vert\cdot \Vert x^{\alpha}u(t)\Vert .
  \end{equation*}
  Consequently, setting 
 \begin{align*} 
v_{a\eta}(t) = \sqrt{\Vert u(t)\Vert^{2}+ \eta} + \sum_{|\alpha| = a}\Bigl(\sqrt{\Vert x^{\alpha}u(t)\Vert^{2}+ \eta}  + \sqrt{\Vert \partial_{x}^{\alpha}u(t)\Vert^{2}+ \eta} \ \Bigr)
\end{align*}
with  a constant $\eta > 0$ for a while, as in the proof of \eqref{5.11} we have
 \begin{equation} \label{5.15}
\frac{d}{dt}v_{a\eta}(t) \leq C_{5}\Vert\Lambda_au(t)\Vert \leq K_{a}\Vert u(t)\Vert_{a} \leq K_{a}v_{a\eta}(t) 
\end{equation}
with constatns $C_{5}$ and $K_{a}$  independent of $\eta > 0$, which shows
  \par
 \begin{equation} \label{5.16}
\Vert U_w(t,s)f\Vert_a \leq e^{K_a(t-s)}\Vert f \Vert_{a}, \ 0 \leq s \leq t \leq T
 \end{equation}
for $f \in (B^{a+2})^l$ as $\eta \to 0$.  From \eqref{5.16} we can prove \eqref{5.6} for $f \in (B^a)^l$ as in the proof of the cases of $a = 0, 1$.
	\end{proof}
	Let $U(t,s)f$ be the solution to \eqref{2.3} with $\Wstx = 0$ as denoted in Theorem 2.1 and $\Wstx \in \Mlc$ the matrix-valued function  in \eqref{2.3}.
We write 
 \begin{align} \label{5.17}
& \Ulw(t,s;\kappa')f = e^{-(t-s)W_{\ds}(\kappa')} U(t,s)f, \notag \\
 &  \Urw(t,s;\kappa')f = U(t,s)\left(e^{-(t-s)W_{\ds}(\kappa')}f\right)
    \end{align}
  with a constant $\kappa' \in [0,T]$.  We often write $\Ujw(t,s;\kappa')f\ (J = L,W)$ as $\Ujw(t,s)f$, omitting $\kappa'$.
  \begin{pro} \label{pro 5.3}
We suppose the same assumptions as in Theorem 2.1.  Then there exist constants $C_a \geq 0$ and $C'_a \geq 0\ (a = 0,1,2,\dots)$
such that 
 \begin{equation} \label{5.18}
\Vert e^{-(t-s)W_{\ds}(\kappa')}f\Vert_a \leq e^{C_a(t-s)}\Vert f \Vert_{a}, 
 \end{equation}
 \begin{equation} \label{5.19}
\Vert \Ujw(t,s;\kappa')f\Vert_a \leq e^{C'_a(t-s)}\Vert f \Vert_{a}, \ 0 \leq s \leq t \leq T
 \end{equation}
 for $J = L, R$ and all $\kappa' \in [0,T]$.
  \end{pro}
  \begin{proof}
We can easily see
 \begin{align} \label{5.20}
\partial_{x_j}&\,e^{-(t-s)W_{\ds}(\kappa',x)}  = -(t-s)\int_0^1 e^{-(1-\theta)(t-s)W_{\ds}(\kappa',x)}\notag \\
& \times \left\{\partial_{x_j}W_{\ds}(\kappa',x)\right\}e^{-\theta(t-s)W_{\ds}(\kappa',x)}d\theta
 \end{align}
 (cf. Lemma 3.1 of \cite{Ichinose 2014}). Using $\Wstx \geq 0$ and \eqref{2.13}, from  \eqref{5.20} we have
 \begin{align*} 
& \Vert e^{-(t-s)W_{\ds}(\kappa')}f\Vert_1 = \Vert e^{-(t-s)W_{\ds}(\kappa')}f\Vert \\
& \quad + \sum_{j=1}^d\Bigl\{\Vert \partial_{x_j}\bigl(e^{-(t-s)W_{\ds}(\kappa')}f\bigr) \Vert + \Vert x_j\bigl(e^{-(t-s)W_{\ds} (\kappa')}f\bigr) \Vert\Bigr\} \\
& \leq \Vert f \Vert_1 + C(t-s)\Vert f \Vert_1 \leq e^{C(t-s)}\Vert f \Vert_1, \ 0 \leq s \leq t \leq T
 \end{align*}
 with a constant $C \geq 0$.  In the same way we can prove \eqref{5.18} in general, where we use 
 \begin{equation} \label{5.21}
<x>^j<\xi>^{a-j}\,  \leq\  <x>^a + <\xi>^a, \ j = 1,2,\dots,a
 \end{equation}
 and (ii) of Lemma 5.1.
 \par
 	From \eqref{5.6} and \eqref{5.18} we have
 \begin{align*} 
&\Vert \Ulw(t,s)f\Vert_a = \bigl\Vert e^{-(t-s)W_{\ds}(\kappa')}U(t,s)f \bigr\Vert_a \\
&\leq e^{C_a(t-s)}\Vert U(t,s)f \Vert_{a} \leq e^{(C_a+K_a)(t-s)}\Vert f \Vert_{a}.
 \end{align*}
  In the same way we have
 \begin{align*} 
&\Vert \Urw(t,s)f\Vert_a = \Bigl\Vert U(t,s)\left(e^{-(t-s)W_{\ds}(\kappa')}f\right) \Bigr\Vert_a \\
&\leq e^{K_a(t-s)}\Vert e^{-(t-s)W_{\ds}(\kappa')}f \Vert_{a} \leq e^{(C_a+K_a)(t-s)}\Vert f \Vert_{a}.
 \end{align*}
 Thus we  obtain  \eqref{5.19}.
  \end{proof}
  \begin{lem} \label{lem 5.4}
We suppose the same assumptions as in Theorem 3.1.  Then we have
 \begin{equation}  \label{5.22}
  \bigl[i\partial_{t} - H(t) - H_{\ds}(t,x)  + iW_{\ds}(t,x)\bigr]\Ujw(t,s;\kappa')f = \Rjw(t,s;\kappa')f,
 \end{equation}
 \begin{equation}  \label{5.23}
 \Vert\Rjw(t,s;\kappa')f\Vert_{a} \leq C_{a}\bigl\{(t-s)+\sigma(|t-\kappa'|)\bigr\}\Vert f \Vert_{a+4},\ 0 \leq s \leq t \leq T
 \end{equation}
 for $J = L, R, \kappa' \in [0,T]$ and $a = 0,1,2,\dots$.
  \end{lem}
  \begin{proof}
  For simplicity we write $H_{\ds}(t,x)$ and $W_{\ds}(t,x)$ as $H_{\ds}(t)$ and $W_{\ds}(t)$ respectively. 
  Let $f\in (B^{a+2})^l\ (a = 0,1,2,\dots)$. Then we see  $U(t,s)f \in C^0_t([s,T];(B^{a+2})^l) \\
  \cap C^1_t([s,T];(B^{a})^l)$
  from Theorem 2.1. Hence we have
\begin{align*}
& i\bigl\{U(t,s)f - f\bigr\} = i(t-s)\int_{0}^{1}\frac{\partial U}{\partial t}(s + \theta\rho,s)fd\theta \\
& = (t-s)\int_{0}^{1}\big\{H(s+\theta\rho) + H_{\ds}(s+\theta\rho)\bigr\}U(s + \theta\rho,s)fd\theta, \ \rho = t-s
\end{align*}
from \eqref{2.3}.  Consequently we have
 \begin{equation}  \label{5.24}
 \Vert U(t,s)f - f\Vert_{a} \leq C_{a}|t-s|\Vert f \Vert_{a+2},\ (t,s \in [0,T])
 \end{equation}
 for $f\in (B^{a+2})^l\ (a = 0,1,2,\dots)$ from Theorem 2.1 and (iii) of Lemma 5.1.
 \par
 	From \eqref{2.3} we have
\begin{align}  \label{5.25}
 & \bigl[i\partial_{t} - H(t) - H_{\ds}(t)  + iW_{\ds}(t)\bigr]U_{RW}(t,s;\kappa')f \notag\\
 & =  \bigl[i\partial_{t} - H(t) - H_{\ds}(t)  + iW_{\ds}(t)\bigr]U(t,s)\bigl(e^{-(t-s)W_{\ds}(\kappa')}f\bigr) \notag \\
 & =  \bigl[i\partial_{t} - H(t) - H_{\ds}(t)\bigr]U(t,s)\bigl(e^{-(t-s)W_{\ds}(\kappa')}f\bigr)
 + iW_{\ds}(t)U(t,s)\bigl(e^{-(t-s)W_{\ds}(\kappa')}f\bigr)
  \notag \\
 & = - iU(t,s)W_{\ds}(\kappa')\bigl(e^{-(t-s)W_{\ds}(\kappa')}f\bigr) + iW_{\ds}(t)U(t,s)\bigl(e^{-(t-s)W_{\ds}(\kappa')}f\bigr)  \notag \\
 & = - i\bigl\{ U(t,s) - I\big\}W_{\ds}(\kappa')\bigl(e^{-(t-s)W_{\ds}(\kappa')}f\bigr) + iW_{\ds}(t)\bigl\{ U(t,s) - I\big\}\bigl(e^{-(t-s)W_{\ds}(\kappa')}f\bigr)
  \notag \\
 & - i \bigl\{W_{\ds}(\kappa') - W_{\ds}(t)\big\}\bigl(e^{-(t-s)W_{\ds}(\kappa')}f\bigr) = R_{RW}(t,s;\kappa')f.
 \end{align}
 Applying \eqref{3.1}, \eqref{5.18} and \eqref{5.24} to \eqref{5.25}, we have
 \begin{align*}  
 & \Vert R_{RW}(t,s;\kappa')f \Vert_{a} \leq C_a\Bigl[ (t-s)\Vert W_{\ds}(\kappa')\bigl(e^{-(t-s)W_{\ds}(\kappa')}f\bigr)\Vert_{a+2}  \\
 & \quad + \Vert (U(t,s) - I)\bigl(e^{-(t-s)W_{\ds}(\kappa')}f\bigr)\Vert_{a+2}
 + \sigma(|t-\kappa'|)\Vert e^{-(t-s)W_{\ds}(\kappa')}f\Vert_{a+2} \Bigr]\\
 & \leq C'_a(t-s)\Vert e^{-(t-s)W_{\ds}(\kappa')}f\Vert_{a+4} + C'_a\sigma(|t-\kappa'|)\Vert f\Vert_{a+2}\\
 & \leq C'_a\bigl\{(t-s) + \sigma(|t-\kappa'|)\bigr\}\Vert f\Vert_{a+4},
  \end{align*}
 which shows  \eqref{5.23} for $J = R$.
 \par
 	We also have
 \begin{align}  \label{5.26}
 & \bigl[i\partial_{t} - H(t) - H_{\ds}(t)  + iW_{\ds}(t)\bigr]U_{LW}(t,s;\kappa')f \notag\\
 & =  \bigl[i\partial_{t} - H(t) - H_{\ds}(t)  + iW_{\ds}(t)\bigr]e^{-(t-s)W_{\ds}(\kappa')}U(t,s)f\notag \\
 & =  \Bigl[i\partial_{t} - H(t) - H_{\ds}(t),e^{-(t-s)W_{\ds}(\kappa')}\Bigr]U(t,s)f + iW_{\ds}(t)U_{LW}(t,s)f
 \notag \\
 & =  i \bigl\{W_{\ds}(t)- W_{\ds}(\kappa')\big\}U_{LW}(t,s)f - \bigl[H(t),e^{-(t-s)W_{\ds}(\kappa')}\bigr]U(t,s)f \notag \\
 &  - \bigl[H_{\ds}(t),e^{-(t-s)W_{\ds}(\kappa')}\bigr]U(t,s)f = R_{LW}(t,s;\kappa')f \equiv -\sum_{j=1}^{3}I_{j}(t,s;\kappa')f.
 \end{align}
 Then we see
 \begin{equation}  \label{5.27}
 \Vert I_{1}(t,s;\kappa')f\Vert_{a} \leq C_a\sigma(|t-\kappa'|)\Vert \Ulw(t,s)f \Vert_{a+2} \leq C'_a\sigma(|t-\kappa'|)\Vert f \Vert_{a+2}
 \end{equation}
 from \eqref{3.1} and \eqref{5.19}.  Noting 
 \[
 [PQ,R] = P[Q,R] + [P,R]Q
 \]
 for operators $P, Q$ and $R$, from \eqref{1.4} we have
 \begin{align}  \label{5.28}
 & I_{2}(t,s;\kappa')f = \bigl[H(t),e^{-(t-s)W_{\ds}(\kappa')}\bigr]U(t,s)f  \notag \\
 & = \frac{1}{2m}\sum_{j=1}^{d}(i^{-1}\partial_{x_{j}} - A_{j})\Bigl[i^{-1}\partial_{x_{j}},e^{-(t-s)W_{\ds}(\kappa',x)}\Bigr]U(t,s)f
 \notag \\
 &\quad + \frac{1}{2m}\sum_{j=1}^{d}\Bigl[i^{-1}\partial_{x_{j}},e^{-(t-s)W_{\ds}(\kappa',x)}\Bigr](i^{-1}\partial_{x_{j}} - A_{j})U(t,s)f.
 \end{align}
 Hence, using \eqref{2.15}, \eqref{5.18} and \eqref{5.20}, we have
 \begin{align}  \label{5.29}
 &\Vert I_{2}(t,s;\kappa')f\Vert_{a} \leq C_a \sum_{j=1}^{d}\Bigl\Vert\Bigl[i^{-1}\partial_{x_{j}},e^{-(t-s)W_{\ds}(\kappa',x)}\Bigr]U(t,s)f\Bigr\Vert_{a+1}
 \notag \\
 & \qquad  + C_a(t-s)\sum_{j=1}^{d}\Vert (i^{-1}\partial_{x_{j}} - A_{j})U(t,s)f\Vert_{a+1}  \notag \\
 & \leq C'_a(t-s)\Vert U(t,s)f\Vert_{a+2} \leq C''_a(t-s)\Vert f\Vert_{a+2}
  \end{align}
 from \eqref{2.9}, \eqref{2.13} and \eqref{2.15}.
 We write 
 \begin{align}  \label{5.30}
 & I_{3}(t,s;\kappa')f = \bigl[H_{\ds}(t),e^{-(t-s)W_{\ds}(\kappa')}\bigr]U(t,s)f  = \bigl\{H_{\ds}(t)e^{-(t-s)W_{\ds}(\kappa')} \notag \\
 &\quad - e^{-(t-s)W_{\ds}(\kappa')}H_{\ds}(t)\bigr\}U(t,s)f =  H_{\ds}(t)\bigl\{e^{-(t-s)W_{\ds}(\kappa')}- I\bigr\}U(t,s)f  \notag \\
 & \qquad - \bigl\{e^{-(t-s)W_{\ds}(\kappa')}- I\bigr\}H_{\ds}(t)U(t,s)f.
 \end{align}
 We have
 \begin{equation*} 
\left\Vert \left(e^{-(t-s)W_{\ds}(\kappa',x)} - I\right)f \right\Vert_a  \leq  C_a(t-s)\Vert f \Vert_{a+2}
\end{equation*}
for $a = 0,1,2,\dots$ from \eqref{2.13} and \eqref{5.18} because of
 \begin{equation*} 
e^{-(t-s)W_{\ds}(\kappa',x)} - I  = -(t-s)\int_0^1W_{\ds}(\kappa',x)e^{-\theta(t-s)W_{\ds}(\kappa',x)}d\theta.
\end{equation*}
 Consequently, we obtain 
 \begin{equation}  \label{5.31}
 \Vert I_{3}(t,s)f\Vert_{a} \leq C'_a(t-s)\Vert U(t,s)f \Vert_{a+4} \leq C''_a(t-s)\Vert f \Vert_{a+4}
 \end{equation}
 from \eqref{2.14} and \eqref{2.15}.
 Thus we can prove \eqref{5.23} for $J = L$ from \eqref{5.26}, \eqref{5.27}, \eqref{5.29} and \eqref{5.31}.
  \end{proof}
  \begin{pro} \label{pro 5.5}
  Under the assumptions of Theorem 3.1 we have
 \begin{equation}  \label{5.32}
 \Vert \Ujw(t,s;\kappa')f - U_w(t,s)f \Vert_{a} \leq C_a\rho\{\rho+\sigma(|t-\kappa'|)\}\Vert f \Vert_{a+6}, \quad 0 \leq s \leq t \leq T
 \end{equation}
 with $\rho = t -s$ for $J = L, R, \kappa' \in [0,T]$ and $\adots$.
  \end{pro}
\begin{proof}
We note that \eqref{5.22} and \eqref{5.23} in Lemma 5.4 are corresponding to \eqref{4.1} and \eqref{4.2} in \cite{Ichinose 2023}, respectively.  Accordingly,
we can prove \eqref{5.33} below from Lemma 5.4  as in the proof of (4.14) in \cite{Ichinose 2023}, where we replace $\mathcal{C}(t,s), U(t,s)$ and $\sqrt{t-s}R(t,s)$ with $\Ujw(t,s), U_w(t,s)$ and $\Rjw(t,s)$ respectively. Writing $\rho = t -s$, from \eqref{2.3} we  have
\begin{align} \label{5.33}
&  i\left\{\Ujw(t,s)f - U_w(t,s)f\right\} = \rho\int_0^1\sqrt{\theta}\Rjw(s+\theta\rho,s)fd\theta + \frac{\rho^2}{i}\int_0^1\theta H_w(s+\theta\rho)d\theta \notag \\
& \quad \cdot\int_0^1\bigl\{H_w(s+\theta'\theta\rho)\Ujw(s+\theta'\theta\rho,s)f 
+ \sqrt{\theta'\theta}\Rjw(s+\theta'\theta\rho,s)f\bigr\}d\theta'\notag \\
& \quad - \frac{\rho^2}{i}\int_0^1\theta H_w(s+\theta\rho)d\theta \int_0^1H_w(s+\theta'\theta\rho)U_w(s+\theta'\theta\rho,s)f d\theta'.
\end{align}
Applying \eqref{2.15}, \eqref{5.19} and \eqref{5.23} to \eqref{5.33}, we have
\begin{align*} 
  \Vert\Ujw(\ts;\kappa')f &- U_w(t,s)f \Vert_a  \leq C\Bigl[ \rho\{\rho+\sigma(|t-\kappa'|)\}\Vert f\Vert_{a+4} +  \rho^{2}\Vert f\Vert_{a+4}
\\
& + \rho^2\{\rho+\sigma(|t-\kappa'|)\}\Vert f\Vert_{a+6}\Bigr] 
 \leq C'\rho\{\rho+\sigma(|t-\kappa'|)\}\Vert f\Vert_{a+6},
\end{align*}
which shows \eqref{5.32}.
\end{proof}
	Now we will prove Theorem 3.1.  The inequalities \eqref{5.32} show
\begin{align} \label{5.34}
& \Vert \Ulw(\kappa_j,\tau_j;\kappa'_j)f - U_w(\kappa_j,\tau_j)f\Vert_a 
 \leq 2C_a(\kappa_j - \tau_j)\sigma(|\Delta|)\Vert f \Vert_{a+6}, \notag \\
 & \Vert \Urw(\tau_{j+1},\kappa_j;\kappa'_j)f - U_w(\tau_{j+1},\kappa_j)f\Vert_a 
 \leq 2C_a(\tau_{j+1} - \kappa_j)\sigma(|\Delta|)\Vert f \Vert_{a+6}
\end{align}
for $j= 0,1,\dots,\nu-1$ because $0 \leq \varrho \leq \sigma(\varrho)$ and $\sigma(\varrho)$ is non-decreasing. Let's consider \eqref{3.2}.  Then we can write
\begin{align} \label{5.35}
& U(t,\kappa_{\nu-1})e^{-(t-\tau_{\nu-1})W_{\ds}(\kappa'_{\nu-1})}
U(\kappa_{\nu-1},\kappa_{\nu-2})e^{-(\tau_{\nu-1}-\tau_{\nu-2})W_{\ds}(\kappa'_{\nu-2})} \cdots e^{-(\tau_{2}-\tau_{1})W_{\ds}(\kappa'_{1})}   \notag \\
&\quad \cdot U(\kappa_{1},\kappa_{0})e^{-\tau_{1}W_{\ds}(\kappa'_{0})}U(\kappa_0,0)f
 =  U(t,\kappa_{\nu-1})e^{-(t-\kappa_{\nu-1})W_{\ds}(\kappa'_{\nu-1})} \notag \\
& \quad \cdot e^{-(\kappa_{\nu-1}-\tau_{\nu-1})W_{\ds}(\kappa'_{\nu-1})}
 U(\kappa_{\nu-1},\tau_{\nu-1})\cdot U(\tau_{\nu-1},\kappa_{\nu-2})e^{-(\tau_{\nu-1}-\kappa_{\nu-2})W_{\ds}(\kappa'_{\nu-2})} 
\notag \\
&\quad \cdot e^{-(\kappa_{\nu-2}-\tau_{\nu-2})W_{\ds}(\kappa'_{\nu-2})} \cdots e^{-(\tau_{2}-\kappa_{1})W_{\ds}(\kappa'_{1})} \cdot e^{-(\kappa_{1}-\tau_{1})W_{\ds}(\kappa'_{1})}  U(\kappa_{1},\tau_1) \notag \\
& \quad \cdot U(\tau_{1},\kappa_{0})e^{-(\tau_{1}-\kappa_{0})W_{\ds}(\kappa'_{0})}
\cdot e^{-\kappa_{0}W_{\ds}(\kappa'_{0})} U(\kappa_{0},0)f \notag \\
& =  \Urw(t,\kappa_{\nu-1};\kappa'_{\nu-1})\Ulw(\kappa_{\nu-1},\tau_{\nu-1};\kappa'_{\nu-1})\Urw(\tau_{\nu-1},\kappa_{\nu-2};\kappa'_{\nu-2})
\notag \\
& \quad \cdot  \cdots \Ulw(\kappa_{1},\tau_{1};\kappa'_{1}) \Urw(\tau_{1},\kappa_{0};\kappa'_{0}) \Ulw(\kappa_{0},0;\kappa'_{0})f.
\end{align}
We also write
\begin{align} \label{5.36}
U_{W}(t,0)f =  & \ U_{W}(t,\kappa_{\nu-1})U_{W}(\kappa_{\nu-1},\tau_{\nu-1}) U_{W}(\tau_{\nu-1},\kappa_{\nu-2})\cdots \notag \\
& \quad \cdot U_{W}(\kappa_{1},\tau_{1}) U_{W}(\tau_{1},\kappa_{0})U_{W}(\kappa_{0},0)f
\end{align}
from Theorem 2.1.  Hence, to prove \eqref{3.2}, we have only to prove
\begin{align} \label{5.37}
& \lim_{|\Delta|\to 0} \bigl\Vert \Urw(t,\kappa_{\nu-1};\kappa'_{\nu-1})\Ulw(\kappa_{\nu-1},\tau_{\nu-1};\kappa'_{\nu-1})\Urw(\tau_{\nu-1},\kappa_{\nu-2};\kappa'_{\nu-2})\cdots  \notag \\
&\qquad  \cdot \Urw(\tau_{1},\kappa_{0};\kappa'_{0}) \Ulw(\kappa_{0},0;\kappa'_{0})f - U_{W}(t,\kappa_{\nu-1})
U_{W}(\kappa_{\nu-1},\tau_{\nu-1})\cdot \notag \\
& \qquad \quad \cdot U_{W}(\tau_{\nu-1},\kappa_{\nu-2})\cdots U_{W}(\tau_{1},\kappa_{0})U_{W}(\kappa_{0},0)f
\bigr\Vert_{a}= 0
\end{align}
 uniformly in $t \in [0,T]$.
Using \eqref{5.6}, \eqref{5.19} and \eqref{5.34}, and  following (5.3) - (5.7)  in the proof of Theorem 2.2 of \cite{Ichinose 2023}, we can easily complete the proof of \eqref{5.37}. 
%
\section{Proof of Theorems 3.2 - 3.5}
	   We assume $C_W = 0$ in Assumption 2.D throughout \S 6 as noted in \S 5.
\begin{lem} \label{lem 6.1}
We assume \eqref{2.11} and \eqref{2.12}.  Let $\varrho_0 > 0$ be a constant. Then we have
 \begin{equation}  \label{6.1}
 \Vert \partial_x^{\alpha}e^{-\varrho\Wstx}\Vert_{\bC^l} \leq C_{\alpha} < \infty, \quad 0 \leq \varrho \leq \varrho_0
 \end{equation}
 in $\domain$ for all $\alpha$ with constants $C_{\alpha} \geq 0$ independent of $\varrho \in [0,\varrho_0]$.
\end{lem}
\begin{proof}
When $\alpha = 0$, \eqref{6.1} is clear. Let $|\alpha| = 1$. Applying \eqref{2.11} and \eqref{2.12} to \eqref{5.20}, we have
\begin{align*} 
& \bigl\Vert \partial_{x_j}\,e^{-\varrho\Wstx}\bigr\Vert_{\bC^l} \leq \varrho \int_0^1  \bigl\Vert e^{-(1-\theta)\varrho\Wstx}\bigr\Vert_{\bC^l} 
\cdot \bigl\Vert \partial_{x_j}\,\Wstx \bigr\Vert_{\bC^l}\cdot \bigl\Vert e^{-\theta\varrho\Wstx}\bigr\Vert_{\bC^l} d\theta
\\
& \leq \varrho \int_0^1 e^{-(1-\theta)\varrho w(t,x)}C_{\alpha}(1+w(t,x)) e^{-\theta\varrho w(t,x)}d\theta
  \end{align*}
with constants $C_{\alpha}$ independent of $\varrho$ from the inequality mentioned below, which shows 
 \begin{equation}  \label{6.2}
 \Vert \partial_{x_{j}}\,e^{-\varrho\Wstx}\Vert_{\bC^l} \leq C_{\alpha}\varrho(1+w(t,x)) e^{-\varrho w(t,x)}, \quad 0 \leq \varrho \leq \varrho_0.
 \end{equation}
 Here we used 
\begin{equation*} 
  \Vert e^{-\varrho\Wstx}\Vert_{\bC^l} = \sup_{\df \in \bC^l,|\df|=1}<e^{-\varrho\Wstx}\df,\df>_{\bC^l}\, \leq  e^{-\varrho w(t,x)}
    \end{equation*}
 (cf. (3.20) in \cite{Mizohata}) because we have $e^{-\varrho\Wstx} \leq  e^{-\varrho w(t,x)}$ from \eqref{2.11}, where $<\cdot,\cdot>_{\bC^l}$ denotes the inner product of $\bC^l$.
 Consequently we have
 \begin{align}  \label{6.3}
 &\Vert \partial_{x_{j}}\, e^{-\varrho\Wstx}\Vert_{\bC^l}  \leq C'_{\alpha}(1+\varrho w) e^{-\varrho w}  \notag\\
 &\quad \leq C'_{\alpha}\sup_{a\geq 0}(1+a)e^{-a} < \infty, \quad 0 \leq \varrho \leq \varrho_0,
 \end{align}
 which shows \eqref{6.1} with $|\alpha| = 1$.
 \par
 	From \eqref{5.20} we have
 \begin{align*} 
&\partial_{x_k}\partial_{x_j}\,e^{-\varrho W_{\ds}(t,x)}  = 
-\varrho \int_0^1 \Bigl\{\partial_{x_k}e^{-(1-\theta)\varrho W_{\ds}(t,x)}\cdot \partial_{x_j} W_{\ds}(t,x)\cdot 
e^{-\theta\varrho W_{\ds}(t,x)}
\notag \\
&\quad   + e^{-(1-\theta)\varrho W_{\ds}(t,x)}\partial_{x_k}\partial_{x_j} W_{\ds}(t,x)\cdot 
e^{-\theta\varrho W_{\ds}(t,x)}  
\\
 & \qquad + e^{-(1-\theta)\varrho W_{\ds}(t,x)}\partial_{x_j} W_{\ds}(t,x)\cdot 
\partial_{x_k}e^{-\theta\varrho W_{\ds}(t,x)}\Bigr\} d\theta.
 \end{align*}
Hence, using \eqref{2.12} and \eqref{6.2}, we see
\begin{align*} 
& \bigl\Vert \partial_{x_k}\partial_{x_j}\,e^{-\varrho\Wstx}\bigr\Vert_{\bC^l} 
\leq C \varrho \int_0^1 \Bigl\{(1-\theta) \varrho  (1+w)e^{-(1-\theta)\varrho w} (1+w)e^{-\theta\varrho w}\notag \\
&\quad + e^{-(1-\theta)\varrho w}  (1+w)e^{-\theta\varrho w} +  e^{-(1-\theta)\varrho w}(1+w)\theta\varrho(1+w)e^{-\theta\varrho w}\Bigr\}d\theta,
  \end{align*}
which shows 
\begin{equation*} 
 \bigl\Vert \partial_{x_k}\partial_{x_j}\,e^{-\varrho\Wstx}\bigr\Vert_{\bC^l} 
\leq C  e^{-\varrho w(t,x)}\bigl\{\varrho(1+w(t,x)) + \varrho^{2}(1+w(t,x))^{2} \bigr\}.
  \end{equation*}
Consequently we can prove \eqref{6.1} with $|\alpha| = 2$ as in the proof of \eqref{6.3}.  In the same way we can prove
\begin{equation}  \label{6.4}
 \bigl\Vert \partial_{x}^{\alpha}\,e^{-\varrho\Wstx}\bigr\Vert_{\bC^l}  
\leq C_{\alpha}  e^{-\varrho w(t,x)}\sum_{n=1}^{|\alpha|}
\varrho^{n}(1+w(t,x))^{n} ,
  \end{equation}
which shows \eqref{6.1} in general.
\end{proof}
{\bf Proof of Theorems 3.2 and 3.3.}
We will first prove Theorem 3.3. Let $f \in (B^{a})^{l}\ (\adots)$.  We can see from Lemma 6.1 that $Z_j(x) = \exp\bigl\{-(\tau_{j+1}-\tau_{j})W_{\ds}(\kappa_{j},x)\bigr\}\ (j = 0,1,2,\dots,\nu-1)$ satisfy \eqref{2.19} with $\mathfrak{M}_j = 0$.  Consequently, from Theorem 2.5 we see that there exists 
	$\Bigl<\prod_{j=0}^{\nu-1}\exp\bigl\{-(\tau_{j+1}-\tau_j) 
	W_{\ds}(\kappa_j,q(\kappa_j))\bigr\}\Bigr>f$ in $(B^a)^l$ and we have
\begin{align} \label{6.5}
& \left<\prod_{j=0}^{\nu-1}\exp\bigl\{-(\tau_{j+1}-\tau_j)W_{\ds}(\kappa_j,q(\kappa_j))\bigr\}\right>f 
= U(t,\kappa_{\nu-1})e^{-(t-\tau_{\nu-1})W_{\ds}(\kappa_{\nu-1})} \notag \\
& 
\quad \cdot U(\kappa_{\nu-1},\kappa_{\nu-2})e^{-(\tau_{\nu-1}-\tau_{\nu-2})W_{\ds}(\kappa_{\nu-2})}U(\kappa_{\nu-2},\kappa_{\nu-3}) \cdots \notag \\
& \qquad \cdot
e^{-(\tau_{2}-\tau_{1})W_{\ds}(\kappa_{1})}U(\kappa_{1},\kappa_{0})e^{-\tau_{1}W_{\ds}(\kappa_{0})}U(\kappa_0,0)f
\end{align}
in $(B^{a})^{l}$.
Applying \eqref{3.2} to the right-hand side of \eqref{6.5}, we have \eqref{3.3} from (3) in Theorem 2.3,
which completes the proof of Theorem 3.3.
\par
	Theorem 3.2 can be derived from Theorem 3.3 as  Theorem 2.1 was done from Theorem 2.2 in \cite{Ichinose 2023}.  That is, 
we first determine a potential $(V',A')$ satisfying Assumption 2.C from Assumption 2.A.  Then Theorem 3.3 with this potential $(V',A')$ holds.
In the end, using the gauge transformation \eqref{2.18}, we can complete the proof of Theorem 3.2.
\par

\vspace{0.5cm}
	{\bf Proof of Theorems 3.4 and 3.5.}
	Let $\omega(\varrho)\ (0 \leq \varrho \leq 1)$ be the function in Theorem 3.4.  We can write
\begin{equation}  \label{6.6}
\omega(\varrho) = \varrho \int_0^1\omega'(\theta\varrho)d\theta \equiv \varrho\omega_1(\varrho)
  \end{equation}
from $\omega(0) = 0$.  Then we have
\begin{equation}  \label{6.7}
0 \leq \omega_1(\varrho)\leq \omega'(\varrho), \quad \rho \in [0,1]
  \end{equation}
because of $\omega'(\theta\varrho) \leq \omega'(\varrho)\ (0 \leq \theta \leq 1)$.
We write 
 \begin{align} \label{6.8}
& \Utlw(t,s;\lambda,\kappa')f = e^{-(t-s)\omega_1(\lambda)W_{\ds}(\kappa')} U(t,s)f, \notag \\
 &  \Utrw(t,s;\lambda,\kappa')f = U(t,s)\left( e^{-(t-s)\omega_1(\lambda)W_{\ds}(\kappa')}f\right)
    \end{align}
with constants $\lambda \in [0,1]$ and $\kappa' \in [0,T]$ as in \eqref{5.17}.
\par
	We can easily see the following from Proposition 5.3, where we replace $W_{\ds}(\kappa')$ with $\omega_1(\lambda)W_{\ds}(\kappa')$.
	\begin{pro} \label{pro 6.2}
Suppose the same assumptions as in Theorem 3.4. Then there exist constants $C_a \geq 0$ and $C'_a \geq 0 \ (\adots)$
such that 
 \begin{equation} \label{6.9}
\Vert e^{-(t-s)\omega_1(\lambda)W_{\ds}(\kappa')}f\Vert_a \leq e^{C_a(t-s)}\Vert f \Vert_{a}, 
 \end{equation}
 \begin{equation} \label{6.10}
\Vert \Utjw(t,s;\lambda,\kappa')f\Vert_a \leq e^{C'_a(t-s)}\Vert f \Vert_{a}, \ 0 \leq s \leq t \leq T
 \end{equation} 
 for $J = L, R, \kappa' \in [0,T]$ and $\lambda \in [0,1]$.
	\end{pro}
	\begin{lem} \label{lem 6.3}
Under the assumptions of Theorem 3.4 we have
 \begin{equation}  \label{6.11}
  \bigl[i\partial_{t} - H(t) - H_{\ds}(t,x) \bigr]\Utjw(t,s;\lambda,\kappa')f = \Rtjw(t,s;\lambda,\kappa')f,
 \end{equation}
 \begin{equation}  \label{6.12}
 \Vert\Rtjw(t,s;\lambda,\kappa')f\Vert_{a} \leq C_{a}\{t-s+\omega'(\lambda)\}\Vert f \Vert_{a+4},\ 0 \leq s \leq t \leq T
 \end{equation}
 for $J = L, R, \kappa' \in [0,T], \lambda \in [0,1]$ and $a = 0,1,2,\dots$.	%
	\end{lem}
	\begin{proof} We omit $\lambda$ and $\kappa'$ in $\Utjw(t,s;\lambda,\kappa')$ for the sake of simplicity.
	We have
 \begin{align}  \label{6.13}
 & \bigl[i\partial_{t} - H(t) - H_{\ds}(t) \bigr]\Utrw(t,s)f \notag\\
 & =  \bigl[i\partial_{t} - H(t) - H_{\ds}(t)\bigr]U(t,s)\bigl(e^{-(t-s)\omega_1(\lambda)W_{\ds}(\kappa')}f\bigr) \notag \\
 & = -U(t,s)\Bigl\{i\omega_1(\lambda)W_{\ds}(\kappa')e^{-(t-s)\omega_1(\lambda)W_{\ds}(\kappa')}f\Bigr\}  \notag \\
 & =  -i\omega_1(\lambda)\Utrw(t,s)W_{\ds}(\kappa')f  = \Rtrw(t,s)f,
 \end{align}
 which shows 
 \begin{align*}  
 & \Vert \Rtrw(t,s)f \Vert_{a} \leq C_a\, \omega_1(\lambda)\Vert W_{\ds}(\kappa')f\Vert_a \leq C'_a\,\omega'(\lambda)
 \Vert f\Vert_{a+2} 
  \end{align*}
from \eqref{6.7} and \eqref{6.10}.  This proves  \eqref{6.12} with $J = R$.
\par
	Next we have
 \begin{align}  \label{6.14}
 & \bigl[i\partial_{t} - H(t) - H_{\ds}(t)\bigr]\Utlw(t,s;\lambda,\kappa')f \notag\\
 & =  \bigl[i\partial_{t} - H(t) - H_{\ds}(t)\bigr]e^{-(t-s)\omega_1(\lambda)W_{\ds}(\kappa')}U(t,s)f\notag \\
 & =  -i\omega_1(\lambda)W_{\ds}(\kappa')\Utlw(t,s)f - \bigl[H(t),e^{-(t-s)\omega_1(\lambda)W_{\ds}(\kappa')}\bigr]U(t,s)f \notag\\
 &\quad - \bigl[H_{\ds}(t),e^{-(t-s)\omega_1(\lambda)W_{\ds}(\kappa')}\bigr]U(t,s)f  = \Rtlw(t,s)f \equiv -\sum_{j=1}^3 I_j(t,s)f
 \end{align}
 as in the proof of \eqref{5.26}.
 We obtain 
 \begin{equation*}  
 \Vert I_{1}(t,s)f\Vert_{a} \leq C_a\,\omega_1(\lambda)\Vert \Utlw(t,s)f \Vert_{a+2} \leq C_a'\,\omega'(\lambda)\Vert f \Vert_{a+2}
 \end{equation*}
 from \eqref{6.7} and \eqref{6.10}.  We can prove
 \begin{equation*}  
 \Vert I_{2}(t,s)f\Vert_{a} \leq C_a(t-s)\Vert f \Vert_{a+2},\quad \Vert I_{3}(t,s)f\Vert_{a} \leq C_a(t-s)\Vert f \Vert_{a+4}
 \end{equation*}
 as in the proof of \eqref{5.29} and \eqref{5.31} respectively. Hence we obtain  \eqref{6.12} for $J = L$ from
 \eqref{6.14}.
	\end{proof}
	\begin{pro} \label{pro 6.4}
Under the assumptions of Theorem 3.4 we have
 \begin{equation}  \label{6.15}
 \Vert \Utjw(t,s;\lambda,\kappa')f - U(t,s)f \Vert_{a} \leq C_a\rho\{\rho+\omega'(\lambda)\}\Vert f \Vert_{a+6}, \quad 0 \leq s \leq t \leq T
 \end{equation}
 with $\rho = t -s$ for $J = L, R, \kappa' \in [0,T], \lambda \in [0,1]$ and $\adots$.
	\end{pro}
	\begin{proof}
We omit  $\kappa' \in [0,T]$ and $\lambda \in [0,1]$ for simplicity. We write
\[
\widetilde{H}(t) = H(t) + H_{\ds}(t)
\]
for a while. Then, using Lemma 6.3, we can prove \eqref{6.16} below  as in the proof of \eqref{4.14} in \cite{Ichinose 2023}, where we replace $\mathcal{C}(t,s), H_w(t)$ and $\sqrt{t-s}R(t,s)$ with $\Utjw(t,s), \widetilde{H}(t)$ and $\Rtjw(t,s)$ respectively. We have
\begin{align} \label{6.16}
&  i\left\{\Utjw(t,s)f - U(t,s)f\right\} = \rho\int_0^1\sqrt{\theta}\Rtjw(s+\theta\rho,s)fd\theta + \frac{\rho^2}{i}\int_0^1\theta \Htilde(s+\theta\rho)d\theta \notag \\
& \quad \cdot\int_0^1\Bigl\{\Htilde(s+\theta'\theta\rho)\Utjw(s+\theta'\theta\rho,s)f 
+ \sqrt{\theta'\theta}\Rtjw(s+\theta'\theta\rho,s)f\Bigr\}d\theta'\notag \\
& \quad - \frac{\rho^2}{i}\int_0^1\theta \Htilde(s+\theta\rho)d\theta \int_0^1\Htilde(s+\theta'\theta\rho)U(s+\theta'\theta\rho,s)f d\theta'.
\end{align}
Hence, applying \eqref{2.15} with $W_{\ds} = 0$, \eqref{6.10}  and \eqref{6.12} to \eqref{6.16}, we obtain
\begin{align*} 
 & \Vert\Utjw(\ts)f - U(t,s)f \Vert_a  \leq C_a \Bigl[\rho\,\{\rho+\omega'(\lambda)\}\Vert f\Vert_{a+4} +  \rho^{2}\Vert f\Vert_{a+4}
+  \\
&\hspace{1cm} + \rho^2\,\{\rho+\omega'(\lambda)\}\Vert f\Vert_{a+6}\Bigr]\leq C'_a\rho\{\rho + \omega'(\lambda)\}\Vert f\Vert_{a+6},
\end{align*}
which shows \eqref{6.15}.
	\end{proof}
	{\bf Proof of Theorem 3.4.}  Let 
 \begin{equation}  \label{6.17}
 \rho_j = \tau_{j+1}-\tau_j\quad (j= 0,1,2,\dots,\nu-1).
 \end{equation}
 Proposition 6.4 shows 
 \begin{align} \label{6.18}
 &\Vert \Utlw(\kappa_{j},\tau_{j};\rho_{j},\kappa'_{j})f - U(\kappa_{j},\tau_{j})f \Vert_{a} \leq 2C_a\rho_j\,\omega'(|\Delta|)\Vert f \Vert_{a+6}, \notag \\
 & \Vert \Utrw(\tau_{j+1},\kappa_{j};\rho_{j},\kappa'_{j})f - U(\kappa_{j},\tau_{j})f \Vert_{a} \leq 2C_a\rho_j\,\omega'(|\Delta|)\Vert f \Vert_{a+6}
 \end{align}
 from $0 \leq \rho_j \leq \omega'(\rho_j) \leq \omega'(|\Delta|)$ for $j= 0,1,\dots,\nu-1$ and $a = 0,1,2,\dots.$
 \par
 	Using \eqref{6.6}, as in the proof of \eqref{5.35} we have
\begin{align} \label{6.19}
& U(t,\kappa_{\nu-1})e^{-\omega(t-\tau_{\nu-1})W_{\ds}(\kappa'_{\nu-1})}
U(\kappa_{\nu-1},\kappa_{\nu-2})e^{-\omega(\tau_{\nu-1}-\tau_{\nu-2})W_{\ds}(\kappa'_{\nu-2})}\cdots 
\notag \\
&\quad \cdot e^{-\omega(\tau_{2}-\tau_{1})W_{\ds}(\kappa'_{1})} U(\kappa_{1},\kappa_{0})e^{-\omega(\tau_{1})W_{\ds}(\kappa'_{0})}U(\kappa_0,0)f
 \notag \\
 & = U(t,\kappa_{\nu-1})e^{-(t-\tau_{\nu-1})\omega_1(\rho_{\nu-1})W_{\ds}(\kappa'_{\nu-1})}
U(\kappa_{\nu-1},\kappa_{\nu-2})e^{-(\tau_{\nu-1}-\tau_{\nu-2})\omega_1(\rho_{\nu-2})W_{\ds}(\kappa'_{\nu-2})} \notag \\
& \quad \cdot  \cdots e^{-(\tau_{2}-\tau_{1})\omega_1(\rho_{1})W_{\ds}(\kappa'_{1})}U(\kappa_{1},\kappa_{0})e^{-\tau_{1}\omega_1(\rho_{0})W_{\ds}(\kappa'_{0})}U(\kappa_0,0)f
 \notag \\
& =  U(t,\kappa_{\nu-1})e^{-(t-\kappa_{\nu-1})\omega_1(\rho_{\nu-1})W_{\ds}(\kappa'_{\nu-1})}\cdot e^{-(\kappa_{\nu-1}-\tau_{\nu-1})\omega_1(\rho_{\nu-1})W_{\ds}(\kappa'_{\nu-1})} U(\kappa_{\nu-1},\tau_{\nu-1})
\notag \\
&\quad \cdot U(\tau_{\nu-1},\kappa_{\nu-2})e^{-(\tau_{\nu-1}-\kappa_{\nu-2})\omega_1(\rho_{\nu-2})W_{\ds}(\kappa'_{\nu-2})}
 \cdots  e^{-(\kappa_{1}-\tau_{1})\omega_1(\rho_{1})W_{\ds}(\kappa'_{1})} U(\kappa_{1},\tau_{1})
   \notag \\
& \qquad \cdot U(\tau_{1},\kappa_{0})e^{-(\tau_{1}-\kappa_{0})\omega_1(\rho_{0})W_{\ds}(\kappa'_{0})}\cdot e^{-\kappa_{0}\omega_1(\rho_{0})W_{\ds}(\kappa'_{0})} U(\kappa_{0},0)f \notag \\
&
=  \Utrw(t,\kappa_{\nu-1};\rho_{\nu-1},\kappa'_{\nu-1})\Utlw(\kappa_{\nu-1},\tau_{\nu-1};\rho_{\nu-1},\kappa'_{\nu-1})
\Utrw(\tau_{\nu-1},\kappa_{\nu-2};\rho_{\nu-2},\kappa'_{\nu-2})
 \notag \\
 & 
 \quad \cdot \cdots \Utlw(\kappa_{1},\tau_{1};\rho_{1},\kappa'_{1})\Utrw(\tau_{1},\kappa_{0};\rho_{0},\kappa'_{0}) \Utlw(\kappa_{0},0;\rho_{0},\kappa'_{0})f.
\end{align}
Hence, using \eqref{5.6} with $W_{\ds}(t,x) = 0$, \eqref{6.10}, \eqref{6.18} and \eqref{6.19}, we can prove Theorem 3.4 as in the proof of Theorem 3.1.
\par
\vspace{0.5cm}
{\bf Proof of Theorem 3.5.}  It holds from Lemma 6.1 that we  have
 \begin{equation}  \label{6.20}
 \Vert \partial_x^{\alpha}e^{-\omega(\varrho)\Wstx}\Vert_{\bC^{l}} \leq C_{\alpha} < \infty, \quad 0 \leq \varrho \leq 1
 \end{equation}
in $\domain$ for all $\alpha$ with constants $C_{\alpha}$.
Hence, using Theorem 2.5, for $f \in (B^a)^l\ (\adots)$ we have
\begin{align} \label{6.21}
& \left<\prod_{j=0}^{\nu-1}\exp\bigl\{-\omega(\tau_{j+1}-\tau_j)W_{\ds}(\kappa_j,q(\kappa_j))\bigr\}\right>f 
= U(t,\kappa_{\nu-1})e^{-\omega(t-\tau_{\nu-1})W_{\ds}(\kappa_{\nu-1})} \notag \\
& 
\quad \cdot U(\kappa_{\nu-1},\kappa_{\nu-2})e^{-\omega(\tau_{\nu-1}-\tau_{\nu-2})W_{\ds}(\kappa_{\nu-2})}U(\kappa_{\nu-2},\kappa_{\nu-3}) \cdots e^{-\omega(\tau_{2}-\tau_{1})W_{\ds}(\kappa_{1})}U(\kappa_{1},\kappa_{0})
\notag \\
& \qquad \cdot
e^{-\omega(\tau_{1})W_{\ds}(\kappa_{0})}U(\kappa_0,0)f
\end{align}
in $(B^{a})^{l}$ as in the proof of \eqref{6.5}.  Applying \eqref{3.5} to the right-hand side of \eqref{6.21}, we have \eqref{3.6} from (3) in Theorem 2.3.
Thus we have proved Theorem 3.5.
\appendix
\section{ An extension of Zworski's theorem}
We will prove the following theorem which is an extension of Theorem 13.13 in \cite{Z}.
\begin{thm} \label{thm A.1}  Let $p_{ij}(x,\xi,x') \ (i,j = 1,2,\dots,l)$ be  functions satisfying \eqref{5.1} with $M=0$ for all
$\alpha, \beta$ and $\beta'$, and define the pseudo-differential operators $p_{ij}(x,\h D_x,x')$ with symbol $p_{ij}(x,\h\xi,x')$ by \eqref{5.2}, where $\h > 0$ is a constant.  We write
\[
p(x,\h D_x,x') = \Bigl(p_{ij}(x,\h D_x,x'); i,j =1,2,\dots,l\Bigr).
\]
  Then we have
 \begin{equation} \label{A.1}
\Vert p(x,\mathfrak{h}  D_x,x')\Vert_{(L^2)^l\to (L^2)^l} \leq \sup_{x,\xi}\Vert p(x,\xi,x)\Vert_{\bC^l} + O(\mathfrak{h} ).
 \end{equation}
\end{thm}
\begin{rem} \label{rem A.1} As a matter of fact, under the assumptions of Theorem A.1 we can  prove 
 \begin{equation} \label{A.2}
\Vert p(x,\mathfrak{h}  D_x,x')\Vert_{(L^2)^l\to (L^2)^l} = \sup_{x,\xi}\Vert p(x,\xi,x)\Vert_{\bC^l} + O(\mathfrak{h} ).
\end{equation}
We don't give the proof of the opposite inequality of \eqref{A.1} in the present paper, because 
 we don't use it. Its proof  will be published in 
\cite{Ichinose 2024} elsewhere.
\end{rem}
\begin{proof}  We write
 \begin{equation} \label{A.3}
b(x,\xi) = \Bigl(b_{ij}(x,\xi);i,j=1,2,\dots,l\Bigr) : =  \Bigl(p_{ij}(x,\xi,x);i,j=1,2,\dots,l\Bigr).
 \end{equation}
Let $b_{ij}^w(x,\h D_x)\ (i,j = 1,2,\dots,l)$ be the Weyl operators defined by
 \begin{equation} \label{A.4}
 b_{ij}^w(x,\h D_x)f = 
\int e^{ix\cdot \xi}\ \dbar\xi \int e^{-iy\cdot \xi}b_{ij}((x+y)/2,\h\xi)f(y)dy
 \end{equation}
for $f \in \Sspace(\bR^d)$ (cf. (4.1.1) in \cite{Z}).  Then from Theorem 13.13 in \cite{Z} we have
 \begin{equation*} 
\Vert b_{ij}^w(x,\h D_x)\Vert_{L^2\to L^2} = \sup_{x,\xi} |b_{ij}(x,\xi)| + O(\mathfrak{h}) = \sup_{x,\xi} |p_{ij}(x,\xi,x)| + O(\mathfrak{h})
\end{equation*}
 and also have
 \begin{equation*} 
\Vert p_{ij}(x,\h D_x,x')\Vert_{L^2\to L^2} = \sup_{x,\xi} |p_{ij}(x,\xi,x)| + O(\mathfrak{h})
\end{equation*}
(cf. (A.12) in \cite{Ichinose 2023}), which show
\par
 \begin{equation} \label{A.5}
\Vert b_{ij}^w(x,\h D_x)\Vert_{L^2\to L^2} = \Vert p_{ij}(x,\h D_x,x')\Vert_{L^2\to L^2}+ O(\mathfrak{h} ).
 \end{equation}
  Hence we have only  to prove
 \begin{equation} \label{A.6}
\Vert b^w(x,\mathfrak{h}  D_x)\Vert_{(L^2)^l\to (L^2)^l} \leq \sup_{x,\xi}\Vert b(x,\xi)\Vert_{\bC^l} + O(\mathfrak{h} )
 \end{equation}
for the proof of \eqref{A.1},
where $b^w(x,\mathfrak{h}  D_x) =  \Bigl(b_{ij}^w(x,\h D_x);i,j= 1,2,\dots,l\Bigr)$.  The inequality \eqref{A.6} is essentially proved as in 
 the proof of Theorem 13.13 in \cite{Z}.  For this reason, we give only a rough sketch of its proof, using the notations and the results in \cite{Z} without detailed explanation.
\par
	For a multi-index $\alpha = (\alpha_1,\dots,\alpha_d)$ and $z_j = x_j + iy_j \in \bC\ (x_j,y_j \in \bR, 
	j= 1,2,\dots,d)$ we write $\overline{z}_j = x_j - iy_j, \rittaire z_j = x_j, \rittaiim z_j = y_j,
	\partial_{z_j}  = \partial/\partial z_j = (\partial_{x_j} - i\partial_{y_j})/2, 
	\partial_{\overline{z}_j} = \partial/\partial \overline{z}_j
	= (\partial_{x_j} + i\partial_{y_j})/2, z = (z_1,\dots,z_d) \in \bC^d, \overline{z} = (\overline{z}_1,\dots,\overline{z}_d),\partial_z^{\alpha} =  \partial_{z_1}^{\alpha_1}\cdots  \partial_{z_d}^{\alpha_d}, \partial_{\overline{z}}^{\alpha} =  \partial_{\overline{z}_1}^{\alpha_1}\cdots  \partial_{\overline{z}_d}^{\alpha_d}$
	  and $z^2 = \sum_{j=1}^d z_j^2$. 
	\par
	We set 
 \begin{equation} \label{A.7}
\Phi(z) = \frac{1}{2}(\rittaiim z)^2 =  \frac{1}{2}y^2 =  -\frac{1}{8}(z - \zbar)^2
 \end{equation}
and introduce the Hilbert space of $L^2$ weighted functions on $\bC^d$
 \begin{equation} \label{A.8}
L^2_{\Phi}(\bC^d) = \biggl\{f:\bC^d \to \bC; \Vert f \Vert_{L^2_{\Phi}} := \biggl(\int_{\bC^d}|f(z)|^2 e^{-2\Phi(z)/\h}dm(z)\biggr)^{1/2}
 < \infty \biggr\},
 \end{equation}
where $dm(z) = dx_1dy_1\cdots dx_ddy_d$ is the Lebsgue measure on $\bC^d$ (cf. (13.1.1) in \cite{Z}). We define 
 \begin{equation} \label{A.9}
H_{\Phi}(\bC^d)= \{f \in L^2_{\Phi}(\bC^d);\partial_{\zbar_j}f = 0, j = 1,2,\dots,d\}
 \end{equation}
consisting of all holomorphic functions in $L^2_{\Phi}(\bC^d)$ in the sense of distribution. We can easily see that $H_{\Phi}(\bC^d)$ is a closed subspace of $L^2_{\Phi}(\bC^d)$.
We note that if $\partial_{\zbar_j}f = 0\ (j = 1,2,\dots,d)$ for $f \in L^2_{\Phi}$, then $f(z) \in C^{\infty}(\bC^d)$ 
because of $f \in L^2_{\text{loc}}(\bC^d)$ and $\sum_{j=1}^d(\partial_{x_j}^2 + \partial_{y_j}^2)f 
= 4\sum_{j=1}^d \partial_{z_j}\partial_{\zbar_j}f = 0$ (cf. Theorem 3.22 in \cite{Mizohata}), which means that 
$f(z) \in H_{\Phi}(\bC^d)$ is a holomorphic function in the sense of complex function theory.  We define a real linear subspace 
in $\bC^{2d}$
 \begin{equation} \label{A.10}
\Lambda_{\Phi} = \left\{\left(z,\frac{2}{i}\frac{\partial \Phi}{\partial z}(z)\right) \in \bC^{2d}; z \in \bC^d\right\}
= \bigl\{(z,-\rittaiim z) \in \bC^{2d}; z \in \bC^d \bigr\},
 \end{equation}
which we identify with $\bC^d$ by using the map: $\bC^d \ni z \to (z, -\rittaiim z) \in \Lambda_{\Phi}$ from now on (cf. Theorem 13.5 in \cite{Z}). Let 
$a(z) \in S(\Lambda_{\Phi}) = S(\bC^d)$, i.e. $a(z)$  a function on $\Lambda_{\Phi}$ satisfying 
 \begin{equation*} 
|\partial_{z}^{\alpha}\partial_{\zbar}^{\beta}a(z)| \leq C_{\alpha\beta} < \infty, \ z \in \bC^d
 \end{equation*}
for all $\alpha$ and $\beta$. Then we define the Weyl quantization of $a(z)$ by (13.4.5) in \cite{Z}, which is written as 
 \begin{align} \label{A.11}
&a^w_{\Phi}(z,\h D_z)u(z) = 
\frac{1}{(2\pi\h)^d}
\cdot 2^d\det \frac{\partial^2\Phi}{\partial z \partial \zbar}  
\int_{\bC^d} a\bigl((z+w)/2\bigr)\notag \\
&\qquad  \times \left\{\exp \left(2\h^{-1}(z-w)\cdot \partial_z\Phi((z+w)/2)\right)\right\}u(w)dm(w)
 \end{align}
for $u \in \Sspace(\bC^d)$ 
by using (13.2.1) in \cite{Z} (cf. the proof of Theorem 13.8 in \cite{Z}), where $z\cdot z'= \sum_{j=1}^d z_jz'_j.$
\par
	We set
 \begin{equation} \label{A.12}
\varphi(z,x) = \frac{i}{2}(z-x)^2
 \end{equation}
for $z$ and $x$ in $\bC^d$ (cf. (13.12) in \cite{Z}).  Then we define the FBI transform  on $\Sspace(\bR^d)$ associated to $\varphi$ by
 \begin{equation} \label{A.13}
(\mathcal{T}_{\varphi}u)(z) = c_{\varphi}\h^{-3d/4}\int_{\bR^d}e^{i\h^{-1}\varphi(z,x)}u(x)dx
 \end{equation}
 with an appropriate constant $c_{\varphi}$ (cf. (13.3.3) in \cite{Z}) and a complex linear symplectic 
 transformation by $\kappa_{\varphi}: \bC^{2d} \ni (x,-\partial_x\varphi(z,x)) \to (z,\partial_z\varphi(z,x)) \in  \bC^{2d}$, i.e.
 \begin{equation} \label{A.14}
\kappa_{\varphi}: \bC^{2d} \ni (x,\xi) \to (x-i\xi,\xi) \in  \bC^{2d}
 \end{equation}
  (cf. (13.3.7) in \cite{Z}).
 \par
	Now we are ready to proceed with the proof of Theorem A.1.  We see from \eqref{A.10} and \eqref{A.14}  that $\kappa_{\varphi}: \bR^{2d} \ni (x,\xi) \to (x-i\xi,\xi) \in  \Lambda_{\Phi}$ is bijection.  For $b(x,\xi)$ defined by \eqref{A.3} we set
 \begin{equation} \label{A.15}
a = (\kappa_{\varphi}^{-1})^*b \in M_l(\bC),
 \end{equation}
 i.e. $a(z) = b(x,-\xi)$ with $z = x + i\xi \in \bC^d$ on $\Lambda_{\Phi}$ because we identify $(x-i\xi,\xi) \in \Lambda_{\Phi}$
  with $x-i\xi \in \bC^d$.
Then we have
 \begin{equation} \label{A.16}
b^w(x,\h D_x) = \Tphi^{\,\dag}a^w_{\Phi}(z,\h D_z)\Tphi
 \end{equation}
on $\Sspace(\bR^d)$ from Theorem 13.9 in \cite{Z}, noting that each component of $a(z) \in M_l(\bC)$ belongs to $S(\Lambda_{\Phi})$,
 where $\Tphi^{\,\dag}$ means the adjoint operator of $\Tphi$ (cf. Theorem 13.7 in \cite{Z}).  Consequently, we have
 \begin{equation} \label{A.17}
\Vert b^w(x,\h D_x)f\Vert_{(L^2)^l} = \Vert \Tphi^{\,\dag}a^w_{\Phi}(z,\h D_z)\Tphi f\Vert_{(L^2)^l} = \Vert a^w_{\Phi}(z,\h D_z)\Tphi f\Vert_{H_{\Phi}^{l}}
 \end{equation}
because we see $a^w_{\Phi}(z,\h D_z)\Tphi f \in H_{\Phi}(\bC^d)^l$ for $f  \in L^2(\bR^d)^l$ from Theorems 13.7 
and 13.8 in \cite{Z} and that
  $\Tphi^{\,\dag}: H_{\Phi}(\bC^d) \to L^2(\bR^d)$
 is unitary. 
   Hence we obtain
 \begin{align} \label{A.18}
&\Vert b^w(x,\h D_x)\Vert_{(L^2)^l \to (L^2)^l} = \sup \bigl\{\Vert b^w(x,\h D_x)f\Vert_{ (L^2)^l}; \Vert f \Vert_{\Ll}= 1\bigr\} \notag \\
& = \sup \bigl\{\Vert a^w_{\Phi}(z,\h D_z)\Tphi f\Vert_{H_{\Phi}^l}; \Vert f \Vert_{\Ll}= 1\bigr\} = 
 \Vert a^w_{\Phi}(z,\h D_z)\Vert_{H_{\Phi}^l\to H_{\Phi}^l}
\end{align}
because $\Tphi: L^2(\bR^d) \to H_{\Phi}(\bC^d)$ is unitary from Theorem 13.7 in \cite{Z}.
\par
	Let $\Pi_{\Phi}: L^2_{\Phi}(\bC^d) \to H_{\Phi}(\bC^d)$ be the orthogonal projector and $M_a: L^2_{\Phi}(\bC^d)^l \ni u(z) \to 
a(z)u(z) \in  L^2_{\Phi}(\bC^d)^l$ the multiplication operator. Then we have
 \begin{equation} \label{A.19}
\Vert \Pi_{\Phi}M_a\Pi_{\Phi} \Vert_{H_{\Phi}^l\to H_{\Phi}^l} = \Vert a^w_{\Phi}(z,\h D_z)\Vert_{H_{\Phi}^l\to H_{\Phi}^l}
+ O(\h)
 \end{equation}
from Theorem 13.10 in \cite{Z}.
Let $g \in H_{\Phi}(\bC^d)^l$.  Then, noting \eqref{A.8} and that $\Pi_{\Phi}: L^2_{\Phi}(\bC^d)\rightarrow H_{\Phi}(\bC^d)$
is the orthogonal projector, we have
 \begin{align} \label{A.20}
& \Vert \Pi_{\Phi}M_a\Pi_{\Phi}g \Vert_{H_{\Phi}^l}^2 = \Vert \Pi_{\Phi}M_ag \Vert_{H_{\Phi}^l}^2 
\leq \Vert M_ag \Vert_{(L^2_{\Phi})^l}^2 \notag \\
&  = \int_{\bC^d}dm(z)\,\bigl<a(z)g(z),a(z)g(z)\bigr>_{\bC^l}\,e^{-2\h^{-1}\Phi(z)} \notag \\
& \leq \int_{\bC^d}dm(z)\,\Vert a(z)\Vert_{\bC^l}^2|g(z)|^2e^{-2\h^{-1}\Phi(z)} \notag \\
& \leq \sup_{z}\Vert a(z) \Vert_{\bC^d}^2 \int_{\bC^d}dm(z)|g(z)|^2\,e^{-2\h^{-1}\Phi(z)} \notag \\
& = \sup_{z}\Vert a(z) \Vert_{\bC^l}^2\cdot \Vert g \Vert_{H_{\Phi}^l}^2,
\end{align}
where $<\cdot,\cdot>_{\bC^l}$ denoted the inner product of $\bC^l$ as in \S 6.
The expression \eqref{A.20} shows 
 \begin{equation} \label{A.21}
\Vert \Pi_{\Phi}M_a\Pi_{\Phi} \Vert_{H_{\Phi}^l\to H_{\Phi}^l}\leq  \sup_{z}\Vert a(z) \Vert_{\bC^d}.
 \end{equation}
Therefore, we obtain \eqref{A.6} from \eqref{A.15}, \eqref{A.18} and \eqref{A.19}, which completes the proof.
\end{proof}
\section{Proof of Theorems 2.2 - 2.5}
	In this section we will always assume $C_W = 0$ in Assumption 2.D as noted in \S 5. Let $M \geq 0$ be an integer and  $p(x,v) \in \Cspace(\bR^{2d})$ a function satisfying
\begin{equation} \label{B.1}
|\partial_{v}^{\alpha}\partial_{x}^{\beta}p(x,v)| \leq C_{\alpha\beta}<x;v>^M, \ (x,v) \in \bR^{2d}
\end{equation}
for all $\alpha$ and $\beta$, where $<x;v> = \sqrt{1 + |x|^2 + |v|^2}$. Let $S(t,s;q)$ be the classical action defined by \eqref{1.3} and 
\begin{equation}  \label{B.2}
  q^{t,s}_{x,y}(\theta) = y + \frac{\theta-s}{t-s}(x-y) = x- \frac{t-\theta}{t-s}(x-y) , \quad s\leq \theta \leq t.
\end{equation}
For $f \in \Czerospace$ we define
\begin{equation}  \label{B.3}
P(t,s)f =
        \begin{cases}
            \begin{split}
              & \sqrt{m/(2\pi i\rho)}^{\ d}
                  \int \bigl(\exp iS(t,s; q^{t,s}_{x,y})\bigr) \\
         &\hspace{2cm} \times    p(x,(x-y)/\sqrt{\rho})f(y)dy,\quad  s < t,
                      \end{split}
                  \\
\begin{split}
        & \sqrt{m/(2\pi i)}^{\ d}
    \text{Os}-\int (\exp im|v|^2/2)\\
             &\hspace{2cm}\times p(x,v)dv\, f(x), \quad s = t.
             \end{split}
        \end{cases}
\end{equation}
Through this section  we  write $\rho = t-s$.
The proposition mentioned below follows from Theorem 4.4 in \cite{Ichinose 1999} and Proposition 3.6 in \cite{Ichinose 2023}.
\begin{pro} \label{pro B.1}  Suppose either Assumption 2.A or 2.B.  In addition, we suppose \eqref{2.9} and that there exists an integer
$M_1 \geq 1$ satisfying
\begin{equation} \label{B.4} 
|\partial_{x}^{\alpha}V(t,x)| \leq C_{\alpha}<x>^{M_1}, \ |\alpha| \geq 1.
\end{equation}
Then there exists a constant $\rho^* > 0$ such that
\begin{equation} \label{B.5} 
\Vert P(t,s)f\Vert_a \leq 
     C_a \Vert f\Vert_{M+aM_1}, \quad  0 \leq t - s \leq \rho^*
\end{equation}
for $a = 0,1,2,\dots$ and all $f \in B^{M+aM_1}$ with constants $C_a \geq 0$.
\end{pro}
\begin{rem} \label{rem B.1} 
The constant $\rho^* >0$ in Proposition B.1 is determined in Lemma 3.3 of \cite{Ichinose 2023}, which is equal to $\rho^*$ in Theorem 2.3 of the present paper.
\end{rem}
\begin{lem} \label{lem B.2}
Suppose \eqref{2.11}, \eqref{2.12} and \eqref{2.17}.  Let $\mathcal{F}_w(t,s;q)$ be the solution to \eqref{2.4} with $\mathcal{U}(s) = I$.  Then the following (1) and (2) hold  for
$0 \leq s \leq s' \leq t' \leq t \leq T$: (1) Let $\{a_n\}_{n=0}^{\infty}$ be an arbitrary sequence satisfying 
$1 = a_0 > a_1> \dots > a_n>\dots > 0.$  Then we have
\begin{align} \label{B.6} 
& \Vert \partial_x^{\alpha}\partial_y^{\beta}\Fw(t',s';\qts)\Vert_{\bC^l}  \notag \\
& \leq C_n \exp \Bigl(-a_n\int_{s'}^{t'}w\bigl(\theta,\qts(\theta)\bigr)d\theta\Bigr), \ |\alpha + \beta| = n
\end{align}
for $n= 0,1,2,\dots$ with constants $C_n \geq 0$, where $C_0 = 1$. (2) Moreover, we assume that 
there exists an 
$M' \geq 0$ satisfying
\begin{equation} \label{B.7} 
\Vert\partial_{x}^{\alpha}\Wstx\|_{\bC^l} \leq C_{\alpha}<x>^{M'}
\end{equation}
in $\domain$ for all $\alpha$.  Then we have
\begin{equation} \label{B.8} 
 \Vert \partial_x^{\alpha}\partial_y^{\beta}\Fw(t',s';\qts)\Vert_{\bC^l} \leq C_{\alpha\beta}(t'-s')<x;y>^{M'} , 
 \ |\alpha + \beta| \geq 1.
\end{equation}
\end{lem}
\begin{proof}  (1) Let $\df \in \bC^l$.  Then, settin $\mathcal{U}(\theta) = \Fw(\theta,s;q)$, we have
\begin{align*} 
& \frac{d}{d\theta}\bigl<\mathcal{U}(\theta)\df,\mathcal{U}(\theta)\df\bigr>_{\bC^{l}} = 2\rittaire \bigl<\frac{d}{d\theta}\,\mathcal{U}(\theta)\df,\mathcal{U}(\theta)\df\bigr> _{\bC^{l}} \\
& = -2\rittaire \bigl<W_{\ds}(\theta,q(\theta))\,\mathcal{U}(\theta)\df,\mathcal{U}(\theta)\df\bigr> _{\bC^{l}}
\leq -2w(\theta,q(\theta))\bigl<\mathcal{U}(\theta)\df,\mathcal{U}(\theta)\df\bigr> _{\bC^{l}}
\end{align*}
from \eqref{2.4} and \eqref{2.11}, which shows 
\begin{equation*} 
 \bigl<\mathcal{U}(t)\df,\mathcal{U}(t)\df\bigr>_{\bC^{l}} \leq e^{-2\int_s^tw(\theta,q(\theta))d\theta}|\df|^2.
\end{equation*}
Hence we obtain 
\begin{equation} \label{B.9} 
 \|\Fw(t,s;q)\|_{\bC^{l}} \leq e^{-\int_s^tw(\theta,q(\theta))d\theta},\ 0 \leq s \leq t \leq T,
\end{equation}
which shows \eqref{B.6} with $|\alpha + \beta| = 0$.
\par
	From \eqref{2.4} we can easily prove
\begin{align} \label{B.10}
   & \frac{\partial}{\partial x_k}\mathcal{F}_w (t',s';\qts)  = -\int_{s'}^{t'} \mathcal{F}_w (t',\theta;\qts)
   \left[\frac{\partial}{\partial x_k}\Bigl\{
iH_{\ds}(\theta,\qts(\theta)) +  W_{\ds}(\theta,\qts(\theta)) \Bigr\}\right]
\notag \\
& \qquad 
\times \mathcal{F}_w (\theta,s';\qts)d\theta
\end{align}
(cf. (3.3) in \cite{Ichinose 2007}).  Consequently, using \eqref{2.12}, \eqref{2.17}, \eqref{B.2} and \eqref{B.9}, we see
\begin{align} \label{B.11}
   & \|\partial_{x_k}\,\mathcal{F}_w (t',s';\qts)\|_{\bC^{l}}  \leq  \int_{s'}^{t'} \|\mathcal{F}_w (t',\theta;\qts)\|_{\bC^{l}} \cdot
   \Bigl\|\partial_{x_k}\,\Bigl\{iH_{\ds}(\theta,\qts(\theta))  \notag \\ 
& \quad +  W_{\ds}(\theta,\qts(\theta)) \Bigr\}\Bigr\|_{\bC^{l}}  \cdot \|\mathcal{F}_w (\theta,s';\qts)\|_{\bC^{l}} d\theta \notag \\ 
& \leq e^{-\int_{s'}^{t'}w(\theta,\qts(\theta))d\theta}\int_{s'}^{t'}\bigl(C + \|(\partial_{x_k} W_{\ds})(\theta,\qts(\theta)\bigr)\|_{\bC^{l}})d\theta
\notag \\
& \leq C'e^{-\int_{s'}^{t'}w(\theta,\qts(\theta))d\theta}\int_{s'}^{t'}\bigl(1 + \sum_{i,j=1}^l |(\partial_{x_k} w_{ij})(\theta,\qts(\theta))|\bigr)d\theta \notag \\
& \leq C''e^{-\int_{s'}^{t'}w(\theta,\qts(\theta))d\theta}\int_{s'}^{t'}\bigl(1 + w(\theta,\qts(\theta))\bigr)d\theta.
\end{align}
Applying $\sup_{b \geq 0}be^{-\epsilon b}\leq C_{\epsilon} < \infty$ for arbitrary constant $\epsilon > 0$ to \eqref{B.11}
as $b = \displaystyle{\int_{s'}^{t'} }w(\theta,\qts(\theta))d\theta$ and $\epsilon = 1 - a_1 > 0$,  we obtain \eqref{B.6} with $|\alpha| = 1$ and $\beta = 0$.
 In the same way we obtain \eqref{B.6} with $|\alpha + \beta| = 1$.  Using \eqref{B.10}, we can prove \eqref{B.6} in general  by the same argument inductively.
\par
	(2)  
	Applying \eqref{2.17}, \eqref{B.7} and \eqref{B.9} to \eqref{B.10}, 
	we can prove \eqref{B.8} with $|\alpha|=1$ and $\beta = 0$. 
	In the same way we can prove \eqref{B.8} in general from \eqref{B.10}, using \eqref{B.6}.
\end{proof}
	We define
\begin{equation}  \label{B.12}
\mathcal{C}_{\ds w}(t,s)f =
        \begin{cases}
            \begin{split}
              & \sqrt{m/(2\pi i\rho)}^{\ d}
                  \int \bigl(\exp iS(t,s; q^{t,s}_{x,y})\bigr) \\
         &\quad \times    \mathcal{F}_w(t,s;\qts)f(y)dy,\quad  s < t,
                      \end{split}
                  \\
         f, \quad s = t
        \end{cases}
\end{equation}
with $\rho = t-s$ for $f \in \Czerospace^l$.
\begin{pro} \label{pro B.3}  Suppose either Assumption 2.A or 2.B.  In addition, we suppose \eqref{2.9}, \eqref{2.11}, \eqref{2.12}, \eqref{2.17} and \eqref{B.4}.
  Let $\rho^* > 0$ be the constant determined in Proposition B.1.  Then there exists a constant $K_0 \geq 0$ such that
\begin{equation}   \label{B.13}
 \Vert \mathcal{C}_{\ds w}(t,s)f\Vert \leq e^{K_0(t-s)}\Vert f \Vert, \quad 0 \leq t - s \leq \rho^*
\end{equation}
           for all $f \in (L^2)^l$.
\end{pro}
\begin{proof} We take the function $\Phi(t,s;x,y,z) = (\Phi_1,\dots,\Phi_d) \in \bR^d$  defined by (3.11) or (3.12) in \cite{Ichinose 2023} correspondingly to Assumption 2.A or 2.B.  The constant $\rho^*>0$  is the same as   in Lemma 3.3 of 
\cite{Ichinose 2023} as noted in Remark B.1.  Let $z(t,s;x,\xi,y) = (z_1,\dots,z_d) \in \bR^d$ be the function determined  in Lemma 3.3 of \cite{Ichinose 2023}.
 We note that $\Phi(t,s;x,y,z)$ and $z \in \bR^d$ introduced above are different from $\Phi(z)$ and $z \in \bC^d$ in \S A.
We write
\begin{equation}   \label{B.14}
\Fw(t,s;x,y) = \Fw(t,s;\qts)
\end{equation}
for a while.  Let $f \in \Czerospace^l$.  Then  we can prove 
\begin{align}  \label{B.15}
    & \Csw^{\dag}\chi(\epsilon\cdot)^2\Csw f = \left(\frac{m}{2\pi\rho}\right)^d\iint_{\bR^{2d}}\chi(\epsilon z)^2\left(\exp i(x - y)\cdot \frac{m\Phi}{t-s} \right)
    \notag \\
    &\qquad \times   \Fw(t,s;z,x)^{\dag}\Fw(t,s;z,y)f(y)dydz \notag \\
    & = \left(\frac{1}{2\pi}\right)^d\int e^{i(x-y)\cdot\eta}d\eta \int \chi(\epsilon z)^2 \Fw(t,s;z,x)^{\dag}\Fw(t,s;z,y) \notag \\
    & \times \det\frac{\partial z}{\partial \xi}(t,s;x,(t-s)\eta/m,y)f(y)dy, \quad 0 \leq t - s \leq \rho^*
  \end{align}
  with $z = z(t,s;x,(t-s)\eta/m,y)$ from \eqref{B.12} as in the proof of (3.25) in \cite{Ichinose 2023}.
  \par
  In Lemma 3.3 of \cite{Ichinose 2023} we have proved
\begin{equation}   \label{B.16}
           \sum_{j=1}^d |\partial_{\xi}^{\alpha}\partial_x^{\beta}\partial_{y}^{\gamma}
            z_j(t,s;x,\xi,y)|  \leq C_{\alpha\beta\gamma},
\  |\alpha + \beta  + \gamma| \geq 1
\end{equation}
in $0 \leq t - s \leq \rho^*$ and $(x,\xi,y) \in \bR^{3d}$, and
\begin{align}  \label{B.17}  
           & \det \frac{\partial z}{\partial\xi}(t,s;x,\xi,y) = 1 + (t - s)h(t,s;x,\xi,y) > 0,\notag \\
            &  |\partial_{\xi}^{\alpha}\partial_x^{\beta}\partial_{y}^{\gamma}
            h(t,s;x,\xi,y)|  \leq C_{\alpha\beta\gamma} < \infty
\end{align}
for all $\alpha, \beta$ and $\gamma$ in $0 \leq t - s \leq \rho^*$ and $(x,\xi,y) \in \bR^{3d}$.
Hence from \eqref{B.6} and \eqref{B.15}-\eqref{B.17} we can prove for $f \in \Sspace^l$
\begin{align}  \label{B.18}
   & \lim_{\epsilon\to 0+0} \Csw^{\dag}\chi(\epsilon\cdot)^2\Csw f  = \int e^{i(x-y)\cdot\eta}\,\dbar \eta 
   \int \Fw(t,s;z,x)^{\dag}
   \notag \\
    &\quad  \times \Fw(t,s;z,y)f(y)dy + (t-s)\int e^{i(x-y)\cdot\eta}\,\dbar \eta \int \Fw(t,s;z,x)^{\dag}
     \notag \\
     &\quad  \times \Fw(t,s;z,y)h(t,s;x,(t-s)\eta/m,y)f(y)dy, \quad 0 \leq t - s \leq \rho^*    
\end{align}
with $z = z(t,s;x,(t-s)\eta/m,y)$ in the topology of $\Sspace^l$, which we write as $\Csw^{\dag}\Csw f $ formally. Noting
\begin{equation*}   
\|\Fw(t,s;z,x)^{\dag}\Fw(t,s;z,y)\|_{\bC^l} \leq \|\Fw(t,s;z,x)\|_{\bC^l} \cdot \|\Fw(t,s;z,y)\|_{\bC^l} \leq 1
\end{equation*}
from \eqref{B.6}, we apply Theorem A.1 to the right-hand side of \eqref{B.18} as $\h = (t-s)/m$.
Then we can prove that the $L^2$-norm of \eqref{B.18} is bounded above by
\begin{equation}   \label{B.19}
\|f\|  + 2K_0(t-s)\|f\|  \leq e^{2K_0(t-s)}\|f\|
\end{equation}
with a constant $K_0 \geq 0$, which shows 
\begin{align*}   
 & \|\Csw f\|^{2}  \leq \lim_{\epsilon \to 0+0}\bigl(\chi(\epsilon\cdot)\Csw f,\chi(\epsilon\cdot)\Csw f\bigr) \\
 & \quad = \lim_{\epsilon \to 0+0}\bigl(\Csw^{\dagger} \chi(\epsilon\cdot)^2\Csw f, f\bigr) = \bigl(\Csw^{\dagger}\Csw f, f\bigr)
  \\
    &\quad  \leq \|\Csw^{\dagger}\Csw f\|\cdot \|f\|
   \leq e^{2K_0(t-s)}\|f\|^2,  \quad 0 \leq t - s \leq \rho^*
\end{align*}
for $f \in \Sspace^{l}$.  Hence we can complete the proof of  \eqref{B.13}.
\end{proof}
	Changing the variables: $\bR^d \ni y \to v = (x-y)/\sqrt{\rho} \in \bR^d$, we can write $P(t,s)f$ defined by \eqref{B.3} as 
\begin{equation}  \label{B.20}
    P(t,s)f  = \sqrt{\frac{m}{2\pi i}}^{\, d}
                  \int e^{i\phi(t,s;x,v)}p(x,v)f(x-\sqrt{\rho}v)dv,\  \rho = t - s > 0
\end{equation}
for $f \in \Czerospace$, where
\begin{align}  \label{B.21}
   & \phi(t,s;x,v)  = \frac{m}{2}|v|^2 + \sqrt{\rho}v\cdot\int_0^1A(t-\theta\rho,x-\theta\sqrt{\rho}v)d\theta
   \notag \\
   &\quad  -\rho\int_0^1V(t-\theta\rho,x-\theta\sqrt{\rho}v)d\theta  \equiv \frac{m}{2}|v|^2 + \psi(t,s;x,\sqrt{\rho}v)
\end{align}
(cf. (3.8) in \cite{Ichinose 2003}).
\begin{lem} \label{lem B.4}
Suppose Assumptions 2.C, 2.D and \eqref{2.17}.   Then, for an arbitrary multi-index $\kappa$ both of commutators $[\partial_x^{\kappa},\Csw]f$ and $[x^{\kappa},\Csw]f$ for $f \in \Czerospace^{l}$ are written in the form
\begin{align}  \label{B.22}
   & (t-s)\sum_{|\gamma| \leq |\kappa|}\widetilde{P}_{\gamma}(t,s)(\partial_x^{\gamma}f) := (t-s) \sum_{|\gamma| \leq |\kappa|}   \sqrt{\frac{m}{2\pi i}}^{\, d}\notag \\
   & \times \int e^{i\phi(t,s;x,v)}p_{\gamma}(t,s;x,\sqrt{\rho}v)(\partial_x^{\gamma}f)(x-\sqrt{\rho}v)dv,             
   \end{align}
   where $p_{\gamma}(t,s;x,\zeta)$ satisfy
\begin{equation}  \label{B.23}
   \|\partial_{\zeta}^{\alpha}\partial_x^{\beta}p_{\gamma}(t,s;x,\zeta)\|_{\bC^l} \leq C_{\alpha\beta}<x;\zeta>^{|\kappa| - |\gamma|}
\end{equation}
for all $\alpha$ and $\beta$.
\end{lem}
\begin{proof}
We can write
\begin{equation*}  
    \Csw f = \sqrt{\frac{m}{2\pi i}}^{\, d}
                  \int e^{i\phi(t,s;x,v)}\Fw(t,s;x,x-\sqrt{\rho}v)f(x-\sqrt{\rho}v)dv
\end{equation*}
as in \eqref{B.20}, where we used \eqref{B.14}.  Let $\lambda$ be an arbitrary multi-index and set $\mathfrak{l}= |\lambda|$. Then we have
\begin{align} \label{B.24}
   & \partial_{x_j}(\Csw \partial_x^{\lambda}f)
- \Csw(\partial_{x_j}\partial_x^{\lambda}f) =  \sqrt{\frac{m}{2\pi i}}^{\,d} \int (i\partial_{x_{j}}\phi) e^{i\phi}\Fw(t,s;x,x-\sqrt{\rho}v) \notag \\
& \quad \times (\partial_x^{\lambda}f)(x -\sqrt{\rho}v)dv  
+ \sqrt{\frac{m}{2\pi i}}^{\,d} \int e^{i\phi}\left\{\frac{\partial}{\partial x_j}\Fw(t,s;x,x-\sqrt{\rho}v)\right\} 
\notag \\
& \qquad \times (\partial_x^{\lambda}f)(x -\sqrt{\rho}v)dv.
\end{align}
Noting \eqref{B.6}, we can see as in the proof of Lemma 3.2 of \cite{Ichinose 2003} that the first term on the right-hand side of \eqref{B.24} is expressed in the form of \eqref{B.22} with $|\kappa| = \dl +1$.  It is also proved from \eqref{B.8} with $M' = 1$ that the second term is expressed in the form of \eqref{B.22} with $|\kappa| = \dl +1$.  Hence we can prove the statement of Lemma B.4 for $[\partial_x^{\kappa},\Csw]f$ by induction as in the proof of Lemma 3.2 of \cite{Ichinose 2003}.  In the same way we have
\begin{align} \label{B.25}
   & x_j\Csw(x^{\lambda}f)- \Csw(x_j x^{\lambda}f) = \sqrt{\rho} \sqrt{\frac{m}{2\pi i}}^{\,d}
   \int e^{i\phi}
      \notag \\
       &\qquad \times \Fw(t,s;x,x-\sqrt{\rho}v)v_j(x -\sqrt{\rho}v)^{\lambda}
     f(x - \sqrt{\rho}v)dv  \notag \\
        & = \frac{i\sqrt{\rho}}{m} \sqrt{\frac{m}{2\pi i}}^{\,d}
\int e^{im|v|^2/2}\partial_{v_j}\Bigl\{e^{i\psi(t,s;x,\sqrt{\rho}v)}
\Fw(t,s;x,x-\sqrt{\rho}v) \notag \\
& \qquad \times(x -\sqrt{\rho}v)^{\lambda}f(x - \sqrt{\rho}v)\Bigr\}dv.
\end{align}
Noting \eqref{B.6}, we can prove as in the proof of Lemma 3.2 of \cite{Ichinose 2003} that
\eqref{B.25} is expressed in the form of \eqref{B.22} with $|\kappa| = \dl +1$.   Hence we can prove the statement of Lemma B.4 for $[x^{\kappa},\Csw]f$ by induction as in the proof of Lemma 3.2 of \cite{Ichinose 2003}.
\end{proof}
\begin{pro} \label{B.5}
Suppose the same assumptions as in Theorem 2.3. Let $\rho^* > 0$ be the constant  in Proposition B.1.  
Then, for $a = 0, 1, 2, \dots$ there exist constants $K_a \geq 0$ such that 
\begin{equation} \label{B.26}
\Vert \Csw f\Vert_a \leq e^{K_a(t - s)}\Vert f\Vert_{a}, \quad 0 \leq t - s \leq \rho^*
\end{equation}
for all $f \in (B^{a}(\bR^d))^l$.
\end{pro}
\begin{proof} 
Using Proposition B.1 with $M_1 = 1$ and Proposition B.3, we can prove Proposition B.5 from Lemma B.4 as in the proof of 
Proposition 3.4 of \cite{Ichinose 2003}.
\end{proof}
\begin{lem} \label{lem B.6}
We suppose \eqref{2.11}, \eqref{2.12} and \eqref{2.17}.  Moreover, we assume that there exists a constant $M \geq 0$ satisfying 
\begin{align} \label{B.27}
 &|\partial_x^{\alpha}V(t,x)| + \sum_{j=1}^d\big\{|\partial_x^{\alpha}A_j(t,x)| + 
 |\partial_x^{\alpha}\partial_tA_j(t,x)|\bigr\} \notag \\
 &\qquad + \|\partial_x^{\alpha}W_{\ds}(t,x)\|_{\bC^l}
 \leq C_{\alpha}<x>^{M}
\end{align}
in $\domain$ for  all $\alpha$.  Then we have
\begin{align}  \label{B.28}
      & \Bigl[i\partial_t  - H(t) - H_{\ds}(t) + iW_{\ds}(t) \Bigr]\Csw f \notag \\
      & = \sqrt{t - s}R_{\ds w}(t,s)f  \equiv \sqrt{t-s}\Bigl(R_{ij}(t,s);i,j=  1,2,\dots,l\Bigr)f,
\end{align}
 where $R_{ij}(t,s)$ are the
operators defined by \eqref{B.3} with $r_{ij}(x,v)$ satisfying \eqref{B.1} for an integer  $M' \geq 0$.
\end{lem}
\begin{proof}
We can prove Lemma B.6 as in the proof of (3.16) in \cite{Ichinose 2007}, where we replace $H_1(t)$ with $H_{\ds} - iW_{\ds}(t)$.
In fact, applying \eqref{B.6} and \eqref{B.8} to (3.21), (3.22) and (3.24) in \cite{Ichinose 2007}, we can complete the proof.
\end{proof}
 \begin{pro} \label{pro B.7}
 Suppose the same assumptions as in Theorem 2.3.   Let $U_w(t,s)f$ be the solution to \eqref{2.3} with $u(s) = f$ found in Theorem 2.1.  Then there exists an integer $M'' \geq 0$ such that  we have
 \begin{equation} \label{B.29}
 \Vert \Csw f - U_w(t,s)f\Vert_a \leq  C_a(t-s)^{3/2}\Vert f\Vert_{a+M''},\quad  0\leq  t-s \leq \rho^*
 \end{equation}
 for $f \in (B^{a+M''})^l\ (a = 0,1,2,\dots)$, 
 where $\rho^* > 0$ is the constant in Proposition B.1.
 \end{pro}
 \begin{proof}
 From \eqref{1.1} we have $\partial_tA_j = - E_j - \partial_{x_j}V\ (j = 1,2,\dots,d)$. Using these expressions, we can easily see that \eqref{B.27} hold for all $\alpha$ under our assumptions (cf. the proof of Proposition 4.2 of \cite{Ichinose 2023}).  Hence,
 using \eqref{2.15}, \eqref{B.5}  with $M_1 = 1$, \eqref{B.26}  and Lemma B.6, we can prove \eqref{B.29} as in the proof of Proposition 4.6 of
 \cite{Ichinose 2023}.
 \end{proof}
 \vspace{0.5cm}
 {\bf Proof of Theorems 2.2 and 2.3.}
 \par
 Let $\kdelta (t,0)f$ be the operator defined by \eqref{2.5}. We can write
\begin{align}  \label{B.30}
    \kdelta (t,0)f  = & \limepsilon{\cal C}_{\ds w}(t,\tau_{\nu-1})\chi(\epsilon\cdot){\cal C}_{\ds w}(\tau_{\nu-1},\tau_{\nu-2})\notag \\
    &\quad \cdot \chi(\epsilon\cdot)\cdots \chi(\epsilon\cdot)
{\cal C}_{\ds w}(\tau_1,0)f
\end{align}
for $f \in C^{\infty}_0(\bR^d)^l$ from \eqref{B.12} as in the proof of (6.3) in \cite{Ichinose 2023}.  
Hence, using \eqref{2.15}, \eqref{B.26} and \eqref{B.29}, we can complete the proof of
Theorems 2.2 and 2.3 as in the proof of Theorems 2.1 and 2.2 of \cite{Ichinose 2023}, respectively.
\par
We will prove Theorem 2.5 below.  Theorem 2.4 can be proved from Theorem 2.5 as Theorem 2.2 can be done from Theorem 2.3.
	For the sake of simplicity we assume $N = 1$.   We can prove Theorem 2.5 generally in the same way. In addition, we will prove 
	the  case of $t_1 \in [0,t)$ below.  In the same way we can prove the case of $t_1 = t$.
\par
	Let $\mathcal{F}_w(\theta,s;q)$ be the solution to \eqref{2.4} with $I$ at $\theta = s$. 
	 Let $0 \leq s \leq t' \leq \theta \leq T.$  Then, both of $\mathcal{F}_w(\theta,s;q)$ and $\mathcal{F}_w(\theta,t';q)\mathcal{F}_w(t',s;q)$
	 satisfy \eqref{2.4} with $\mathcal{F}_w(t',s;q)$ at $\theta = t'$. Hence we have
	 \begin{equation} \label{B.31}
	 \mathcal{F}_w(\theta,s;q) = \mathcal{F}_w(\theta,t';q)\mathcal{F}_w(t',s;q).
	 \end{equation}
\par
	
	We take $j$ such that $\tau_j \leq t_1 < \tau_{j+1}$.  Then, using \eqref{B.31}, we can write 
\begin{align}  \label{B.32}
     & \Fw(t,t_1;\qdelta)Z_1(\qdelta(t_1))\Fw(t_1,0;\qdelta) = \Fw(t,\tau_{j+1};\qdelta)\Fw(\tau_{j+1},t_1;\qdelta)
                   \notag \\
     &\quad \cdot Z_1(\qdelta(t_1))\Fw(t_1,\tau_j;\qdelta)\Fw(\tau_{j},0;\qdelta) 
     = \Fw(t,\tau_{\nu-1};q^{t,\tau_{\nu-1}}_{x,x^({\nu-1})})\cdots 
     \notag \\
     & \quad \cdot \Fw(\tau_{j+2},\tau_{j+1};q^{\tau_{j+2},\tau_{j+1}}_{x^{(j+2)},x^{(j+1)}})
     \Fw(\tau_{j+1},t_1;q^{\tau_{j+1},\tau_{j}}_{x^{(j+1)},x^{(j)}})Z_1(q^{\tau_{j+1},\tau_{j}}_{x^{(j+1)},x^{(j)}}(t_1))
     \notag \\
     & \quad \cdot \Fw(t_1,\tau_{j};q^{\tau_{j+1},\tau_{j}}_{x^{(j+1)},x^{(j)}})
     \Fw(\tau_{j},\tau_{j-1};q^{\tau_{j},\tau_{j-1}}_{x^{(j)},x^{(j-1)}})\cdots 
     \Fw(\tau_{1},0;q^{\tau_{1},0}_{x^{(1)},x^{(0)}}).
     \end{align}
Consequently, from \eqref{2.20} and \eqref{B.12} we can write
\begin{align}  \label{B.33}
     & \left<Z_1(\qdelta(t_1))\right>_w f = \lim_{\epsilon \to 0+0}\mathcal{C}_{\ds w}(t,\tau_{\nu-1})\chi(\epsilon \cdot)
     \cdots \chi(\epsilon \cdot)\mathcal{C}_{\ds w}(\tau_{j+2},\tau_{j+1})\chi(\epsilon \cdot)
                   \notag \\
     &\quad \cdot \widetilde{\mathcal{C}}_{\ds w}(\tau_{j+1},t_1,\tau_{j})\chi(\epsilon \cdot)
     \mathcal{C}_{\ds w}(\tau_{j},\tau_{j-1})\chi(\epsilon \cdot)\cdots \chi(\epsilon \cdot)
     \mathcal{C}_{\ds w}(\tau_{1},0)f,
              \end{align}
	where 
\begin{align}  \label{B.34}
      \widetilde{\mathcal{C}}_{\ds w}(t,t_1,s)f & = \sqrt{\frac{m}{2\pi i\rho}}^{\ d}\int e^{iS(t,s; q^{t,s}_{x,y})}
     \mathcal{F}_w(t,t_1;\qts)Z_1\bigl(\qts(t_1)\bigr)
                   \notag \\
     &\qquad \times \mathcal{F}_w(t_1,s;\qts)f(y)dy, \quad \rho = t - s > 0.
                         \end{align}
Applying \eqref{B.6} to \eqref{B.34}, we have
\begin{equation}  \label{B.35}
      \Vert \widetilde{\mathcal{C}}_{\ds w}(t,t_1,s)f\Vert_a \leq C_a\Vert f \Vert_{a+\dM_1}
  \end{equation}
from \eqref{B.5} with $M_1 = 1$. Hence, noting \eqref{B.26}, the expression \eqref{B.33} shows 
\begin{align}  \label{B.36}
     & \left<Z_1(\qdelta(t_1))\right>_w f = \mathcal{C}_{\ds w}(t,\tau_{\nu-1})\mathcal{C}_{\ds w}(\tau_{\nu-1},\tau_{\nu-2})
     \cdots \mathcal{C}_{\ds w}(\tau_{j+2},\tau_{j+1})
                   \notag \\
     &\quad \cdot \widetilde{\mathcal{C}}_{\ds w}(\tau_{j+1},t_1,\tau_{j})
     \mathcal{C}_{\ds w}(\tau_{j},\tau_{j-1})\cdots 
     \mathcal{C}_{\ds w}(\tau_{1},0)f
 \end{align}
in $(B^a)^l$ for $f \in (B^{a+\mathfrak{M}_1})^l\ (a = 0,1,2,\dots)$ and 
\begin{equation}  \label{B.37}
      \Vert \left<Z_1(\qdelta(t_1))\right>_w f\Vert_a \leq C_ae^{K_a'T}\Vert f \Vert_{a+\dM_1}
  \end{equation}
with the same constants $C_a$ as in \eqref{B.35} and $K'_a = \max\{K_a,K_{a+\dM_1}\}$ for all $\Delta$ satisfying $|\Delta| \leq \rho^*$, which proves (1) in Theorem 2.5.
\par
  We have
\begin{align*}  
    & Z_1(\qts(t_1))- Z_1(y) = Z_1\Bigl(y + \frac{t_1-s}{t-s}(x-y)\Bigr) - Z_1(y)
   \notag \\
    & = \sqrt{t-s}\cdot \frac{t_1-s}{t-s}\sum_{j=1}^d\frac{x_j-y_j}{\sqrt{t-s}}
    \int_0^1\frac{\partial Z_1}{\partial x_j} \Bigl(y + \theta\frac{t_1-s}{t-s}(x-y)\Bigr)d\theta \\
    & \equiv \sqrt{t-s}\,\ddp_1(t,t_1,s;x,v) 
\end{align*}
from \eqref{B.2}, where $v = (x-y)/\sqrt{t-s}$ or $y = x - \sqrt{t-s}v$.  In addition, since 
 we  see
\begin{align*}  
    &\mathcal{F}_w(t,s;q) - I = \int_{0}^{1}\frac{d}{d\theta}\mathcal{F}_w(s+\theta(t-s),s;q) d\theta  
     \notag \\
    &\quad  = -(t-s) \int_{0}^{1}\Bigl\{iH_{\ds}\bigl(s+\theta(t-s),q(s+\theta(t-s))\bigr) \\
    & \qquad + W_{\ds}\bigl(s+\theta(t-s),q(s+\theta(t-s))\bigr)\Bigr\}\mathcal{F}_w(s+\theta(t-s),s;q) d\theta
\end{align*}
 from \eqref{2.4},  we can write
\begin{align*}  
    &\mathcal{F}_w(t_{1},s;\qts) - I = (t_{1}  -s)\ddp_{2}(t,t_{1},s;x,v),
     \notag \\
    &\mathcal{F}_w(t,t_{1};\qts) - I = (t-{t_{1}})\ddp_{3}(t,t_{1},s;x,v).
\end{align*}
We note 
\begin{align*}   
    &\qts(s+\theta(t_{1}-s)) = y+ \theta\frac{t_{1}-s}{t-s}(x-y),
     \notag \\
    &\qts(t_{1}+\theta(t-t_{1})) = x- (1 - \theta)\frac{t-t_{1}}{t-s}(x-y)
\end{align*}
from \eqref{B.2}. Applying the above expressions of $Z_1,\mathcal{F}_w$ and $\qts$ to \eqref{B.34}, we have
\begin{equation} \label{B.38}
\widetilde{\mathcal{C}}_{\ds w}(t,t_1,s)f = \mathcal{C}_{\ds w}(t,s)Z_1(\cdot)f + \sqrt{t-s}\mathfrak{P}(t,t_{1},s)f,
\end{equation}
where $\mathfrak{P}(t,t_{1},s)f = \Bigl(\mathfrak{P}_{ij}(t,t_{1},s);i,j = 1,2,\dots,l\Bigr)f$ is the operator defined by \eqref{B.3}.
Noting \eqref{B.6}, we can easily see that $\ddp_{ij}(t,t_{1},s;x,v)\ (i,j = 1,2,\dots,l) $ satisfy \eqref{B.1}  with $M = \dM_{1}+2$.
Hence from \eqref{B.36} we obtain 
\begin{align}  \label{B.39}
     & \left<Z_1(\qdelta(t_1))\right>_w f = \mathcal{C}_{\ds w}(t,\tau_{\nu-1})\mathcal{C}_{\ds w}(\tau_{\nu-1},\tau_{\nu-2})
     \cdots \mathcal{C}_{\ds w}(\tau_{j+2},\tau_{j+1})
                   \notag \\
     &\quad \cdot \Big\{\mathcal{C}_{\ds w}(\tau_{j+1},\tau_j)Z(\cdot) + \sqrt{\tau_{j+1}-\tau_{j}}\,\mathfrak{P}(\tau_{j+1},t_{1},\tau_{j})\Bigr\}
     \mathcal{C}_{\ds w}(\tau_{j},\tau_{j-1})\cdots 
     \mathcal{C}_{\ds w}(\tau_{1},0)f.
 \end{align}
	Let $M''$ be the integer determined in Proposition B.7 and set $\dM'_1 = \max \{\dM_1+ M'',\dM_1+2\}$.  Then, 
	using \eqref{2.15}, \eqref{B.5}, \eqref{B.26} and \eqref{B.29},  from \eqref{B.39} we can prove
\begin{equation}  \label{B.40}
     \lim_{|\Delta|\to 0} \Bigl\Vert \left<Z_1(\qdelta(t_1))\right>_w f - U_{w}(t,\tau_{j})Z_1(\cdot)U_{w}(\tau_{j},0)f\Bigr\Vert_{a} = 0
      \end{equation}
	uniformly in $t \in [t_{1},T]$ for $f \in (B^{a+\dM'_{1}})^{l}$ as in the proof of Theorem 2.3. 
	We can easily see
\begin{align*}  
     & U_{w}(t,\tau_j)Z_1(\cdot)U_{w}(\tau_{j},0)f -  U_{w}(t,t_{1})Z_1(\cdot)U_{w}(t_{1},0)f  \\
     & = U_{w}(t,t_1)\bigl\{ U_{w}(t_1,\tau_{j})-I\bigr\}Z_1(\cdot)U_{w}(\tau_{j},0)f  \\
     & \quad + U_{w}(t,t_1)Z_1(\cdot)\bigl\{I - U_{w}(t_{1},\tau_{j})\bigr\}U_{w}(\tau_{j},0)f 
     \end{align*}
and 
\begin{equation*}  
     \Vert  U_{w}(t,s)f - f\Vert_{a} \leq C_{a}(t-s)\Vert f \Vert_{a+2}
      \end{equation*}
for $a = 0,1,2,\dots$ from \eqref{2.3}.  Consequently we have
\begin{equation}  \label{B.41}
      \bigl\Vert U_{w}(t,\tau_{j})Z_1(\cdot)U_{w}(\tau_{j},0)f -  U_{w}(t,t_{1})Z_1(\cdot)U_{w}(t_{1},0)f\bigr\Vert_{a} 
      \leq C'_{a} |\Delta|\Vert f \Vert_{a+\dM_{1}+2} \\
     \end{equation}
for $a = 0,1,2,\dots$ from \eqref{2.15}, which shows 
\begin{equation}  \label{B.42}
     \lim_{|\Delta|\to 0} \Bigl\Vert \left<Z_1(\qdelta(t_1))\right>_w f - U_{w}(t,t_{1})Z_1(\cdot)U_{w}(t_{1},0)f\Bigr\Vert_{a} = 0
      \end{equation}
uniformly in $t \in [t_{1},T]$ for $f \in (B^{a+\dM'_{1}})^{l}$  together with \eqref{B.40}.  The inequalities \eqref{2.15} and \eqref{B.37} show 
\begin{equation}  \label{B.43}
      \Bigl\Vert \left<Z_1(\qdelta(t_1))\right>_w f - U_{w}(t,t_{1})Z_1(\cdot)U_{w}(t_{1},0)f\Bigr\Vert_{a} \leq C_{a}\Vert f \Vert_{a+\dM_{1}}
      \end{equation}
	for $f \in (B^{a+\dM_{1}})^{l}, t \in [t_{1},T]$ and all $\Delta$ satisfying $|\Delta| \leq \rho^{*}$.  Then we have
\begin{align*}  
     &  \Bigl\Vert \left<Z_1(\qdelta(t_1))\right>_w f - U_{w}(t,t_{1})Z_1(\cdot)U_{w}(t_{1},0)f\Bigr\Vert_{a} 
     \leq \Bigl\Vert \left<Z_1(\qdelta(t_1))\right>_w g\\
     &\qquad   - U_{w}(t,t_{1})Z_1(\cdot)U_{w}(t_{1},0)g \Bigr\Vert_{a} 
     + C_{a}\Vert f - g \Vert_{a+\dM_{1}}
      \end{align*}
	for $f \in (B^{a+\dM_{1}})^{l}$ and $g \in (B^{a+\dM'_{1}})^{l},  t \in [t_{1},T]$ and all $\Delta$ satisfying $|\Delta| \leq \rho^{*}$.
	Using this expression and \eqref{B.42}, we can complete the proof of (2) and (3) in Theorem 2.5 as in the proof of Theorem 2.3.
\vspace{0.2cm}
\par
{\bf Data availability statement.} Data sharing is not applicable to this
article as no new data were created or analyzed in this study.
%
 
%
%
%


\begin{thebibliography}{99} %
%
%
\bibitem{Albeverio et all} Albeverio, S. A., H{\o}egh-Krohn, R. J., Mazzucchi, S.: 
Mathematical Theory of
Feynman Path Integrals, An Introduction, 2nd Corrected and Enlarged Edition. Lecture Notes in Math. {\bf 523},  Berlin, 
Heidelberg: Springer-Verlag, 2008
%
%
\bibitem{Aharonov-Bohm} Aharonov, Y., Bohm, D.: Significance of electromagnetic potentials in the quantum theory. Phys. Rev. 
{\bf 115}, 485-491 (1959)
%
%
\bibitem{Calarco} 
Calarco, T.:
Impulsive quantum measurements: restricted path integral versus von Neumann collapse. Nuovo Cimento B
{\bf 110}, 1451-1461 (1995)
%
\bibitem{Caves} 
Caves, C. M.:
Quantum-mechanical model for continuous position measurements. Phys. Rev. A
{\bf 36}, 5543-5555 (1987)
%
%
\bibitem{Dirac}
Dirac, P. A. M.:
	  The Principles of Quantum Mechanics, Fourth Edition. London: Oxford University Press, 1958
%
%
\bibitem{Exner-Ichinose} Exner, P., Ichinose, T.:
Note on a product formula related to quantum Zeno dynamics. Ann. Henri Poincar\'e 
{\bf 22}, 1669-1697 (2021)
%
%
%
\bibitem{Feynman 1948} Feynman, R. P.:
Space-time approach to non-relativistic quantum mechanics. Rev. Mod. Phys. 
{\bf 20}, 367-387
  (1948)
%
\bibitem{F-L-S II}Feynman, R. P., Leighton, R. B., Sands, M.:
The Feynman Lectures on Physics, Vol. II. Reading, MA: Addison-Wesely, 1964 
%
\bibitem{F-L-S III} Feynman, R. P., Leighton, R. B., Sands, M.:
The Feynman Lectures on Physics, Vol. III. Reading, MA: Addison-Wesely, 1965  
%
%
\bibitem{Feynman-Hibbs} Feynman, R. P., Hibbs, A. R.:
Quantum Mechanics and Path Integrals. New York: McGraw-Hill, 1965
%
\bibitem{Friedman} Friedman, C. N.:
Semigroup product formulas, compressions, and continual observations in quantum mechanics. Indiana Univ. Math. J.
{\bf 21}, 1001-1011
  (1972)
%
\bibitem{Ichinose 1995}
  Ichinose,  W.: A note on the existence and $\hbar$-dependency  of the
solution of equations in quantum mechanics. Osaka J. Math.  {\bf 32}, 327-345 (1995)
%
%
%
\bibitem{Ichinose 1999} 
Ichinose, W.: On convergence of the Feynman path integral formulated
through broken line paths. Rev. Math. Phys.  {\bf 11}, 1001-1025 (1999)
%
%
%
\bibitem{Ichinose 2003}
Ichinose, W.:  Convergence of the Feynman path integral in the weighted
Sobolev spaces and the representation of correlation functions. J. Math. Soc. Japan
{\bf 55}, 957-983 (2003)
%
%
\bibitem{Ichinose 2006} 
Ichinose, W.: A mathematical theory of the phase space Feynman path integral of the functional. Commun. Math. Phys. {\bf 265}, 739-779 (2006)
%
\bibitem{Ichinose 2007}
Ichinose, W.: A mathematical theory of the  Feynman path integral for the generalized Pauli equations. J. Math. Soc. Japan
{\bf 59}, 649-668 (2007)
%
%
\bibitem{Ichinose 2014} 
Ichinose, W.: On  the Feynman path integral for the Dirac equation in the general dimensional spacetime. Commun. Math. Phys. {\bf 329}, 483-508  (2014)
%
%
%
%
\bibitem{Ichinose 2020} 
Ichinose,  W., Aoki, T.: Notes on the Cauchy problem for the self-adjoint and non-self-adjoint  Schr\"odinger equations with polynomially growing potentials. J. Pseudo-Differ. Oper. Appl.
{\bf 11}, 703-731 (2020)
%
%
%
\bibitem{Ichinose 2023} Ichinose, W.: On the mathematical formulation of the restricted Feynman path integrals through broken line paths. Osaka. J. Math. {\bf 60}, 105-132 (2023)
%
%
\bibitem{Ichinose 2024} Ichinose, W.: On an extension of Zworski's theorem  to the matrix case on the $L^2$ boundedness for pseudo-differential operators. (in preparation)
%
%
\bibitem{Jacobs} 
Jacobs, K.: Quantum Measurement Theory and Its Applications. Cambridge: Cambridge University Press, 2014
%
%
%
\bibitem{Kato} 
Kato, T.: Perturbation Theory for Linear Operators. 
New York: Springer-Verlag, 1966
%
%
%
\bibitem{Kolmogorov-Fomin} 
Kolmogorov, A.N., Fomin, S.V.: Elements of the Theory of Functions and Functional Analysis, Fourth Edition. Moscow: Nauka, 1976 (in Russian)
%
%
%
%
\bibitem{Kumano-go} 
Kumano-go, H.: Pseudo-Differential Operators. 
Cambridge: MIT Press,
1981
%
%
%
\bibitem{Matsushima} 
Matsushima, Y.: Differentiable Manifolds. 
New York: Marcel Dekker Inc.,
1972
%
%
\bibitem{Mazzucchi} Mazzucchi, S.: Mathematical Feynman Path Integrals and Their Applications. Singapore: World Scientific Publishing Co., 2009
%
\bibitem{Mensky 1993}
Mensky, M. B.:  Continuous Quantum Measurements and Path Integrals. Bristol: IOP Publishing, 1993
%
%
\bibitem{Mensky 1998}
Mensky, M. B.:  Decoherence and the theory of continuous quantum measurements. Physics-Uspekhi {\bf 41}, 923-940 (1998)
%
%
\bibitem{Mensky 2000} 
Mensky, M. B.: Quantum Measurements and Decoherence. Dordrecht: Kluwer Academic Publishers, 2000
%
%
%
\bibitem{Mizohata}
 Mizohata, S.:  The Theory of Partial Differential Equations. New York: Cambridge University Press, 1973
%
%
%
\bibitem{Neumann} 
Von Neumann, J.: Mathematische Grundlagen der Quantenmechanik. Berlin: Springer Verlag, 1932; 
English translation by Beyer, R. T.: Mathematical Foundations of Quantum Mechanics. Princeton: Princeton University Press, 1955
%
%
%
%
\bibitem{Peres}
Peres, A.:  Quantum Theory: Concepts and Methods. Dordrecht: Kluwer  Academic Publishers, 1993
%
%
\bibitem{Sakurai}
Sakurai, J. J, Napolitano, J.:  Modern Quantum Mechanics, Second Edition.  Cambridge: Cambridge University Press, 2017
%
%
\bibitem{Sverdlov}
Sverdlov, R.:  How Mensky's continuous measurement can emerge from GRW on larger time scales. 
Foundation of Physics {\bf 46}, 825-835 (2016)
%
%
%
%
%
%
\bibitem{Z} Zworski, M.: Semiclassical Analysis. Providence, RI: American Mathematical Society, 2012
%
%
%
\end{thebibliography}
\end{document}